\documentclass[12pt, a4paper]{article}
\usepackage[fleqn]{amsmath}      
\numberwithin{equation}{section}
\usepackage{amssymb}
\usepackage{amsthm}
\usepackage{appendix}
\usepackage{array}
\usepackage{arydshln}
\usepackage{authblk}
\usepackage{blkarray}
\usepackage[english]{babel}
\usepackage{bbm}
\usepackage{bm}
\usepackage{booktabs}
\usepackage[font=small,labelfont=bf]{caption}
\usepackage{color}
\usepackage{econometrics}
\usepackage{enumerate}
\usepackage{float}
\usepackage[margin=2.25cm]{geometry}
\usepackage{graphicx}
\usepackage[latin1]{inputenc}
\usepackage{mathrsfs}
\usepackage{mathtools}
\usepackage{microtype}
\usepackage{multicol}
\usepackage{multirow}
\usepackage[authoryear,round]{natbib}
\newcommand\cites[1]{\citeauthor{#1}'s\ (\citeyear{#1})}

\usepackage{pdflscape}
\usepackage{soul}
\usepackage{subcaption}
\usepackage{subfiles}

\usepackage{times} 
\usepackage{tablefootnote}
\usepackage[flushleft]{threeparttable}
\usepackage{txfonts}
\usepackage{url}
\usepackage{verbatim}
\usepackage[all]{xy}
\usepackage{xr}
\newcommand{\FChol}{\mF}

\makeatletter
\let\@addpunct\@gobble
\makeatother

\newtheoremstyle{mystyle}
  {0.5cm}
  {0.5cm}
  {\itshape}
  {}
  {\bfseries}
  { }
  {\newline}
  {\thmname{#1}\thmnumber{ #2}\thmnote{ (#3)}}
\theoremstyle{mystyle}
\newtheorem{theorem}{Theorem}
\newtheorem{remark}{Remark}

\newtheorem{example}{Example}
\newtheorem{assumption}{Assumption}
\newtheorem{lemma}{Lemma}

\makeatletter
\renewenvironment{proof}[1][\proofname]{\par
  \pushQED{\qed}%
  \normalfont \topsep6\p@\@plus6\p@\relax
  \trivlist
  \item[\hskip\labelsep
        \itshape
    #1]\ignorespaces
}{%
  \popQED\endtrivlist\@endpefalse
}
\makeatother

\makeatletter
\renewcommand\@biblabel[1]{}
\makeatother

\newcommand{\argmin}{\operatornamewithlimits{arg\;min}}

\DeclarePairedDelimiter\abs{\lvert}{\rvert}%
\DeclarePairedDelimiter\norm{\lVert}{\rVert}%
\makeatletter
\let\oldabs\abs
\def\abs{\@ifstar{\oldabs}{\oldabs*}}
\let\oldnorm\norm
\def\norm{\@ifstar{\oldnorm}{\oldnorm*}}
\makeatother

\makeatletter
\newcommand{\rmnum}[1]{\romannumeral #1}
\newcommand{\Rmnum}[1]{\expandafter\@slowromancap\romannumeral #1@}
\makeatother

\newcommand{\tran}{'}
\newcommand{\distrequal}{\stackrel{d}{=}}

\newcommand{\brown}{\bm{B}}
\newcommand{\brownnormal}{B}
\newcommand{\wiener}{\bm{W}}
\newcommand{\sumtT}{\sum_{t=1}^T}
\newcommand{\myint}{\int_0^1}
\newcommand{\MChol}{\bm{\mathcal{M}}}

\newcommand{\p}{\mathbb{P}}
\newcommand{\SChol}{\bm{\mathcal{S}}}
\providecommand{\lagpol}[1]{\bm{\mathcal{#1}}}

\newcommand{\romanone}{(\mathbf{\uppercase\expandafter{\romannumeral1}})}
\newcommand{\romantwo}{(\mathbf{\uppercase\expandafter{\romannumeral2}})}
\newcommand{\romanthree}{(\mathbf{\uppercase\expandafter{\romannumeral3}})}
\newcommand{\romanfour}{(\mathbf{\uppercase\expandafter{\romannumeral4}})}
\newcommand{\vvarphi}{\bm \varphi}
\newcommand{\row}{\operatorname{row}}
\newcommand{\vbias}{{\bm{\mathcal{B}}}}

\title{Efficient Estimation by Fully Modified GLS with an Application to the Environmental Kuznets Curve}
\author[1]{Yicong Lin}
\author[2]{Hanno Reuvers\thanks{Corresponding author: Department of Econometrics, Erasmus University Rotterdam, 3062 PA Rotterdam, The Netherlands. E-mail address: reuvers@ese.eur.nl.}}
\affil[1]{Department of Econometrics and Data Science, Vrije Universiteit Amsterdam, 1081 HV Amsterdam, The Netherlands}
\affil[2]{Department of Econometrics, Erasmus University Rotterdam, 3062 PA Rotterdam, The Netherlands}

\date{\today}

\begin{document}
\maketitle

\begin{abstract}

\noindent
This paper develops the asymptotic theory of a Fully Modified Generalized Least Squares estimator for multivariate cointegrating polynomial regressions. Such regressions allow for deterministic trends, stochastic trends and integer powers of stochastic trends to enter the cointegrating relations. Our fully modified estimator incorporates: (1) the direct estimation of the inverse autocovariance matrix of the multidimensional errors, and (2) second order bias corrections. The resulting estimator has the intuitive interpretation of applying a weighted least squares objective function to filtered data series. Moreover, the required second order bias corrections are convenient byproducts of our approach and lead to standard asymptotic inference. We also study several multivariate KPSS-type of tests for the null of cointegration. A comprehensive simulation study shows good performance of the FM-GLS estimator and the related tests. As a practical illustration, we reinvestigate the Environmental Kuznets Curve (EKC) hypothesis for six early industrialized countries as in \cite{wagnergrabarczykhong2019}.

\bigskip
\noindent
\textbf{JEL Classification}: C12, C13, C32, Q20

\bigskip
\noindent
\textbf{Keywords}: Cointegrating Polynomial Regression, Cointegration Testing, Environmental Kuznets Curve, Fully Modified Estimation, Generalized Least Squares

\end{abstract}

\newpage
\section{Introduction}
In recent years, there has been an increasing interest in the theoretical properties and theoretical justifications of nonlinear cointegrating relations. For theoretical properties we refer to the textbook treatise by \cite{wang2015} , the recent review article by \cite{tjostheim2020}, and the extensive references found in either of them. Theoretical justifications are in some cases refinements of existing economic theory, e.g. nonlinear cointegration among bond yields with different times to maturity due to yield-dependent risk premia as discussed in \cite{breitung2001}, or nonlinear purchasing power parity due to transaction/transportation costs and trade barriers (e.g. \cite{hongphillips2010}). In other cases, economic theory postulates a nonlinear cointegrating relation from the outset. A popular example of the latter is the Environmental Kuznets curve described in \cite{grossmankrueger1995}.\footnote{There is no direct reference to Kuznets in the original paper by \cite{grossmankrueger1995}. But their nonlinear relations between environmental indicators and per capita GDP do remind strongly of the inverted U-shaped between income inequality and economic growth proposed by Kuznets (1901-1985). The term `environmental Kuznets curve' was used later.}

There are three branches of literature on the estimation of such nonlinear cointegrating relations. First, the papers by \cite{parkphillips1999} and \cite{parkphillips2001} are concerned with nonlinear cointegration analysis of a parametric form. Second, there is a literature on nonparametric kernel estimation of nonlinear cointegrating equations, see for example \cite{wangphillips2009} or \cite{lihillipsgao2017}. The third approach is reminiscent of a nonparametric sieve estimation with power polynomial basis. That is, one estimates a cointegrating relation containing integer powers of integrated regressors. \cite{wagnerhong2016} named this a cointegrating polynomial regression (CPR). The multivariate seemingly unrelated regressions extension is available in \cite{wagnergrabarczykhong2019}. Our model specification builds on this Seemingly Unrelated Cointegrating Polynomial Regression (SUCPR) setup.

We make two theoretical contributions to the literature on cointegrating polynomial regressions. First, we propose the Fully Modified Generalized Least Squares (FM-GLS) estimator. This estimator requires two main steps: (1) It employs the inverse covariance matrix of the $2nT$-dimensional innovation vector, that is, the covariance matrix of the vector which stacks the $n$ disturbances in the cointegrating equations and the $n$ disturbances driving the $I(1)$ regressors over the time span $T$. The estimation of this inverse covariance matrix is based on the Modified Cholesky Decomposition (MCD) originating from Pourahmadi (1999). The approach is computationally simple because the required quantities are obtained from the coefficients and prediction error variances of best linear least squares predictors. In our setting this translates into estimating multiple VAR models up to some maximum lag order $q$. Sufficient conditions for consistency are provided. (2) We exploit the previous results to correct the second-order biases, resulting in improved efficiency and standard chi-square inference. Also note that the approach differs from the linear cointegration results in \cite{markogakisul2005} and \cite{moonperron2005} since our bias corrections do not rely on leads and lags augmentation. Second, a multi-equation cointegration specification asks for a multivariate cointegration test. Building upon the work by \cite{choisaikkonen2010}, we propose three such tests. The first test uses pre-filtered residuals to account for serial correlation, whereas the other two are direct multivariate generalizations of the KPSS-type of test in \cite{wagnerhong2016}. The estimator and cointegration tests are subsequently studied by Monte Carlo simulation. In our simulations, the FM-GLS estimator has a higher estimation accuracy and its implied Wald test has better size control and higher size-adjusted power. We find by simulation that prefiltering improves the size control of the cointegration tests but has an adverse effect on power. In the empirical application there is a surprisingly large spread in the widths of the confidence intervals. It turns out that FM-SUR, and to a lesser degree FM-SOLS, underestimates the parameter uncertainty compared to FM-GLS.

The plan of this paper is as follows. Section \ref{sec:model} introduces the model and the modified Cholesky block decomposition. This decomposition is the main ingredient for the fully modified GLS estimator. The related asymptotic theory and stationarity tests are discussed in Section \ref{sec:asymptotictheory} whereas a finite sample simulation study is presented in Section \ref{sec:simulations}. The empirical application can be found in Section \ref{sec:empappl} where we look at the environmental Kuznets curve. Section \ref{sec:conclusion} concludes. All proofs are collected in the Appendices.\footnote{The Appendices contains the proofs of all the results that are related to the generalized least squares estimator. Supplementary material is available on the websites of the authors.}

Some words on notation. $C$ denotes a generic positive constant. The integer part of the number $a\in \SR^{+}$ is denoted by $[a]$. For a vector $\vx\in\SR^n$,  its dimension is abbreviated by $\dim(\vx)$ and its $p$-norm by $\|\vx\|_p=(\sum_{i=1}^{n}|x_i|^p)^{1/p}$. When applied to a matrix, $\|\mA\|_p$ signifies the induced norm defined by $\|\mA\|_p=\sup_{\vx\neq \vzeros} \|\mA\vx\|_p/\|\vx\|_p$. The subscripts are omitted whenever $p=2$, e.g. $\|\vx\|=\left(\sum_{i=1}^{n}|x_i|^2\right)^{1/2}$ and $\|\mA\|=\left(\lambda_{max}\left(\mA'\mA\right)\right)^{1/2}$ where $\lambda_{max}(\cdot)$ is the largest eigenvalue. Similarly, $\lambda_{min}(\cdot)$ denotes the smallest eigenvalue. The Frobenius norm is denoted as $\|\cdot\|_{\calF}$. The $(n\times n)$ identity matrix is written as $\mI_n$. The $i^{th}$ row or $i^{th}$ column of an arbitrary matrix $\mA$ are selected using $\col_i(\mA)$ and $\row_i(\mA)$, respectively. The Kronecker product is denoted ``$\otimes$''. We use the symbol ``$\wto$'' to signify weak convergence and the symbol ``$\distrequal$'' for equality in distribution. The stochastic order and strict stochastic order relations are indicated by $O_p(\cdot)$ and $o_p(\cdot)$.

\section{The Model}\label{sec:model}
As in \cite{wagnergrabarczykhong2019}, we study a system of seemingly unrelated cointegrating polynomial regressions (SUCPR), that is
\begin{equation}
\begin{aligned}
 \vy_t =\mZ_t\tran \vbeta+\vu_t,\qquad\qquad \text{for }t=1,2,\ldots, T,
 \label{eq:basemodel}
\end{aligned}
\end{equation}
where the dependent variable $\vy_t:=[y_{1t},y_{2t},\ldots,y_{nt}]\tran$ and innovations $\vu_t:=[u_{1t},u_{2t},\ldots,u_{nt}]\tran$ are $(n\times 1)$ random vectors. For the cross-sectional unit $i$, we use as explanatory variables: (1) deterministic components such as an intercept and polynomial time trends up to order $d_i$, and (2) integer powers of the $I(1)$ regressors $x_{it}$ up to degree $s_i$. Defining $\vd_{it}=[1,t,\ldots,t^{d_i}]\tran$, $\vs_{it}=[x_{it},\ldots,x_{it}^{s_{i}}]$, and $\vz_{it}=[\vd_{it}\tran,\vs_{it}\tran]\tran$, we subsequently collect all explanatory variables in the block diagonal matrix $\mZ_t=\diag[\vz_{1t},\ldots,\vz_{nt}]$. We are interested in the $d$-dimensional parameter vector $\vbeta$ where $d=\sum_{i=1}^n (d_i+s_i+1)$. Overall, each cross-sectional unit in \eqref{eq:basemodel} specifies a single cointegrating relation containing polynomials in deterministic and stochastic trends. For each $i$, the highest orders of these polynomials, i.e. $d_i$ and $s_i$, are assumed to be fixed and known. We do not allow for cointegration in the cross-sectional dimension.

The innovation series $\{\vu_t\}$ is allowed to exhibit dependencies over time and across series. We assume that these dependencies can be modeled by a stationary VAR($\infty$) process, that is
\begin{equation}
 \lagpol{A}(L) \vu_t= \left(\mI_n-  \sum_{j=1}^\infty \mA_j L^j \right)\vu_t=\veta_t,
\end{equation}
(see Assumption \ref{assumpt1:linearproc} for further details). Efficient estimation of the parameter vector $\vbeta$ now requires the use of generalized least squares (GLS). Our \cite{zellner1962}-type GLS estimator relies on the inverse of the $(nT\times nT)$ matrix $\mSigma_{\vu}=\E(\vu\vu\tran)$ where $\vu=[\vu_1\tran,\vu_2\tran,\ldots,\vu_T\tran]\tran$. In this paper, we directly estimate $\mSigma_{\vu}^{-1}$ using a multivariate extension of the modified Cholesky decomposition by \cite{pourahmadi1999}. This extension was named the \emph{Modified Cholesky Block Decomposition} (MCBD\label{abr:MCBD}) by \cite{kimzimmerman2012} and \cite{kohligarciapourahmadi2016}. The latter papers used the MCBD to parametrize the covariance matrix of multivariate longitudinal data. As in \cite{beutnerlinsmeekes2019}, we use the MCBD for the time series application mentioned above, i.e. the computation of $\mSigma_{\vu}^{-1}$. The decomposition is closely related to linear minimum MSE predictors.

We define
\begin{equation}
\begin{aligned}
 \mA(\ell) &=
 \begin{bmatrix}
  \mA_1(\ell) & \cdots & \mA_\ell(\ell)
 \end{bmatrix}
 = \argmin_{(\mTheta_1,\ldots,\mTheta_\ell)\in \SR^{n\times n\ell}} \E\left\|\vu_t-\mTheta_1 \vu_{t-1}-\cdots-\mTheta_\ell \vu_{t-\ell} \right\|^2, \\
 \mS(\ell) &= \E\left[\vu_t-\mA_{1}(\ell)\vu_{t-1}-\cdots-\mA_{\ell}(\ell)\vu_{t-\ell} \right]\left[\vu_t-\mA_{1}(\ell)\vu_{t-1}-\cdots-\mA_{\ell}(\ell)\vu_{t-\ell} \right]\tran,
\end{aligned}
\label{eq:populationMCD}
\end{equation}
and $\mS(0)=\E(\vu_t^{}\vu_t\tran)$. The inverse of the covariance matrix $\mSigma_{\vu}$ is then given by
\begin{equation}
 \mSigma_{\vu}^{-1}=\MChol_{\vu}\tran \SChol_{\vu}^{-1}\MChol_{\vu}^{},
\label{eq:modCholdecomp}
\end{equation}
where $ \SChol_{\vu} = \diag \Big(\mS(0),\mS(1),\ldots,\mS(T-1) \Big)$,
\begin{equation}
\begin{aligned}
 \MChol_{\vu} &=\left[\bm{m}_{\vu}^{ij} \right]_{1\leq i,j\leq T},\text{ with }\bm{m}_{\vu}^{ij}=
 \begin{cases}
    \mZeros_{n\times n},&\text{if}\quad i<j,\\
    \mI_n,&\text{if}\quad i=j,\\
    -\mA_{i-j}(i-1),&\text{if}\quad 2\leq i\leq T,\;1\leq j\leq i-1,
    \end{cases}
\end{aligned}
\end{equation}
and the $\mA_{j}(i)$ follow from the partitioning $\mA(\ell)=\big[\mA_{1}(\ell),\ldots,\mA_{\ell}(\ell)\big]$.

Weak stationarity of $\{\vu_t\}$ implies that the block elements of $\MChol_{\vu}$ being far below the main diagonal are small. This suggests a banding approach in which small elements are replaced by zeros. More specifically, we construct a \emph{Banded Inverse Autocovariance Matrix} (BIAM) \label{ref:BIAM} as
\begin{equation}
\mSigma_{\vu}^{-1}(q)=\MChol_{\vu}\tran(q)\SChol_{\vu}^{-1}(q)\MChol_{\vu}^{}(q),
\label{eq:bimam}
\end{equation}
where $1\leq q\ll T$ is called the banding parameter, $\SChol_{\vu}(q)=\diag\Big(\mS(0),\mS(1),\ldots,\mS(q),\ldots,\mS(q)\Big)$ and $\MChol_{\vu}(q)=\left[\bm{m}_{\vu}^{ij}(q) \right]_{1\leq i,j\leq T}$ with
\begin{equation}
\vm_{\vu}^{ij}(q)=\begin{cases}
\mZeros_{n\times n},&\text{if}\quad i<j\;\text{or}\; \{q+1<i\leq T,\;1\leq j\leq i-q-1\}\\
\mI_n,&\text{if}\quad i=j\\
-\mA_{i-j}(i-1),&\text{if}\quad 2\leq i\leq q,\;1\leq j\leq i-1\\
-\mA_{i-j}(q),&\text{if}\quad q+1\leq i\leq T,\;i-q\leq j\leq i-1.
\end{cases}
\label{eq:Mmatrixconstrucion}
\end{equation}

\begin{example} \label{example:VAR1}
 Consider a stationary $n$-dimensional VAR($3$) process specified as $\vu_t= \sum_{j=1}^3 \mA_j \vu_{t-j}+\veta_t$ with $\veta_t\stackrel{i.i.d.}{\sim}(\vzeros,\mSigma_{\eta\eta})$. For $T=4$, the MCBD $\mSigma_{\vu}^{-1}=\MChol_{\vu}\tran \SChol_{\vu}^{-1}\MChol_{\vu}^{}$ is based on 
 \begin{equation}
 \MChol_{\vu}^{}
 =
  \left[
  \begin{smallmatrix}
   \mI_n		& \mZeros		& \mZeros		& \mZeros \\
   -\mA_1(1)	& \mI_n		& \mZeros		& \mZeros \\
  -\mA_2(2)	& -\mA_1(2)	& \mI_n		& \mZeros \\
  -\mA_3		&-\mA_2		& -\mA_1		& \mI_n
  \end{smallmatrix}
  \right]
  ,
  \qquad\qquad
 \SChol_{\vu}=
 \left[
 \begin{smallmatrix}
  \mS(0) \\
  & \mS(1) \\
  & & \mS(2) \\
  & & & \mSigma_{\eta\eta}
 \end{smallmatrix}
 \right].
\end{equation}
Alternatively, with banding parameter $q=2$, the related banded inverse autocovariance matrix is $\mSigma_{\vu}^{-1}(2)=\MChol_{\vu}\tran(2)\SChol_{\vu}^{-1}(2)\MChol_{\vu}^{}(2)$ with
 \begin{equation}
 \MChol_{\vu}^{}(2)
 =
  \left[
  \begin{smallmatrix}
   \mI_n		& \mZeros		& \mZeros		& \mZeros \\
   -\mA_1(1)	& \mI_n		& \mZeros		& \mZeros \\
  -\mA_2(2)	& -\mA_1(2)	& \mI_n		& \mZeros \\
  \mZeros		&-\mA_2(2)	& -\mA_1(2)	& \mI_n
  \end{smallmatrix}
  \right]
  ,
  \qquad\qquad
 \SChol_{\vu}(2)=
 \left[
 \begin{smallmatrix}
  \mS(0) \\
  & \mS(1) \\
  & & \mS(2) \\
  & & & \mS(2)
 \end{smallmatrix}
 \right].
\end{equation}
\end{example}

The model of \eqref{eq:basemodel} can be stacked over time to yield the representation $\vy=\mZ \vbeta+\vu$ with $\vy=[\vy_1\tran,\vy_2\tran,\ldots,\vy_T\tran]\tran$, $\mZ=[\mZ_1,\mZ_2,\ldots,\mZ_T]\tran$ and $\vu$ as before. For the moment, \emph{we will assume $\mSigma_{\vu}^{-1}(q)$ to be known} and focus on the following estimator:
\begin{equation}
 \widehat{\vbeta}_{GLS}:=\left(\mZ\tran \mSigma_{\vu}^{-1}(q) \mZ \right)^{-1} \mZ\tran \mSigma_{\vu}^{-1}(q) \vy.
\end{equation}
A discussion on the properties of this infeasible estimator is informative because: (1) the incurred estimation error of an appropriately constructed estimator$\widehat{\mSigma_{\vu}^{-1}}(q)$ will be asymptotically negligible, and (2) we can suppress the effect of banding by letting $q$ increase with sample size.

Two remarks related to $\widehat{\vbeta}_{GLS}$ are instructive. First, the GLS estimator differs from the usual least squares estimator $\widehat{\vbeta}_{OLS}:=\left(\mZ\tran \mZ \right)^{-1} \mZ\tran \vy$ by a weighing with the inverse covariance matrix $\mSigma_{\vu}^{-1}(q)$. It is well documented in standard econometric textbooks (e.g. chapter 7 of \cite{davidsonmackinnon2004}) that this weighing may lead to substantial efficiency gains. Second, it is illustrative to substitute the Modified Cholesky Decomposition of $\mSigma_{\vu}^{-1}(q)$ into the definition of this infeasible GLS estimator. The result is $\widehat{\vbeta}_{GLS}=(\mZ_{filt}' \SChol_{\vu}^{-1}(q) \mZ_{filt}^{} )^{-1}\mZ_{filt}' \SChol_{\vu}^{-1}(q) \vy_{filt}$ where $\mZ_{filt}^{}=\MChol_{\vu}(q)\mZ$, and $\vy_{filt}=\MChol_{\vu}(q) \vy$. The premultiplications by $\MChol_{\vu}(q)$ have the effect of filtering and take care of serial correlation. $\SChol_{\vu}^{-1}(q)$ applies scaling and rotation to account for the correlations between the series. The following univariate autoregressive setting exemplifies this intuition.

\begin{example} \label{example:praiswinston}
 A regression model $y_t=\beta t +u_t$ has AR($1$) innovations $u_t=\rho u_{t-1}+\eta_t$ where $\eta_t\stackrel{i.i.d.}{\sim}(0,\sigma^2)$ and $|\rho|<1$. Taking $n=1$, the expressions of Example \ref{example:VAR1} are easily adapted to yield:
 \begin{equation}
   \SChol_{\vu}^{}=
  \diag\left(\frac{\sigma^2}{1-\rho^2},\sigma^2,\ldots,\sigma^2 \right),
  \qquad
    \MChol_{\vu}^{}\vy=
  \begin{bmatrix}
   1		 \\
   -\rho	&1 \\
   \vdots	& \ddots	&\ddots \\
   0		& \cdots	& -\rho	& 1
  \end{bmatrix}
  \begin{bmatrix}
   y_1 \\
   y_2 \\
   \vdots \\
   y_T
  \end{bmatrix}
  =
  \begin{bmatrix}
   y_1 \\
   y_2-\rho y_1 \\
   \vdots \\
   y_T-\rho y_{T-1}
  \end{bmatrix},
 \end{equation}
 and a similar transformation for the linear trend. The implied GLS estimator coincides with the estimator from \cite{praiswinston1954}.
\end{example}

\section{Asymptotic Theory} \label{sec:asymptotictheory}
In this section, we present the asymptotic results. More specifically, we derive: (1) the limiting distribution of the GLS estimator, (2) the fully modified GLS (FM-GLS) estimator that corrects for second order bias terms, (3) a Wald test statistic, and (4) several multivariate KPSS-type of tests for the null of cointegration. We will also compare this FM-GLS estimator with the two fully modified estimators defined in Proposition 1 of \cite{wagnergrabarczykhong2019}. The following assumption will facilitate the development of the asymptotic theory.

\begin{assumption}[Innovation Processes]\label{assumpt1:linearproc}
The innovations processes in the model satisfy the following assumptions:
 \begin{enumerate}[(a)]
  \item The process $\vzeta_t^{}=[\veta_t\tran,\vepsi_t\tran]\tran$ is an independent and identically distributed (i.i.d.) sequence with $\E(\vzeta_t^{} \vzeta_t\tran)=\left[\begin{smallmatrix} \mSigma_{\eta\eta} & \mSigma_{\eta \epsilon} \\ \mSigma_{\epsilon \eta} & \mSigma_{\epsilon \epsilon} \end{smallmatrix} \right] \succ 0$ and $\E(\|\vzeta_t\|^{2r})\leq C_r<\infty$ for some constant $C_r>0$ and some $r>2$.
  \item $\det\big(\lagpol{A}(z)\big)\neq 0$ for all $|z|\leq 1$ and $\sum_{j=0}^\infty j \| \mA_j \|_\calF<\infty$.
  \item $\diff \vx_t=\vv_t$ admits the VAR($\infty$) process $\lagpol{D}(L)\vv_t=\vepsi_t$, where $\lagpol{D}(L)=\mI_n-\sum_{j=1}^\infty \mD_j L^j$. Moreover,  $\det\big(\lagpol{D}(z)\big)\neq 0$ for all $|z|\leq 1$ and $\sum_{j=0}^\infty j \, \|\mD_j \|_{\calF}<\infty$.
 \end{enumerate}
\end{assumption}

The stationary VAR($\infty$) specifications for $\{\vu_t\}$ and $\{\vv_t\}$ are natural given the linear minimum MSE predictor formulae that underly the definitions of the MCBD and BIAM. Moreover, the conditions in Assumption \ref{assumpt1:linearproc} ensure that the lag polynomials $\lagpol{A}(L)$ and $\lagpol{D}(L)$ are invertible (see for example Theorem 7.4.2 of \cite{hannandeistler2012}), thereby showing that our Assumption \ref{assumpt1:linearproc} is similar to the linear processes assumptions that are regularly adopted in the literature on nonlinear cointegration, cf. \cite{choisaikkonen2010}, \cite{wagnerhong2016}, and \cite{wagnergrabarczykhong2019}. The assumption $\det\big(\lagpol{D}(1)\big)\neq 0$ rules out cointegration among the components of $\{\vx_t\}$.

Under Assumption \ref{assumpt1:linearproc}(a), an invariance principle holds for $\vzeta_t$, i.e. $\frac{1}{T^{1/2}} \sum_{t=1}^{[rT]} \vzeta_t\wto \brown_{\vzeta}(r)\equiv \left[\begin{smallmatrix}   \brown_\eta(r)\\  \brown_\epsilon(r) \end{smallmatrix} \right]$ where $\brown_{\vzeta}$ denotes an $2n$-dimensional Brownian motion with covariance matrix $\left[\begin{smallmatrix} \mSigma_{\eta\eta} & \mSigma_{\eta \epsilon} \\ \mSigma_{\epsilon \eta} & \mSigma_{\epsilon \epsilon} \end{smallmatrix} \right] $. Moreover, Assumptions \ref{assumpt1:linearproc}(b)-(c) justify the use of the Beveridge-Nelson decomposition (\cite{phillipssolo1992}). A functional central limit theorem for linear processes is thus also applicable to $\vxi_t^{}=[\vu_t\tran,\vv_t\tran]\tran$, that is
\begin{equation}
 \frac{1}{T^{1/2}} \sum_{t=1}^{[rT]} \vxi_t \wto \brown_{\vxi}(r)
 \equiv
 \begin{bmatrix}
  \brown_u(r) \\
  \brown_v(r)
 \end{bmatrix}
 \equiv
 \begin{bmatrix}
  \lagpol{A}(1)		& \mZeros \\
  \mZeros			& \lagpol{D}(1)
 \end{bmatrix}^{-1}
 \begin{bmatrix}
 	\brown_\eta(r)\\
	\brown_\epsilon(r)
 \end{bmatrix},
 \label{eq:brownianequivalences}
\end{equation}
where the Brownian motion $\brown_{\vxi}(r)$ of dimension $2n$ has covariance matrix
\begin{equation}
\mOmega=
\begin{bmatrix}
  \mOmega_{uu}	& \mOmega_{uv} \\
  \mOmega_{vu}	& \mOmega_{vv}
\end{bmatrix}
=
 \begin{bmatrix}
  \lagpol{A}(1)		& \mZeros \\
  \mZeros			& \lagpol{D}(1)
 \end{bmatrix}^{-1}
\begin{bmatrix} \mSigma_{\eta\eta} & \mSigma_{\eta \epsilon} \\ \mSigma_{\epsilon \eta} & \mSigma_{\epsilon \epsilon} \end{bmatrix}
  \begin{bmatrix}
  \lagpol{A}(1)\tran		& \mZeros \\
  \mZeros				& \lagpol{D}(1)\tran
 \end{bmatrix}^{-1}.
\label{eq:covariancetransform}
\end{equation}
Apart from this long-run covariance matrix $\mOmega=\sum_{h=-\infty}^\infty \E\big(\vxi_t^{}\vxi_{t+h}\tran\big)$, we also introduce the one-sided long-run covariance matrix $\mDelta= \left[\begin{smallmatrix} \mDelta_{uu} & \mDelta_{uv} \\ \mDelta_{vu} & \mDelta_{vv} \end{smallmatrix} \right]= \sum_{h=0}^\infty \E\big(\vxi_t^{}\vxi_{t+h}\tran\big)$. The Brownian motion defined by $\brown_{u.v}=\brown_u-\mOmega_{uv}^{} \mOmega_{vv}^{-1} \brown_v^{}$ is by construction orthogonal to $\brown_v$. Its $(n\times n)$ covariance matrix equals $\mOmega_{u.v}^{}=\mOmega_{uu}^{}-\mOmega_{uv}^{} \mOmega_{vv}^{-1} \mOmega_{vu}^{}$.

\subsection{Infeasible GLS}
We start our analysis assuming that the $(nT\times nT)$ covariance matrix $\mSigma_{\vu}(q)$ is a known quantity for each $q$. The modified Cholesky block decomposition of page \pageref{abr:MCBD} can now be used to derive the limiting distribution of this infeasible GLS estimator. A insightful exposition of our results requires further notation.
 \begin{enumerate}[(a)]
  \item Introduce scaling matrices: $\mG_{\vd_i,T}:=T^{-1/2}\diag[1,T^{-1},\ldots,T^{-d_i}]$ for the time trends, and $\mG_{\vs_i,T}:=T^{-1/2}\diag[T^{-1/2},T^{-1},\dots,T^{-s_i/2}]$ for the stochastic trends. Moreover, we define $\mG_T:=\diag\left[\mG_{1,T},\dots,\mG_{n,T}\right]$, where $\mG_{i,T}:=\diag\left[\mG_{\vd_i,T},\mG_{\vs_i,T}\right]$.
  \item Let $\vd_i(r):=\left[1,r,\dots,r^{d_i}\right]\tran$, $\brown_{s_i}(r):=\left[\brownnormal_{v_i}^{}(r),\brownnormal_{v_i}^2(r),\dots,\brownnormal_{v_i}^{s_i}(r)\right]\tran$ and $\vj_i(r):=\left[\vd_i(r)',\brown_{s_i}(r)'\right]\tran$. Define $d\times n$ block-diagonal random matrix $\mJ(r):=\diag\left[\vj_1(r),\dots,\vj_n(r)\right]$.
  \item $\vb_i=\Big[\vzeros_{d_i+1}\tran,1,2\int_{0}^{1}\brownnormal_{v_i}^{}(r)dr,\dots,s_i\int_{0}^{1}\brownnormal_{v_i}^{s_i-1}(r)dr\Big]\tran$.
 \end{enumerate}
 Finally, we use $\brown_{v_j}$ as shorthand notation for the $j^{th}$ component of $\brown_v$.
 
  \begin{theorem}[Limiting Distribution of the infeasible GLS Estimator]\label{thm:infeasGLS}
  If Assumption \ref{assumpt1:linearproc} holds, and if $q=q(T)$ satisfies $\frac{1}{q}+\frac{q}{T}\to 0$ as $T\to\infty$, then
  \begin{equation}
  \begin{aligned}
	\mG_T^{-1}\left(\widehat{\vbeta}_{GLS}-\vbeta\right) &\wto\left(\myint \mJ(r)\mOmega_{uu}^{-1}\mJ(r)'dr\right)^{-1}\\
	&\qquad\times\left(\myint\mJ(r)\mOmega_{uu}^{-1}d\brown_{u.v}^{}(r)+\myint \mJ(r)\mOmega_{uu}^{-1}\mOmega_{uv}^{}\mOmega_{vv}^{-1}d\brown_{v}^{}(r)+\vbias\right),
\end{aligned}
\label{eq:infeasibleGLSlimiting}
\end{equation}
	where $\vbias=\left[\vbias_{1}',\dots,\vbias_{n}'\right]'$ and $\vbias_{i}=\row_i\Big(\mSigma_{\epsilon\eta}\Big)\col_i\Big(\mSigma_{\eta\eta}^{-1}\Big)~\vb_i$.
\end{theorem}

The limiting result in \eqref{eq:infeasibleGLSlimiting} coincides with the limiting distribution of the MSUR estimator, $\widetilde{\vbeta}_{MSUR}:= \left(\mZ\tran(\mI_T\otimes \widehat{\mOmega}_{uu}^{-1}) \mZ \right)^{-1} \Big( \mZ\tran(\mI_T\otimes \widehat{\mOmega}_{uu}^{-1}) \vy\Big)$, as reported in \cite{wagnergrabarczykhong2019}, see their Proof of Proposition 1. The equivalence of these limiting distributions is caused by the facts that: (1) applying a linear filter to an integrated series only affect its long-run variance (e.g. \cite{phillipspark1988}), and (2) the previous statement remaining true when applying a linear filter to higher integer powers of integrated series. The terms $\myint \mJ(r)\mOmega_{uu}^{-1}\mOmega_{uv}^{}\mOmega_{vv}^{-1}d\brown_{v}^{}(r)$ and $\vbias$ in \eqref{eq:infeasibleGLSlimiting} reflect the presence of second order bias terms caused by serial correlation and endogeneity. In Section \ref{sec:fullymodifiedinf}, we introduce the fully modified (FM) correction that adjust these bias terms and leads to standard inference. We first introduce a feasible version of the GLS estimator.

\subsection{Consistent Estimation of $\mSigma_{\vu}^{-1}(q)$ and Feasible GLS}\label{subsec:covarianceconsistency}
Up to this point we have discussed the infeasible estimator  $\widehat{\vbeta}_{GLS}:=\left(\mZ\tran \mSigma_{\vu}^{-1}(q) \mZ \right)^{-1} \mZ\tran \mSigma_{\vu}^{-1}(q) \vy$. A feasible GLS approach requires a consistent estimator of the $(nT\times nT)$ matrix $\mSigma_{\vu}^{-1}(q)$. Several authors, e.g. \cite{wupourahmadi2009} and \cite{mcmurypolitis2010}, have constructed consistent estimators of large covariance matrices using banding or tapering to reduce the number of unknown parameters. Direct usage of their results poses two difficulties because: (1) numerical inversion of large matrices is computationally expensive for large $nT$, and (2) matrix inversion might even be impossible because the estimated covariance matrix cannot be guaranteed to be positive definite. In the light of the such considerations, we will estimate $\mSigma_{\vu}^{-1}(q)$ directly and ensure it to be positive definite. The approach is the sample counterpart of the BIAM described on page \pageref{ref:BIAM}. That is, we replace true innovations by first stage OLS residuals $\widehat{\vu}_t=\vy_t-\mZ_t\widehat{\vbeta}_{OLS}$, and subsequently minimise a sample moment in estimated residuals rather than the population mean squared forecasting error. This method was previously used by \cite{chengingyu2015} and \cite{ICG2016} for univariate time series. For a multivariate time series, we define
\begin{equation}
\begin{aligned}
\widehat{\mA}(\ell)&=
 \begin{bmatrix}
  \widehat{\mA}_1(\ell) & \cdots & \widehat{\mA}_\ell(\ell)
 \end{bmatrix}
=
\argmin_{(\mTheta_1,\ldots,\mTheta_\ell)\in \SR^{n\times n\ell}}\sum_{t=\ell+1}^{T}\left\|\widehat{\vu}_t-\mTheta_1 \widehat{\vu}_{t-1}-\cdots-\mTheta_\ell \widehat{\vu}_{t-\ell} \right\|^2\\
\widehat{\mS}(\ell)&=\frac{1}{T-\ell}\sum_{t=\ell+1}^{T}\left[\widehat{\vu}_t-\widehat{\mA}_1(\ell) \widehat{\vu}_{t-1}-\cdots-\widehat{\mA}_\ell(\ell)\widehat{\vu}_{t-\ell}\right]\left[\widehat{\vu}_t-\widehat{\mA}_1(\ell) \widehat{\vu}_{t-1}-\cdots-\widehat{\mA}_\ell(\ell)\widehat{\vu}_{t-\ell}\right]\tran,
\end{aligned}
\label{eq:sampleMCD}
\end{equation}
$1\leq \ell\leq q$, and $\widehat{\mS}(0)=\frac{1}{T}\sum_{t=1}^{T}\widehat{\vu}_t^{}\widehat{\vu}_t\tran$. Similarly to \eqref{eq:bimam}-\eqref{eq:Mmatrixconstrucion}, we subsequently construct the matrices  $\widehat{\MChol}_{\vu}(q)=\left[\widehat{\vm}_{\vu}^{ij}(q) \right]_{1\leq i,j\leq T}$ and $\widehat{\SChol}_{\vu}(q)=\diag\left(\widehat{\mS}(0),\widehat{\mS}(1),\ldots,\widehat{\mS}(q),\ldots,\widehat{\mS}(q)\right)$, and obtain the implied multivariate BIAM estimator as
\begin{equation}
\widehat{\mSigma_{\vu}^{-1}}(q)=\widehat{\MChol}_{\vu}\tran(q)\widehat{\SChol}_{\vu}^{-1}(q)\widehat{\MChol}_{\vu}^{}(q).
\label{eq:BIAM}
\end{equation}

\begin{assumption}\label{assump3:residuals}
For $\widehat{\vu}=[\widehat{\vu}_1\tran,\ldots,\widehat{\vu}_T\tran]\tran$ and $\vu=[\vu_1\tran,\ldots,\vu_T\tran]\tran$, assume $\|\widehat{\vu}-\vu\|^{2}=O_p(1)$.
\end{assumption}

\begin{assumption}\label{assump4:bandingparameter}
Assume $q=q_T$ satisfies $\frac{1}{q_T}+\frac{q_T^3}{T}\to 0$ as $T\to\infty$.
\end{assumption}

Assumption \ref{assump3:residuals} requires the residuals to be sufficiently close to the true innovations. It is a rather mild assumption and it is satisfied if residuals are computed by least squares. Assumption \ref{assump4:bandingparameter} places constraints on the banding parameter $q_T$. First, Assumption \ref{assump4:bandingparameter} requires the banding parameter to diverge with sample size. This ensures that no nonzero elements are (asymptotically) set to zero. Moreover, the assumption $q_T^3/T\to 0$ establishes an upper bound for the growth rate of $q_T$. The definition of $\widehat{\mA}(\ell)$, see \eqref{eq:sampleMCD}, shows that we are fitting a vector autoregression (VAR) of increasing lag order to the residuals. Identical rate requirements are reported by \cite{lewisreinsel1985} when they derive consistency and asymptotic normality results when finite VAR models are fitted to infinite order VAR processes. The following theorem shows the consistent estimation of $\mSigma_{\vu}^{-1}$ and implies that the infeasible and feasible GLS estimator have the same limiting distribution.

\begin{theorem}[Consistent Estimation of $\mSigma_{\vu}^{-1}$]\label{thm:consistentDECOMP}
 If Assumptions \ref{assumpt1:linearproc}-\ref{assump4:bandingparameter} hold, then
 	\begin{equation}
	\begin{aligned}
	\left\|\widehat{\mSigma_{\vu}^{-1}}(q_T)-\mSigma_{\vu}^{-1}\right\|
	&\leq\Big\|\widehat{\mSigma_{\vu}^{-1}}(q_T)-\mSigma_{\vu}^{-1}(q_T)\Big\|+\Big\|\mSigma_{\vu}^{-1}(q_T)-\mSigma_{\vu}^{-1}\Big\| \\
	&=O_p\left(\,\sqrt{q_T^3/T} \, \right)+O\left(\frac{1}{\sqrt{q_T}}\sum_{s=q_T+1}^{\infty}s\left\|\mA_{s}\right\|_{\calF}\right) \pto 0 \text{ as $T\rightarrow \infty$}.
	\end{aligned}
	\end{equation}
\end{theorem}

\subsection{Fully Modified Inference} \label{sec:fullymodifiedinf}
The asymptotic results of Theorem \ref{thm:infeasGLS} is not immediately useful for statistical inference. There are two difficulties. First, the second order bias dislocates the limiting distribution which can translate into substantial finite sample bias. This leads to a loss in efficiency. Second, possible dependencies between the Brownian motions $\brown_u$ and $\brown_v$ cause the limiting distribution to depend on nuisance parameters. Critical values would therefore be nuisance parameter dependent as well.

These two issues have received extensive attention in the linear cointegration literature. A (non-exhaustive) list of solution methods is: joint modeling as in \cite{phillips1991}, \cites{saikkonen1992} dynamic least squares, and the integrated modified OLS and fixed-b approaches by \cite{vogelsangwagner2014}. We rely on the fully modified (FM) approach advocated by \cite{phillipshansen1990} and \cite{phillips1995}. The idea is a twofold modification of the estimator: (1) second order bias terms are subtracted, and (2) a transformation of the dependent variable is introduced to obtain a zero-mean Gaussian mixture limiting distribution. Recently, \cite{wagnergrabarczykhong2019} have proposed two estimators within the framework of seemingly unrelated cointegrating polynomial regressions. These estimators, FM-SOLS and FM-SUR, rely on kernel estimators of the one- and two-sided long-run covariance matrix (see Theorem \ref{thm:fullymodifiedOLSGLS}). As such, we introduce the following assumption.

\begin{assumption}[Consistent Estimation of Long-run Covariance Matrices]\label{assumpt:consistentlongrun}
 $\widehat{\mOmega}$ and $\widehat{\mDelta}$ are consistent kernel estimators of the long-run covariance matrix $\mOmega$ and the one-sided long-run covariance matrix $\mDelta$, respectively.
\end{assumption}

\cite{andrews1991} and \cite{neweywest1994} use kernel estimators for long-run covariance estimation. Their method involves the calculation of weighted sums of the autocovariance matrices of the residuals. These weights are determined by a kernel function and bandwidth parameter. Our Assumption \ref{assumpt:consistentlongrun} is easily satisfied by imposing suitable conditions on the kernel function and bandwidth parameter. We refer to \cite{phillips1995} and \cite{jansson2002} for an enumeration of such conditions.

Alternatively, we can obtain consistent one- and two-sided long-run covariance estimators within the BIAM framework of Section \ref{subsec:covarianceconsistency}.\footnote{An overview of the procedure is given here. Section \ref{detailsFMinference} in the Supplement provides further details.} This approach resembles \cite{berk1974}. The GLS estimator and its FM counterpart are thus constructed within a single framework. The estimators are as follows. For all $t=1,2,\ldots,T$, we first stack $\widehat{\vu}_t$ and $\Delta \vx_t = \vv_t$ in the $2n$-dimensional vector $\widehat{\vxi}_t=[\widehat{\vu}_t\tran,\Delta \vx_t\tran]\tran$. Since the BIAM estimator is fitting VAR processes up to order $q_T$, we will use the estimated VAR($q_T$) approximations to define the long-run covariance estimators. For $\mOmega$, the estimator is $\widehat{\mOmega}_{q_T}=\left(\mI_{2n}-\sum_{j=1}^{q_T} \widehat{\mF}_j^{}(q_T) \right)^{-1} \widehat{\mSigma}_{q_T}\left(\mI_{2n}-\sum_{j=1}^{q_T} \widehat{\mF}_j\tran(q_T) \right)^{-1}\label{eq:twosidedLRV}$, where $\widehat{\mSigma}_{q_T}=\widehat{\mS}(q_T)$ and $\widehat{\mF}_j(q_T)$ denote respectively the estimated prediction error variance and the coefficient matrix of the $j^{th}$ lag when a VAR($q_T$) is fitted to $\{\widehat{\vxi}_t \}_{t=1}^T$. The population one-sided long-run covariance matrix is $\mDelta= \sum_{h=0}^\infty \E\big(\vxi_t^{}\vxi_{t+h}\tran\big)$. It is thus intuitive to approximate this quantity by a finite sum of estimated covariance matrices of $\{\widehat{\vxi}_t \}_{t=1}^T$. These covariance matrices are nothing but subblocks of the matrix $\widehat{\mSigma}_{\vxi}^{}(q_T)=\widehat{\MChol}_{\vxi}^{-1}(q_T)\widehat{\SChol}_{\vxi}^{}(q_T)\widehat{\MChol}_{\vxi}^{-1\prime}(q_T)$.\footnote{We use $\mSigma_{\vxi}$ to denote the $(2nT\times 2nT)$ matrix $\E(\vxi \vxi\tran)$ where $\vxi=[\vxi_1\tran,\vxi_2\tran,\ldots,\vxi_T\tran]\tran$. The matrices $\widehat{\MChol}_{\vxi}(q)$ and $\widehat{\SChol}_{\vxi}(q)$ are defined similarly to respectively $\widehat{\MChol}_{\vu}(q)$ and $\widehat{\SChol}_{\vu}(q)$ (see page \pageref{eq:BIAM}). The matrix $\widehat{\MChol}_{\vxi}(q_T)$ is lower triangular with identity matrices on the main diagonal. Therefore, its matrix inverse exists and is fast to compute.} We therefore use
\begin{equation}
 \widehat{\mDelta}_{q_T,r_T}= \mQ_{r_T}\tran \widehat{\mSigma}_{\vxi}^{}(q_T) \mQ_1^{}
\end{equation}
where $\mQ_r=\left[\mZeros_{2n \times 2n}, \cdots,  \mZeros_{2n \times 2n}, \mI_{2n}, \cdots, \mI_{2n}\right]\tran$ is an $\big(2n T \times 2n \big)$ block matrix of zeros of which the last $r$ blocks have been replaced by identity matrices. To ensure consistency, we place the following rate restriction on the number of included autocovariance matrices.

\begin{assumption} \label{assumpt:rT}
 As $T\rightarrow\infty$, $r_T\rightarrow\infty$, $\frac{r_T^{}q_T^3}{T}\rightarrow 0$, and $r_T=O(q_T)$.
\end{assumption}

Definitions and limiting results for FM estimators are presented in Theorem \ref{thm:fullymodifiedOLSGLS}. The FM-SOLS, FM-SUR and FM-GLS estimator all depend on estimators for $\mDelta$ and $\mOmega$. It is only the consistency of these estimators that is relevant for the asymptotic analysis, not whether the kernel or BIAM approach is employed. As such, we will not complicate notation by introducing additional notation to indicate whether the kernel or BIAM approach is used. In subsequent theorems, simulation results and the empirical application we will use kernel estimators for FM-SOLS and FM-SUR, and the BIAM approach for FM-GLS. This seems to be the logical choice for these estimators.

\begin{theorem}\label{thm:fullymodifiedOLSGLS}
For $i=1,\ldots,n$, define $\widehat{\vb}_i=\left[\vzeros_{d_i+1}\tran,T, 2 \sumtT x_{it},\ldots, s_i \sumtT x_{it}^{s_i-1} \right]\tran$, for $i=1,\ldots,n$. Also, define the $(n\times n)$ matrix $\widehat{\mDelta}_{vu}^+$ as the (implied) consistent estimator of $\mDelta_{vu}^+=\mDelta_{vu}^{}-\mDelta_{vv}^{}\mOmega_{vv}^{-1} \mOmega_{vu}^{}$.
 \begin{enumerate}[(a)]
  
  \item Define the FM-SOLS estimator as
  \begin{equation}
    \widehat{\vbeta}_{SOLS}^+ = \left( \mZ\tran \mZ \right)^{-1} \left( \mZ\tran \vy^+ -\widehat{\mA}  \, \right),
  \end{equation}
  where $\vy^+:=[\vy_{1}^{+\prime},\vy_{2}^{+\prime},\ldots,\vy_{T}^{+\prime}]\tran$ with $\vy_t^+ = \vy_t^{}-\widehat{\mOmega}_{uv}^{} \widehat{\mOmega}_{vv}^{-1} \diff \vx_t^{}$, and $\widehat{\mA}:= [\widehat{\mA}_1\tran,\ldots,\widehat{\mA}_n\tran]\tran$ with $\widehat{\mA}_i= \widehat{\mDelta}_{v_i u_i}^+\widehat{\vb}_i$ and $\widehat{\mDelta}_{v_i u_i}^+$ being the $i^{th}$ element on the main diagonal of $\widehat{\mDelta}_{v u}^+$. If Assumptions \ref{assumpt1:linearproc} and \ref{assumpt:consistentlongrun} hold, then
  \begin{equation}
   \mG_T^{-1} \left( \widehat{\vbeta}_{SOLS}^+  - \vbeta \right) \wto \left( \myint \mJ(r) \mJ(r)\tran dr \right)^{-1} \myint  \mJ(r)  d\brown_{u.v}(r).
  \end{equation}
  
  \item Define the FM-SUR estimator as
 \begin{equation}
  \widehat{\vbeta}_{SUR}^+ = \left(\mZ\tran (\mI_T\otimes \widehat{\mOmega}_{u.v}^{-1}) \mZ  \right)^{-1} \left( \mZ\tran (\mI_T\otimes \widehat{\mOmega}_{u.v}^{-1}) \vy^+ - \widetilde{\mA}^* \right),
 \end{equation}
 where $\widetilde{\mA}^*:= [\widetilde{\mA}_1^*,\ldots,\widetilde{\mA}_n^*]$ with $\widetilde{\mA}_i^*= \row_i\left(\widehat{\mDelta}_{vu}^+\right) \col_i\left( \widehat{\mOmega}_{u.v}^{-1} \right) \widehat{\vb}_i$. If Assumptions \ref{assumpt1:linearproc} and \ref{assumpt:consistentlongrun} hold, then
 \begin{equation}
  \mG_T^{-1} \left( \widehat{\vbeta}_{SUR}^+ - \vbeta \right) \wto \left( \myint \mJ(r) \mOmega_{u.v}^{-1} \mJ(r)\tran dr \right)^{-1} \myint \mJ(r) \mOmega_{u.v}^{-1} d\brown_{u.v}^{}(r).
 \end{equation}
 
 \item Define the FM-GLS estimator as
 \begin{equation}
  \widehat{\vbeta}_{FGLS}^+=\Big(\mZ'\widehat{\mSigma_{\vu}^{-1}}(q)\mZ\Big)^{-1}\left[\mZ'\widehat{\mSigma_{\vu}^{-1}}(q)\vy-\mZ'\Big(\mI_T\otimes \widehat{\mOmega}_{uu}^{-1}\widehat{\mOmega}_{uv}^{}\widehat{\mOmega}_{vv}^{-1}\Big)\vv-\widehat{\vbias}^{+}\right],
 \label{eq:FMGLSestimator}
 \end{equation}
 where $\vv := [\diff \vx_1\tran,\ldots,\diff \vx_T\tran]\tran = [\vv_1\tran,\ldots,\vv_T\tran]\tran$, and $\widehat{\vbias}^{+}=\big[\widehat{\vbias}_1^{+\prime},\dots,\widehat{\vbias}_n^{+\prime}\big]\tran$ with
 $$\widehat{\vbias}_{i}^{+}=\left[\row_i\Big(\widehat{\mSigma}_{\epsilon\eta}\Big)\col_i\Big(\widehat{\mSigma}_{\eta\eta}^{-1}\Big)-\row_i\Big(\widehat{\mDelta}_{vv}\Big)\col_i\Big(\widehat{\mOmega}_{vv}^{-1}\widehat{\mOmega}_{vu}^{}\widehat{\mOmega}_{uu}^{-1}\Big)\right]~\widehat{\vb}_i.
 $$
 If Assumptions \ref{assumpt1:linearproc}-\ref{assump4:bandingparameter} and \ref{assumpt:rT} hold, then
 \begin{equation}
  \mG_T^{-1} \left( \widehat{\vbeta}_{FGLS}^+ - \vbeta \right) \wto \left(\myint \mJ(r)\mOmega_{uu}^{-1}\mJ(r)'dr\right)^{-1} \myint \mJ(r)\mOmega_{uu}^{-1}d\brown_{u.v}^{}(r).
 \end{equation}
 \end{enumerate}
\end{theorem}

The FM-GLS estimator is new to the seemingly unrelated CPR literature, whereas the FM-SOLS and FM-SUR estimators have recently appeared in \cite{wagnergrabarczykhong2019}. Theorem \ref{thm:fullymodifiedOLSGLS} indicates that all three estimators have a zero-mean Gaussian mixture limiting distribution implying that standard inference is applicable for each. However, we also see from Theorem \ref{thm:fullymodifiedOLSGLS} that the limiting distributions are generally different because different types of weighing are used in the construction of the estimators.\footnote{There are special cases in which some (pairs of) estimators become asymptotically equivalent. For example, if $n=1$, then all estimators are asymptotically equivalent because the weighting matrices $\mOmega_{u.v}^{-1}$ and $\mOmega_{uu}^{-1}$ are now scalars. Also, under exogeneity, we have $\mOmega_{uu} = \mOmega_{u.v}$ and the FM-SUR an FM-GLS estimators share the same limiting distribution.} 

For completeness, we also detail how the FM-GLS estimator can be used to test linear hypotheses. A formal presentation of such a result is more involved because of the different convergence rates of the individual parameter estimators. That is, the parameters with the lowest convergence rate will dominate the asymptotic distribution and one should take care to avoid a degenerate limiting distribution. We will rule out such complications by considering hypothesis tests on individual parameters.\footnote{For general linear hypothesis, we refer the reader to \cite{simsstockwatson1990} where a reordering based on convergence rates is used to establish the limiting distribution of the Wald $F$ statistic for general linear hypothesis. The same approach is applicable in our setting but we will not explore this in greater detail.}  Therefore, let $\mR$ denote a $(k \times s)$ selection matrix in which every row contains a single 1 and zeros otherwise. The null hypothesis $\mR \vbeta= \vr$ can be tested using the standard chi-squared limiting distribution of the Wald statistic (Theorem \ref{thm:test}). These tests are practically relevant. For example, exclusion restrictions of the type $\mR \vbeta=\vzeros$ allow us to test whether the cointegrating relation is linear.

\begin{theorem}\label{thm:test}
	Consider the null hypothesis $H_0:\mR\vbeta=\vr$, which imposes $k$ linearly independent restrictions. Under the assumptions of Theorem \ref{thm:fullymodifiedOLSGLS}(c), the Wald test statistic
	\begin{equation}
	\calW=\Big(\mR\widehat{\vbeta}_{FGLS}^+-\vr\Big)'\widehat{\mPhi}^{-1}\Big(\mR\widehat{\vbeta}_{FGLS}^+-\vr\Big)\wto \chi_k^2,
	\end{equation}
	where $\widehat{\mPhi}=\mR\left[\mZ\tran \left( \mI_T \otimes \widehat{\mOmega}_{uu}^{-1} \right) \mZ \right]^{-1}\left[\mZ'\Big(\mI_T\otimes \widehat{\mOmega}_{uu}^{-1}\widehat{\mOmega}_{u.v}^{}\widehat{\mOmega}_{uu}^{-1}\Big)\mZ\right]\left[\mZ\tran \left( \mI_T \otimes \widehat{\mOmega}_{uu}^{-1} \right) \mZ \right]^{-1}\mR'$.
\end{theorem}

\subsection{Testing the Null of Cointegration} \label{sec:cointtest}
Stationarity tests are used to avoid spurious regressions and to verify the correct specification of the cointegrating relation. To test for stationarity of the seemingly unrelated cointegrating polynomial regressions (SUCPR) errors, we combine the test statistic from \cite{nyblomharvey2000} with the sub-sampling approach found in \cite{choisaikkonen2010} and \cite{wagnerhong2016}. We consider three test statistics. To treat all test statistics in a unified framework, we define
\begin{equation}
 \vvarphi_{j,b}(\{\vx\}) = \left[ \vx_j\tran ,  \sum_{s=j}^{j+1} \vx_s\tran,\ldots, \sum_{s=j}^{j+b-1} \vx_s\tran \right]\tran,
\end{equation}
that is, a vector of length $nb$ stacking the cumulative sums of $\{\vx_j,\ldots, \vx_{j+b-1} \}$. If the true innovations $\vu_1,\ldots,\vu_T$ were observed, then we could use the full-sample KPSS-type of test statistic $\frac{1}{T^2} \vvarphi_{1,T}(\{\vu\})\tran(\mI_T\otimes \widehat{\mOmega}_{uu}^{-1})\vvarphi_{1,T}(\{\vu\})=\tr\left[  \widehat{\mOmega}_{uu}^{-1} \frac{1}{T^2} \sum_{t=1}^T \left( \sum_{s=1}^t \vu_s \right)\left( \sum_{s=1}^t \vu_s \right) \tran \right]$ to test for stationarity of the innovations. Under the null of stationarity, this test statistic would converge weakly to $\myint \| \bm W(r) \|^2 dr$ with $\bm W(r)$ denoting an $n$-dimensional standard Brownian motion. This limiting distribution is free of nuisance parameters and the cumulative distribution function is available as a series expansion (see the Supplement).

The innovations $\vu_1,\ldots,\vu_T$ are only available when cointegrating relations are pre-specified. If these coefficients are estimated, then this additional parameter uncertainty will contaminate the limiting distribution with nuisance parameters.\footnote{There are exceptions. \cite{shin1994} reports a nuisance parameter free limiting distribution for a single-equation linear cointegrating relation. This remains true if only a single integrated variable enters the cointegrating regression with a higher power, see Proposition 5 in \cite{wagnerhong2016}.} The idea behind the subsampling approach is to construct a test statistic incorporating $b = b_T$ residuals while computing parameter estimators from all $T$ observations. If $b_T$ increases slowly with sample size, then the parameter estimation error will be negligible relative to the randomness in the errors and the asymptotic distribution remains $\myint \| \bm W(r) \|^2 dr$.
 
The three KPSS-type of test are based on the following residuals: $\hat{\vu}_{t,SOLS}^+= \vy_t^+-  \mZ_t \widehat{\vbeta}_{SOLS}^+$, $\hat{\vu}_{t,SUR}^+= \vy_t^+ -  \mZ_t \widehat{\vbeta}_{SUR}^+$, and $\hat{\vu}_{t,FGLS}=\vy_t-  \mZ_t \widehat{\vbeta}_{FGLS}^+$. The test statistic are:
\begin{equation}
 K_{j,b_T}^i = \frac{1}{b_T^2}  \vvarphi_{j,b_T}(\{\hat{\vu}_i^+\})\tran \left(\mI_{b_T} \otimes \widehat{\mOmega}_{u.v}^{-1} \right)\vvarphi_{j,b_T}(\{\hat{\vu}_i^+\}),\qquad\qquad \text{for }i\in\{SOLS, SUR\},
\label{eq:KPSSresidualtype}
\end{equation}
and
\begin{equation}
 K_{j,b_T}^{BIAM} = \frac{1}{b_T^2}  \vvarphi_{j,b_T}(\{\hat{\vu}_{FGLS}\})\tran \widehat{\mSigma_{\vu}^{-1}}(q_T,b_T) \vvarphi_{j,b_T}(\{\hat{\vu}_{FGLS}\}),
 \label{eq:KPSSbiam}
\end{equation}
where $\widehat{\mSigma_{\vu}^{-1}}(q_T,b_T)$ is the $(n b_T\times n b_T)$ submatrix of $\widehat{\mSigma_{\vu}^{-1}}(q_T)$ obtained by selecting the rows and columns related to all time indices in the set $\{ n(T-b_T)+1, n(T-b_T)+2 ,\ldots,nT\}$. The test statistic in \eqref{eq:KPSSbiam} fits naturally into the FM-GLS estimation framework.

\begin{theorem}\label{thm:kpss_subtest}
 Let the assumptions from Theorem \ref{thm:fullymodifiedOLSGLS} hold.
 \begin{enumerate}[(a)]
  \item  If $\frac{1}{b_T}+\frac{b_T}{T}\to 0$ as $T\to \infty$, then
  $$
   K_{j,b_T}^i \wto \myint \| \wiener(r) \|^2 dr, \qquad 1 \leq j \leq T-b_T+1, \qquad \text{for }i\in\{SOLS, SUR \}.
  $$
  \item If $\frac{q_T}{b_T}+\frac{b_T}{T}\to 0$ as $T\to \infty$, then $ K_{j,b_T}^{BIAM}\wto \myint \| \wiener(r) \|^2 dr$ for any $1 \leq j \leq T-b_T+1$.
 \end{enumerate}
\end{theorem}

A sample of size $T$ allows for up to $M=\lfloor T/b_T\rfloor$ series of nonoverlapping blocks of residuals of length $b_T$. Similarly to \cite{choisaikkonen2010}, we apply the Bonferroni procedure to use all these series and thereby increase power. The approach is applicable to any of the three test statistics in Theorem \ref{thm:kpss_subtest}. As such, we keep the notation general and use a generic $K_j$ to denote a test statistic based on the $j^{th}$ subseries, $j=1,2,\ldots,M$. In the Bonferrroni procedure we compute $K_{max}=\left\{K_1,K_2,\ldots,K_M\right\}$ and do not reject the null hypothesis whenever $K_{max}\leq c_{\alpha/M}$ with $c_{\alpha/M}$ defined by $\p\left(\myint \| \wiener(r) \|^2 dr\geq c_{\alpha/M}\right)=\alpha/M$. The Bonferroni inequality implies $\lim_{T\to \infty} \p\left(K_{max} \leq c_{\alpha/M} \right) \geq 1-  \lim_{T\to\infty} \sum_{j=1}^M\p\left(K_j>c_{\alpha/M} \right)=1- \alpha$ and we see that the probability of a type-I error does not exceed the significance level $\alpha$.

\begin{remark}
We suggest to follow \cite{choisaikkonen2010} in terms of the implementation of the subsampling approach. That is, the block size $b_T$ is selected using the minimum volatility rule by \cite{romanowolf2001}. For this particular block size we subsequently select subsamples by taking non-overlapping blocks from alternatively the start and the end of the sample.
\end{remark}

\begin{remark}
The limiting results in Theorem \ref{thm:kpss_subtest} guarantee a correct asymptotic size. Our simulations show (1) that these tests have power against various alternative hypotheses and (2) that power increases with sample size. A theoretical investigation of the power properties is outside of the scope of this paper.
\end{remark}

\section{Simulations} \label{sec:simulations}

We now study the finite sample performance of the estimators and stationarity tests. First, we compare the FM-GLS estimator with the FM-SOLS and FM-SUR estimators from \cite{wagnergrabarczykhong2019}. All long-run covariance matrices are computed using a Bartlett kernel and the automatic bandwidth selection approach due to \cite{andrews1991}. For FM-GLS, the banding parameter $q_T$ is selected using the subsampling and risk-minimization approach explained in section 5 from \cite{bickellevina2008}.\footnote{More details concerning the implementation can be found in the Supplement.} Infeasible counterparts of the estimator are constructed assuming the knowledge of the true serial correlation and/or cross-sectional dependence pattern. These estimators are denoted by infSOLS, infSUR, and infGLS. Second, we look at the cointegration tests. We consider three test statistics: $K^{SOLS}$ and $K^{SUR}$ use the residuals as in \eqref{eq:KPSSresidualtype}, whereas $K^{BIAM}$ employs the pre-filtered residuals from \eqref{eq:KPSSbiam}. All tests are implemented with minimum volatility block size selection and Bonferroni correction.

We consider $T\in\{100,200,500\}$ and $n\in\{3,5\}$. All tests are performed at a nominal significant level of $5\%$. For stationary processes, a presample of 200 observations is used to remove the influence of the starting values. All results are based on $2.5\times 10^4$ Monte Carlo replicates.

\subsection{Monte Carlo Designs}

We generate data according to a quadratic seemingly unrelated CPR. That is, we adopt the DGP in \eqref{eq:basemodel} with $\vz_{it}=\big[1,t,x_{it},x_{it}^2\big]\tran$. The integrated variables satisfy $\vx_0=\vzeros$ and $\diff \vx_t = \vv_t$. We explore two error processes.

\bigskip
\noindent\textbf{Setting A (Errors as in \cite{wagnergrabarczykhong2019})}: As a benchmark, we revisit the simulation setting in \cite{wagnergrabarczykhong2019} and generate innovations according to
\begin{equation} 
\vu_t=\rho_1\vu_{t-1}+\vepsi_{t}+\rho_2\ve_t,\qquad \vv_t=\ve_t+0.5\ve_{t-1},
\end{equation}
where $\vepsi_t\stackrel{i.i.d.}{\sim}\rN\big(\vzeros,\mSigma(\rho_3)\big)$, $\ve_t\stackrel{i.i.d.}{\sim}\rN\big(\vzeros,\mSigma(\rho_4)\big)$ and
\begin{equation}\label{eq:toeplitz_structure}
\mSigma(\rho)=\begin{bmatrix}
1      & \rho   & \cdots & \rho\\
\rho   &   1    & \cdots & \rho\\
\vdots & \vdots & \ddots & \vdots\\ 
\rho   & \rho   & \cdots & 1\\
\end{bmatrix}
\end{equation}
is a symmetric Toeplitz matrix. The parameter $\rho_1$ controls the level of serial correlation and $\rho_2$ measures the degree of endogeneity. The parameters $\rho_3$ and $\rho_4$ indicate the extent of correlation across equations induced through $\vepsi_t$ and $\ve_t$, respectively. For simplicity, we assume identical values $\rho_1=\rho_2=\rho_3=\rho_4=\rho\in\{0,0.3,0.6,0.8\}$. The true coefficient vector is $\vbeta=\big[\vbeta_1',\dots,\vbeta_n'\big]'$, where $\vbeta_i=[1,1,5,\beta_{i,4}]'$ with $\beta_{i,4}=-0.3$, $i=1,\dots,n$. 

\bigskip
\noindent\textbf{Setting B (VARMA Errors)}: To further investigate the importance of serial correlation, we consider a second specification of the innovation process:
\begin{equation}
\vu_t=\mLambda_1\vu_{t-1}+\veta_t+\mLambda_2\veta_{t-1},\qquad \vv_t=\mLambda_3\vv_{t-1}+\vepsi_{t},
\end{equation}
where $\veta_t$ and $\vepsi_{t}$ are generated as $\left[\begin{smallmatrix}
\veta_t\\
\vepsi_t
\end{smallmatrix}\right]\stackrel{i.i.d.}{\sim}
\rN\big(\vzeros,\mSigma(\theta)\big)$ and $\mSigma(\theta)\in \SR^{2n\times 2n}$ as in \eqref{eq:toeplitz_structure} but with parameter $\theta$. The matrices $\mLambda_i$ ($i=1,2,3$) are generated independently and similarly to \cite{chang2004}. That is, we take the following three steps:
\begin{enumerate}[(a)]
	\item Generate an $n\times n$ random matrix $\mU_i$ from $\text{U}[0,1]$ and construct the orthogonal matrix $\mH_i=\mU_i^{}\Big(\mU_i'\mU_i^{} \Big)^{-1/2}$.
	\item Generate $n$ eigenvalues $\lambda_{i1},\dots,\lambda_{in}\stackrel{i.i.d.}{\sim}\text{U}\big[\underline{\lambda},\bar{\lambda}\big]$.
	\item Let $\mL_i=\diag\left(\lambda_{i1},\dots,\lambda_{in}\right)$ and compute $\mLambda_i=\mH_i^{}\mL_i^{}\mH_i'$. 
\end{enumerate}
The parameter $\theta\in\{0.3,0.5\}$ governs regressor-error correlation and cross-equation correlation. The amount of serial correlation is specified through $\underline{\lambda}$ and $\bar{\lambda}$. The three scenarios $\big(\underline{\lambda},\bar{\lambda}\big)\in \big\{\left(0.1,0.5\right),\left(0.5,0.8\right),\left(0.8,0.95\right)\big\}$ steadily increase the autocorrelation in the generated data.

\bigskip
\noindent\textbf{Setting C (Cointegration Tests)}: We continue to construct innovations according to Setting B. Moreover, we fix $\left[\begin{smallmatrix}
\veta_t\\
\vepsi_t
\end{smallmatrix}\right]\stackrel{i.i.d.}{\sim}
\rN\big(\vzeros,\mSigma(\theta)\big)$ with $\theta=0.3$, and we construct the matrices $\mLambda_2$ and $\mLambda_3$ using $\big(\underline{\lambda},\bar{\lambda}\big)=(0.1,0.5)$. The eigenvalues of $\mLambda_1$ are varied to explore both size and power properties. We always estimate a \emph{quadratic} seemingly unrelated CPR.

\begin{description}
	\item[Size DGP.] We generate the eigenvalues of $\mLambda_1$ as before. That is, take $\lambda_{11},\dots,\lambda_{1n}\stackrel{i.i.d.}{\sim}\text{U}\big[\underline{\lambda},\bar{\lambda}\big]$, where $\big(\underline{\lambda},\bar{\lambda}\big)\in \left\{\left(0.1,0.5\right),\left(0.5,0.8\right),\left(0.8,0.95\right)\right\}$.
	\item[Power DGP1.] We set $\lambda_{1j}=1$ for $1\leq j\leq J_1$ and generate $\lambda_{1j}\stackrel{i.i.d.}{\sim}\text{U}\big[0.1,0.5\big]$ for $J_1+1\leq j\leq n$. The integer $J_1\in\{1,2,n\}$ represents the number of unit roots in $\{\vu_t\}$.
	\item[Power DGP2.] The eigenvalues of $\mLambda_1$ are sampled as $\lambda_{11},\dots,\lambda_{1n}\stackrel{i.i.d.}{\sim}\text{U}\big[0.1,0.5\big]$, and the first $J_2\in\{1,2,n\}$ series follow a cubic SUCPR specification:
	\begin{equation*}
	y_{it}=\begin{cases}
	1+t+5x_{it}-0.3x_{it}^2+0.01x_{it}^3+u_{it},&\quad 1\leq i\leq J_2,\\
	1+t+5x_{it}-0.3x_{it}^2+u_{it},& \quad J_2+1\leq i\leq n.
	\end{cases}
	\end{equation*}
	\item[Power DGP3.] We again take $\lambda_{11},\dots,\lambda_{1n}\stackrel{i.i.d.}{\sim}\text{U}\big[0.1,0.5\big]$ and construct
	\begin{equation*}
	y_{it}=\begin{cases}
	\sum_{s=1}^{t} u_{is},&\quad 1\leq i\leq J_3,\\
	1+t+5x_{it}-0.3x_{it}^2+u_{it},& \quad J_3+1\leq i\leq n,
	\end{cases}
	\end{equation*}
	where $J_3\in\{1,2,n\}$ represents for the number of equations that specify a spurious relation.
\end{description}
Overall, the Power DGPs 1-3 consider: missing $I(1)$ regressors, omitted higher order powers of the $I(1)$ regressor $x_{it}$, and spurious regressions, respectively.

\subsection{Discussion of the Simulation Results}
Tables \ref{table:efficiencyQSUCPR} and \ref{tab:efficiency_varma} report the empirical mean squared error (MSE) for both feasible and infeasible estimators. As results are qualitatively similar across equations, we only report on the estimators for $\beta_{1,4}$ (the coefficient in front of $x_{1t}^2$). The column with FGLS contains the numerical value of the MSE and the MSEs of all other estimators are expressed relative to this benchmark. Values above 1 indicate a better performance of FM-GLS. We make the following observations:
\begin{enumerate}[(a)]
 \item The FM-GLS estimator generally has the lowest MSE among all feasible estimators. These efficiency gains are small at low levels of endogeneity and serial correlation, but become sizeable at higher levels. Moreover, the Monte Carlo outcomes for the infeasible estimators indicate that these gains remain when the estimators are informed about the true endogeneity and serial correlation properties. It is thus the GLS weighting of the data that improves estimation accuracy.
 \item There is one particular instance in Table \ref{tab:efficiency_varma} in which the performance of the FM-GLS estimator has a high MSE, namely the case of high persistency $\big(\underline{\lambda},\bar{\lambda}\big)=(0.8, 0.95)$, high endogeneity $\theta=0.5$, and small sample size $T=100$. This is caused by an inaccurate BIAM estimator resulting from the combination of a small sample size, high endogeneity, and high persistency. The problem disappears when $T$ increases.
 \end{enumerate}

The subsequent set of simulations evolves around hypothesis testing, see Table \ref{table:sizeQSUCPR} and Figures \ref{fig:size_sc}-\ref{fig:power_jointtest_n5}. The errors are simulated using Setting A and we use the following Wald-type test statistics: the Wald-SOLS and Wald-SUR tests as developed in Proposition 2 in \cite{wagnergrabarczykhong2019}, and the Wald-FGLS test from Theorem \ref{thm:test}. We consider: (\textit{i}) the single equation test $H_0:\beta_{1,4}=-0.3$ against the two-sided alternative $H_1: \beta_{1,4}\neq -0.3$, and (\textit{ii}) the joint test $H_0: \beta_{1,4}=\beta_{2,4}=\ldots= \beta_{n,4}=-0.3$ against the alternative which rejects when at least one coefficient is unequal to $-0.3$. Some general remarks regarding size and size-corrected power are as follows.
\begin{enumerate}[(a)]
\setcounter{enumi}{2}
 \item The Wald tests are typically oversized but the three tests react differently to changes in $\rho$. Increases in $\rho$ result in an increasing size for the SOLS and SUR version of the Wald test, whereas increases in $\rho$ lead to size decreases for the GLS type of Wald test. Overall, the GLS test provides better size control.
 \item In Figures \ref{fig:size_sc} and \ref{fig:size_endo}, we vary the serial correlation parameter $\rho_1$ and the endogeneity parameter $\rho_2$ separately. Overall, variation in $\rho_1$ has a larger influence on size with the SUR test being most sensitive and the GLS test being least sensitive.
 \item For all three Wald-type of tests, the size of the tests improves with sample size $T$.
 \item The ordering in terms of size-corrected power is the same throughout Figures \ref{fig:power_singletest_n3}-\ref{fig:power_jointtest_n5}. That is, size-corrected power is lowest for Wald-SOLS, increases for Wald-SUR, and is highest for the Wald-FGLS test.
 \end{enumerate}
 
 The simulation results for the KPSS-type of cointegration tests can be found in Table \ref{tab:ct_tests}. The general conclusions are as follows.
 \begin{enumerate}[(a)]
\setcounter{enumi}{6}
 \item The empirical sizes of the $K^{SOLS}$ and $K^{SUR}$ tests are similar. We see: very conservative results for low serial correlation, decent size for medium serial correlation, and strongly oversized tests at high levels of serial correlation. These findings are completely in line with the simulation results that are reported in table 3 of \cite{choisaikkonen2010}. The same behaviour is observed for the $K^{BIAM}$ test but the deviations from the 5\% level are less extreme.
 \item The power of the $K^{FOLS}$, $K^{SUR}$, and $K^{BIAM}$ tests behaves as expected: (1) power always increases with sample size, and (2) power increases when more unit roots, more misspecified equations, or more spurious relationships are incorporated in the DGP. The $K^{BIAM}$ test has the lowest power among the three tests. This is caused by the fact that the filter can nearly difference the data and hence make it appear more stationary.
\end{enumerate}
 
 \section{Empirical Application}\label{sec:empappl}
The Environmental Kuznets Curve (EKC) conjectures an inverted U-shaped relation between environmental degradation and income per capita. That is, there is an initial decline in environmental quality with increasing economic activity, but beyond a certain turning point (caused by e.g. industrial transformation and increasing environmental awareness), economic growth goes hand in hand with environmental improvement. A more detailed description and historical overview of the EKC can be found in \cite{stern2004} and \cite{stern2017}, respectively. The implications of further economic growth on pollution, e.g. the emission of greenhouse-gases, are also key in understanding the future of global warming (\cite{nordhaus2013}).

We builds upon and compare with \cite{wagnergrabarczykhong2019}. That is, we look at carbon dioxide $(\text{CO}_2)$ emissions and GDP as proxies for environmental pollution and economic development (both per capita and in logarithms), respectively. The data is collected from the Maddison Project Database (MPD) and the homepage of the Carbon Dioxide Information Analysis Center (CDIAC).\footnote{The Maddison Project Database, \cite{madison2018}, contains the data on population size and real GDP. The data on $\text{CO}_2$ originates from \cite{cdiac2017}. We follow the official guidelines and multiply by $3.667$ and $10^6$ to convert the reported fossil-fuel emissions into total carbon dioxide emissions.} As in \cite{wagnergrabarczykhong2019}, we consider Austria (AT), Belgium (BE), Finland (FI), the Netherlands (NL), Switzerland (CH) and the United Kingdom (UK). Our yearly data spans the period from 1870 to 2014. We refer to the latter paper for a discussion of the stationarity properties of all series as well as the motivation for this particular set of countries. Overall, the dataset consist of $n=6$ countries with $T=145$ time series observations each. Such a panel with small $n$ and large $T$ is ideally suited for our FMGLS approach since the multivariate banded inverse autocovariance matrix remains computable.

We estimate the quadratic model specification:
\begin{equation} 
 e_{it}^{}=\beta_{i,1}^{}+\beta_{i,2}^{} t +\beta_{i,3}^{}g_{it}^{}+\beta_{i,4}^{} g_{it}^2+u_{it}^{},\quad i=1,2,\ldots,6,\quad t=1,2,\ldots,145,
\label{eq:ekc_quadraticmodel}
\end{equation}
where $e_{it}$ and $g_{it}$ are $\text{CO}_2$ emissions and GDP, respectively. As the first step in our analysis we employ the multivariate stationarity tests of Section \ref{sec:cointtest} to check this model specification (Table \ref{tab:multivariateKPSS}). All three tests reject the null of cointegration at a 5\% level signalling inappropriateness of the quadratic formulation. Figure \ref{fig:KPSSresiduals} shows the residuals on which these tests are based. What stands out in these graphs is the erratic behaviour of the series around the two world wars. Based on this fact, and to be able to compare to \cite{wagnergrabarczykhong2019}, we will continue the analysis using model \eqref{eq:ekc_quadraticmodel} and the given collection of countries. Before doing so, it will be worthwhile to discuss the time series properties of these residuals.

We consider the series $\{\hat \vu_{t,FGLS}\}$ in the remainder of this section but the other residuals series will provide qualitatively similar outcomes. When fitting the VAR($p$) models with $1\leq p \leq 8$ to these residuals, the BIC information criterion selects a lag order of $p=1$. The absolute eigenvalues of the estimated coefficient matrix are $(0.55,0.55,0.51,0.31,0.31,0.11)$, and the estimate for the error correlation matrix is
$$
 \begin{blockarray}{c c c c c c c}
	& AT & BE & FI & NL & CH & UK  \\
\begin{block}{c [c c c c c c]}
  AT	&  1		& 0.16	& 0.10	& 0.16	& 0.18	& 0.22\\
  BE	& \bullet 	& 1 		& 0.51	& 0.09	& 0.23	& 0.27\\
  FI 	& \bullet	& \bullet	& 1 		& 0.13	& 0.26	& 0.22\\
  NL	& \bullet	& \bullet	& \bullet	& 1		& 0.22	& 0.10 \\
  CH	& \bullet	& \bullet	& \bullet	& \bullet	& 1 		& 0.18\\
  UK	& \bullet	& \bullet	& \bullet	& \bullet	& \bullet	& 1 \\
\end{block}
\end{blockarray}.
$$
There is thus serial and cross-sectional correlation to be exploited by the FM-GLS estimator.

The FM-SOLS, FM-SUR and FM-GLS estimation results of Model \eqref{eq:ekc_quadraticmodel} are reported in Table \ref{tab:estimationresults}. An inspection of the coefficient estimates and their confidence intervals reveals that: (1) $\beta_{i,3}$ is positive for each country, (2) $\beta_{i,4}$ is negative for each country, and (3) all coefficients are significant at the 5\% level. All these three facts are in line with the EKC hypothesis.\footnote{This is non-surprising because \cite{wagnergrabarczykhong2019} have selected the current set of countries because they display the EKC behaviour. Also, our estimation results are slightly different from those in \cite{wagnergrabarczykhong2019} due to the additional data for 2014, possible data updates, and/or differences in the bandwidth selection of the long-run covariance matrices.} Accordingly, there exists a turning point after which further per capita economic growth reduces per capita carbon dioxide emissions. The numerical values for the turning points are heterogeneous between countries.

The widths of the confidence intervals for $\beta_{i,3}$ and $\beta_{i,4}$ display a similar pattern. From shortest to longest, the ordering is always FM-SUR, FM-SOLS, and FM-GLS and we also see how widths vary substantially between methods. To uncover the origin of these findings we conduct one final simulation study with a parameter specification that closely mimics the properties of the dataset.\footnote{The details of this simulation DGP are provided in Section \ref{empiricalillustration} of the Supplement. A visualisation of the data and the model fit are also provided there.} The average empirical coverage probabilities of asymptotic 95\% confidence intervals are 78.0\%, 66.5\% and 89.0\% for FM-SOLS, FM-SUR, and FM-GLS, respectively. In other words, the calculated confidence intervals are generally too short. By reverse engineering it turns out that the confidence intervals should be scaled by factors of 1.67, 2.15 and 1.24 to bring them back to the desired nominal level. Overall, the applied researcher should be careful when using the confidence intervals as indications for parameter uncertainty. 

\section{Conclusion} \label{sec:conclusion}
We proposed a  framework to conduct inference on cointegrating polynomial regressions. Parameters are obtained using a Fully Modified GLS estimator and we studied a cointegration test that is based on filtered residuals. Monte Carlo simulations revealed the advantages and disadvantages of these methods. The empirical researcher should realize that all estimation approaches have a tendency to underestimate parameter uncertainty and thus provide confidence intervals that are too small. The FM-GLS estimator suffers the least from this problem. Several interesting questions are left for future research. From a theoretical viewpoint, it is interesting to study the behaviour of the modified Cholesky decomposition (and BIAM) when the series under consideration is nonstationary. This would give insights into the behaviour of: (1) the FM-GLS estimator while estimating spurious regressions, and (2) the power properties of the cointegration tests. From a practical viewpoint, there seems a need to obtain more acurate standard errors of the parameter estimators.

\section*{Acknowledgements}
This paper has been presented at the 2018 CFE meeting in Pisa, the NESG 2019 conference in Amsterdam, and the $6^{th}$ RCEA Time Series Econometrics Workshop in Larnaca. We would like to thank conference participants, especially Peter Pedroni and Peter Phillips, for useful comments and suggestions. We extend our thanks to Eric Beutner, Dick van Dijk, Richard Paap, Franz Palm, and Stephan Smeekes for their valuable feedback on earlier versions of this manuscript. All remaining errors are our own.

\clearpage
\bibliographystyle{chicagoa}

\clearpage
\begin{appendices}
\section{Proofs of Main Theorems} \label{section:appendixproofs}

\begin{lemma} \label{lemma:BN}
Let $\lagpol{A}_q(L)=\mI_n - \sum_{j=1}^q \mA_j(q) L^j$ denote the lag polynomial implied by the coefficient matrices in \eqref{eq:populationMCD}. By the Beveridge-Nelson (BN) decomposition, we also have $\lagpol{A}_q(L) = \lagpol{A}_q(1) + (1-L) \widetilde{\lagpol{A}}_q(L)$ where $\widetilde{\lagpol{A}}_q(L)= \sum_{j=1}^q \widetilde{\mA}_j(q) L^{j-1}$ with $\widetilde{\mA}_j(q) = \sum_{i=j}^q \mA_i(q)$. If Assumption \ref{assumpt1:linearproc} holds, then
\begin{enumerate}[(a)]
 \item $\lagpol{A}_q(1)=\lagpol{A}(1)+O\left(\sum_{j=q+1}^{\infty}j^{1/2}\big\|\mA_j\big\|_{\calF}\right)$
 \item There exists a $q^*>0$ such that $\sum_{j=1}^{q}\big\|\widetilde{\mA}_j(q)\big\|_{\calF}<\infty$ for all $q>q^*$
 \end{enumerate}
\end{lemma}
\begin{proof}
 \textbf{(a)} By Cauchy-Schwartz, we have $\big\|\lagpol{A}_q(1)-\lagpol{A}(1)\big\|_{\calF}\leq \sum_{j=1}^{q}\big\|\mA_j(q)-\mA_j\big\|_{\calF}+\sum_{j=q+1}^{\infty}\big\|\mA_j\big\|_{\calF}\leq \Big(q\sum_{j=1}^{q}\big\|\mA_j(q)-\mA_j\big\|_{\calF}^2 \Big)^{1/2}+\sum_{j=q+1}^{\infty}\big\|\mA_j\big\|_{\calF}$ and subsequently use Lemma \ref{L2Baxter} (see Supplement). \textbf{(b)} Using Baxter's inequality in Theorem 6.6.12 in \cite{hannandeistler2012} (also see their Remark 3) for the final inequality, we derive $\sum_{j=1}^{q}\big\|\widetilde{\mA}_j(q)\big\|_{\calF}\leq \sum_{j=1}^{q}j\big\|\mA_j(q)\big\|_{\calF}\leq \sum_{j=1}^{q}j\big\|\mA_j(q)-\mA_j\big\|_{\calF}+\sum_{j=1}^{q}j\big\|\mA_j\big\|_{\calF}\leq C\sum_{j=1}^{q}j\big\|\mA_j\big\|_{\calF}\leq C$.
\end{proof}

\begin{proof}[\textbf{Proof of Theorem \ref{thm:infeasGLS}}]
 The premultiplication by $\MChol_{\vu}^{}(q)$ applies a linear filter whereas $\SChol_{\vu}^{-1}(q)$ implies weighting. Since the behaviour of the first $q\ll T$ elements does not affect the asymptotic results, we take $\mZ_{t}=\vu_t=\mZeros$ for $t\leq 0$ and for all $t=1,2\dots$, we apply the transformations implied by $\lagpol{A}_q(L)$ and $\mS^{-1}(q)$.\footnote{The same argumentation is used in \cite{phillipspark1988}. The modification to obtain a rigorous proof is straightforward.} Consequently, we have
 \begin{equation}
\begin{aligned}
	\mG_T^{-1}\left(\widehat{\vbeta}_{GLS}-\vbeta\right) &=\left[\sum_{t=1}^{T}\left(\lagpol{A}_q(L)\mZ_{t}'\mG_T^{}\right)'\mS^{-1}(q)\left(\lagpol{A}_q(L)\mZ_{t}'\mG_T^{}\right)\right]^{-1}\\
	&\qquad\qquad\qquad\times \sum_{t=1}^{T}\left(\lagpol{A}_q(L)\mZ_{t}'\mG_T^{}\right)'\mS^{-1}(q)\left(\lagpol{A}_q(L)\vu_t\right)+o_p(1).
\label{eq:infGLSparameter_decomp}
\end{aligned}
\end{equation}
Using the BN decomposition, Lemma \ref{lemma:BN}, we first show
\begin{equation}
\begin{aligned}
 \lagpol{A}_q(L)\mZ_{t}\tran\mG_T^{}&=\lagpol{A}_q(1)\mZ_{t}\tran \mG_T^{}+\sum_{j=1}^{q}\widetilde{\mA}_j^{}(q)\Delta\mZ_{t-j+1}\tran\mG_T^{} =\lagpol{A}_q(1)\mZ_{t}\tran\mG_T^{}+O_p\big(T^{-1}\big) \\
 	&=\lagpol{A}(1)\mZ_{t}\tran\mG_T^{}+O_p\left((qT)^{-1/2}\right).
\label{eq:BN_mZ}
\end{aligned}
\end{equation}
Note that $\left\| \sum_{j=1}^{q}\widetilde{\mA}_j^{}(q)\Delta\mZ_{t-j+1}\tran\mG_T^{}  \right\| \leq \sum_{j=1}^{q}\big\|\widetilde{\mA}_j(q)\big\|~\big\|\Delta \mZ_{t-j+1}'\mG_T^{}\big\|$ and that for any $t$, 
	\begin{equation}\label{eq:DeltaZtG_bound}
	\big\|\Delta \mZ_{t}'\mG_T^{}\big\|=\max_{1\leq i\leq n}\left\|\mG_{i,T}\Delta \vz_{it}\right\|=\max_{1\leq i\leq n}\left(\left\|\mG_{\vd_i,T}\Delta \vd_{it}\right\|^2+\left\|\mG_{\vs_i,T}\Delta \vs_{it}\right\|^2\right)^{1/2}.
	\end{equation}
The vector $\mG_{\vd_i,T}\Delta \vd_{it}$ typically contains elements $T^{-(k+\frac{1}{2})}\big[t^k-(t-1)^k\big]$ where $k=0,1,\cdots,q_i$. By the inequality $(a+b)^n\leq a^n+nb(a+b)^{n-1}$, for $a,b\geq 0$, $n\in\SN$, we obtain $0\leq t^k-(t-1)^k\leq k t^{k-1}$, and thus $T^{-(k+\frac{1}{2})}\big[t^k-(t-1)^k\big]\leq q_i~T^{-3/2}\leq CT^{-3/2}$. As a result, $\left\|\mG_{\vd_i,T}\Delta \vd_{it}\right\|^2\leq CT^{-3}$. The vector $\mG_{\vs_i,T}\Delta \vs_{it}$ typically contains elements $T^{-(k+1)/2}\big(x_{it}^k-x_{it-1}^k\big)$, where $k=1,\dots,p_i$. The binomial expansion implies $x_{it}^k-x_{it-1}^k=\sum_{m=0}^{k-1}{k\choose m}x_{it-1}^mv_{it}^{k-m}=O_p\big(T^{(k-1)/2}\big)$, and thus $T^{-(k+1)/2}\big(x_{it}^k-x_{it-1}^k\big)=O_p\big(T^{-1}\big)$. It further implies that $\left\|\mG_{\vs_i,T}\Delta \vs_{it}\right\|^2=O_p\big(T^{-2}\big)$. Combining $\big\|\Delta \mZ_{t}'\mG_T^{}\big\|=O_p\big(T^{-1}\big)$ with Lemma \ref{lemma:BN}(b) we establish the second equality in \eqref{eq:BN_mZ}. Now \eqref{eq:BN_mZ} follows from Lemma \ref{lemma:BN}(a) and $\sum_{j=q+1}^{\infty}j^{1/2}\big\|\mA_j\big\|_\calF=o\big(q^{-1/2}\big)$. Using \eqref{eq:BN_mZ} and $\big\|\mS(q)-\mSigma_{\eta\eta}\big\|\rightarrow 0$ (see \eqref{eq:convergence_mSq} in the supplementary material), we obtain
\begin{equation}
\begin{aligned}
	\sum_{t=1}^{T} &\left(\lagpol{A}_q(L)\mZ_{t}'\mG_T^{}\right)'\mS^{-1}(q)\left(\lagpol{A}_q(L)\mZ_{t}'\mG_T^{}\right)
	=\sum_{t=1}^{T}\mG_T^{}\mZ_{t}^{}\mOmega_{uu}^{-1}\mZ_{t}'\mG_T^{}+O_p\big(q^{-1/2}\big) \\
	&\wto\int_{0}^{1}\mJ(r)\mOmega_{uu}^{-1}\mJ(r)'dr.
\label{eq:infGLS_parameter_inverse}
\end{aligned}
\end{equation}
We rewrite the second part in \eqref{eq:infGLSparameter_decomp} as
\begin{equation}
\begin{aligned}
	\sum_{t=1}^{T}\left(\lagpol{A}_q(L)\mZ_{t}'\mG_T^{}\right)'\mS^{-1}(q)&\left(\lagpol{A}_q(L)\vu_t\right) =\sum_{t=1}^{T}\left(\lagpol{A}_q(L)\mZ_{t}'\mG_T^{}\right)'\mS^{-1}(q)\veta_t\\
	&+\sum_{t=1}^{T}\left(\lagpol{A}_q(L)\mZ_{t}'\mG_T^{}\right)'\mS^{-1}(q)\left(\lagpol{A}_q(L)\vu_t-\veta_t\right)=:\Rmnum{1}+\Rmnum{2},
\label{eq:2ndpart_mainterms}
\end{aligned}
\end{equation}
and we will repeatedly use the identity
\begin{equation}
	\sum_t\mJ_t\mD\ve_t=\sum_t\left(\ve_t'\otimes\mJ_t\right)\vec(\mD)=\left[\sum_t\mJ_te_{1t},\sum_t\mJ_t e_{2t},\cdots,\sum_t\mJ_te_{nt}\right]\vec(\mD)
\label{eq:vec_identity}
\end{equation}
	for any matrices $\mJ_t\in\SR^{d\times n}$, $\mD\in\SR^{n\times n}$ and $\ve_t=[e_{1t},\dots,e_{nt}]'\in \SR^{n\times1}$. Using this identity, we have
\begin{equation}
	\Rmnum{1}=\left[\sum_{t=1}^{T}\left(\lagpol{A}_q(L)\mZ_{t}'\mG_T^{}\right)'\eta_{1t},\cdots,\sum_{t=1}^{T}\left(\lagpol{A}_q(L)\mZ_{t}'\mG_T^{}\right)'\eta_{nt}\right]\vec\left(\mSigma_{\eta\eta}^{-1}+o(1)\right),
\label{eq:decomp_termI}
\end{equation}
where $\eta_{it}$ is the $i^{th}$ entry of $\veta_t$. For $1\leq i\leq n$, we subsequently use the BN decomposition to obtain:
\begin{equation}
\begin{aligned}
	\sum_{t=1}^{T}&\left(\lagpol{A}_q(L)\mZ_{t}'\mG_T^{}\right)'\eta_{it}\\
	&=\sum_{t=1}^{T}\mG_T\mZ_{t}\eta_{it}\lagpol{A}_q(1)'+\sum_{t=1}^{T}\mG_T\Delta\mZ_{t}\eta_{it}\widetilde{\mA}_1(q)'+\sum_{j=2}^{q}\sum_{t=1}^{T}\mG_T\Delta\mZ_{t-j+1}\eta_{it}\widetilde{\mA}_j(q)'.
\label{eq:decomp_termI_maincomponent}
\end{aligned}
\end{equation}
By definition, we have  $\sum_{t=1}^{T}\mG_T\mZ_{t}\eta_{it}=\diag\left[\sum_{t=1}^{T}\mG_{1,T}\vz_{1t},\dots,\sum_{t=1}^{T}\mG_{n,T}\vz_{nt}\right]\eta_{it}$, where the limiting distribution of each block follows from Proposition 1 of \cite{wagnerhong2016}. More specifically, the $k^{th}$ block will converge to a stochastic integral and a second order bias term which is proportional to $\mSigma_{\epsilon_k \eta_i}:=\E\left(\epsilon_{kt}\eta_{it}\right)$ (the $(k,i)^{th}$ element of $\mSigma_{\epsilon\eta}$), $1\leq k,i\leq n$, and thus
\begin{equation}
	\sum_{t=1}^{T}\mG_T\mZ_{t}\eta_{it}\lagpol{A}_q(1)'\wto\left(\int_{0}^{1}\mJ(r)d\brown_{\eta_i}(r)+\mB_i\right)\lagpol{A}(1)',
\label{eq:decomp_termI_maincomponent1}
\end{equation}
	where $\mB_i:=\diag\left[\mSigma_{\epsilon_1\eta_i}~\vb_1,\cdots,\mSigma_{\epsilon_n \eta_i}~\vb_n\right]$. 
\end{proof}
As $\sum_{t=1}^{T}\mG_T\Delta\mZ_{t}\eta_{it}\widetilde{\mA}_1(q)'=\diag\left[\sum_{t=1}^{T}\mG_{1,T}\Delta\vz_{1t},\dots,\sum_{t=1}^{T}\mG_{n,T}\Delta\vz_{nt}\right]\eta_{it}\left(\sum_{j=1}^{\infty}\mA_j+o(1)\right)'$, we again consider the limiting distributions block-wise. Every element in the blocks will rely on one of the following three results. (1) As derived below \eqref{eq:DeltaZtG_bound}, we have $\big|T^{-(j+\frac{1}{2})}\sum_{t=1}^{T}\big[t^j-(t-1)^j\big]\eta_{it}\big|\leq CT^{-3/2}\sum_{t=1}^{T}|\eta_{it}|=o_p(1)$. (2) By Assumption \ref{assumpt1:linearproc}, for any $1\leq k,i\leq n$, $v_{kt}$ and $\eta_{it}$ are Near Epoch Dependent in $L_4$-norm on $\left\{[\veta_t',\vepsi_t']'\right\}_{t\in\SZ}$ of size $-1$ and arbitrary size, respectively. A small variation on Theorem 17.9 from \cite{davidson1994} shows that $\{v_{it}\eta_{it}\}$ are $L_2$-NED of size $-1$. The i.i.d. assumption on $\left\{[\veta_t',\vepsi_t']'\right\}$ allows for a LLN for the sequence $\{v_{it}\eta_{it}\}$, see e.g. Theorem 20.21 of \cite{davidson1994}, implying  $T^{-1}\sum_{t=1}^{T}\Delta x_{kt}\eta_{it}=T^{-1}\sum_{t=1}^{T}v_{kt}\eta_{it} \pto \mSigma_{\epsilon_k \eta_i}$. (3) $T^{-(j+1)/2}\sum_{t=1}^{T}\Delta x_{kt}^j\eta_{it}^{}\wto j\mSigma_{\epsilon_k \eta_i}\int_{0}^{1}\brown_{v_k}^{j-1}(r)dr$, where $j\geq 2$. The specific reason is as follows. By the binomial expansion (below \eqref{eq:DeltaZtG_bound}), we have
	\begin{align*}
	T^{-(j+1)/2}\sum_{t=1}^{T}\Delta x_{kt}^j\eta_{it}&=T^{-(j+1)/2}\sum_{m=0}^{j-1}{j\choose m}\sum_{t=1}^{T}x_{kt-1}^mv_{kt}^{j-m}\eta_{it}^{}=jT^{-(j+1)/2}\sum_{t=1}^{T}x_{kt-1}^{j-1}v_{kt}^{}\eta_{it}^{}+o_p(1)\\
	&=j\mSigma_{\epsilon_k \eta_i}\frac{1}{T}\sum_{t=1}^{T-1}\left(\frac{x_{kt}}{\sqrt{T}}\right)^{j-1}+\frac{j}{\sqrt{T}}\frac{1}{\sqrt{T}}\sum_{t=1}^{T-1}\left(\frac{x_{kt}}{\sqrt{T}}\right)^{j-1}\left(v_{kt}\eta_{it}-\mSigma_{\epsilon_k \eta_i}\right)+o_p(1)\\
	&\wto j\mSigma_{\epsilon_k \eta_i}\int_{0}^{1}\brown_{v_k}^{j-1}(r)dr,
	\end{align*}
	where $\frac{1}{\sqrt{T}}\sum_{t=1}^{T-1}\left(\frac{x_{kt}}{\sqrt{T}}\right)^{j-1}\left(v_{kt}\eta_{it}-\mSigma_{\epsilon_k \eta_i}\right)=O_p(1)$. To see this, we refer to \cite{dejong2002}. The moment and NED conditions in his Assumption 1 are satisfied. Moreover, since $F(x)=x^{j-1}$ is homogeneous of degree $j-1$, his Assumption 2 holds as well. The desired result now follows from Theorem 1 in \cite{dejong2002}. Combining these results, we obtain
	\begin{equation}\label{eq:decomp_termI_maincomponent2}
	\sum_{t=1}^{T}\mG_T\Delta\mZ_{t}\eta_{it}\widetilde{\mA}_1(q)'\wto \mB_i\left(\sum_{j=1}^{\infty}\mA_j\right)'.
	\end{equation}
Finally, the last term in \eqref{eq:decomp_termI_maincomponent} is bounded by $\sum_{j=2}^{q}\big\|\mA_j(q)\big\|~\big\|\sum_{t=1}^{T}\mG_T\Delta\mZ_{t-j+1}\eta_{it}\big\|$. Using similar arguments above, we conclude that $\sum_{t=1}^{T}\mG_T\Delta\mZ_{t-j+1}\eta_{it}=o_p(1)$ (lags of $\Delta\mZ_t$ will lead to $\E\big(v_{kt-j}\eta_{it}\big)=0$ for any $j>0$). Hence, $\sum_{j=2}^{q}\sum_{t=1}^{T}\mG_T\Delta\mZ_{t-j+1}\eta_{it}\widetilde{\mA}_j(q)'=o_p(1)$. Combining \eqref{eq:decomp_termI_maincomponent}, \eqref{eq:decomp_termI_maincomponent1} and \eqref{eq:decomp_termI_maincomponent2}, we have
	\begin{equation}\label{eq:decomp_termI_result}
	\sum_{t=1}^{T}\left(\lagpol{A}_q(L)\mZ_{t}'\mG_T^{}\right)'\eta_{it}\wto\int_{0}^{1}\mJ(r)d\brown_{\eta_i}(r)\lagpol{A}(1)'+\mB_i.
	\end{equation}
Note that $\vec\Big(\mSigma_{\eta\eta}^{-1}\Big)=\left[\begin{smallmatrix}
	\col_1(\mSigma_{\eta\eta}^{-1})\\
	\vdots\\
	\col_n(\mSigma_{\eta\eta}^{-1})
	\end{smallmatrix}\right]$. Inserting \eqref{eq:decomp_termI_result} into \eqref{eq:decomp_termI}, we eventually have
	\begin{align}
	\Rmnum{1}&\wto\int_{0}^{1}\left(d\brown_{\eta}(r)'\otimes \mJ(r)\mA(1)'\right)\vec\left(\mSigma_{\eta\eta}^{-1}\right)+\left[\mB_1,\cdots,\mB_n\right]\vec\left(\mSigma_{\eta\eta}^{-1}\right)\nonumber\\
	&=\int_{0}^{1}\mJ(r)\mA(1)'\mSigma_{\eta\eta}^{-1}d\brown_{\eta}^{}(r)+\sum_{i=1}^{n}\mB_i\col_i\left(\mSigma_{\eta\eta}^{-1}\right)=\int_{0}^{1}\mJ(r)\mOmega_{uu}^{-1}d\brown_{u}^{}(r)+\vbias_{\epsilon\eta}\label{eq:1st_termfinal}
	\end{align}
where the symmetry property of $\mSigma_{\eta\eta}^{-1}$ is used in the final step.

	Now we consider the term $\Rmnum{2}$ in \eqref{eq:2ndpart_mainterms}. If we define $\vu_t^*=\big[u_{1t}^*,\dots,u_{nt}^*\big]:=\lagpol{A}_q(L)\vu_t-\veta_t=-\sum_{j=1}^{\infty}\big(\mA_j(q)-\mA_j\big)\vu_{t-j}$, where $\mA_j(q)=\mZeros$ for $j>q$, and then apply \eqref{eq:vec_identity}, we have $\Rmnum{2}=\Big[\sum_{t=1}^{T}\big(\lagpol{A}_q(L)\mZ_{t}'\mG_T^{}\big)'u_{1t}^*,\dots,\sum_{t=1}^{T}\big(\lagpol{A}_q(L)\mZ_{t}'\mG_T^{}\big)'u_{nt}^*\Big]\vec\Big(\mSigma_{\eta\eta}^{-1}+o(1)\Big)$. For any block $i=1,\dots,n$, by the BN decomposition \eqref{eq:BN_mZ}, we have	
	\begin{align*}
	\left\|\sum_{t=1}^{T}\left(\lagpol{A}_q(L)\mZ_{t}'\mG_T^{}\right)'u_{it}^*\right\|&=\left\|\sum_{t=1}^{T}\mG_T^{}\mZ_{t}^{}u_{it}^*\lagpol{A}_q(1)'+O_p\big(T^{-1}\big)\sum_{t=1}^{T}u_{it}^*\right\|\\
	&\hspace{-0.8cm}\leq C\max_{1\leq k\leq n}\left\|\sum_{t=1}^{T}\mG_{k,T}^{}\vz_{kt}^{}u_{it}^*\right\|+O_p\big(T^{-1}\big)\sum_{j=1}^{\infty}\left\|\row_i\left(\mA_j(q)-\mA_j\right)\right\|\left\|\sum_{t=1}^{T}\vu_{t-j}\right\|\\
	&\hspace{-0.8cm}=O_p\left(\sum_{j=q+1}^{\infty}j^{1/2}\big\|\mA_j\big\|_{\calF}\right)=o_p(1).
	\end{align*}
	It implies $\Rmnum{2}=o_p(1)$. The theorem now follows from \eqref{eq:infGLSparameter_decomp}, \eqref{eq:infGLS_parameter_inverse}, and \eqref{eq:1st_termfinal}.	

\begin{proof}[\textbf{Proof of Theorem \ref{thm:consistentDECOMP}}]
We start wih the estimation error $\widehat{\mSigma_{\vu}^{-1}}(q)-\mSigma_{\vu}^{-1}(q)$. Repeated addition and subtraction yields
\begin{equation}
\begin{aligned}
 	\Big\|\widehat{\mSigma_{\vu}^{-1}}&(q)-\mSigma_{\vu}^{-1}(q)\Big\|
\leq \left\| \widehat{\MChol}_{\vu}(q)- \MChol_{\vu}(q)\right\|\,\left\|\widehat{\SChol}_{\vu}^{-1}(q)\right\|\,\left\|\widehat{\MChol}_{\vu}(q)\right\|\\
 &+\left\|\MChol_{\vu}(q)\right\|\,\left\|\widehat{\SChol}_{\vu}^{-1}(q)-\SChol_{\vu}^{-1}(q)\right\|\,\left\|\widehat{\MChol}_{\vu}(q)\right\|+\left\|\MChol_{\vu}(q)\right\|\,\left\|\SChol_{\vu}^{-1}(q)\right\|\,\left\|\widehat{\MChol}_{\vu}(q)-\MChol_{\vu}(q)\right\|.
\end{aligned}
\label{eq:decom_partI}
\end{equation}
We will only consider the terms $\big\|\widehat{\MChol}_{\vu}(q)-\MChol_{\vu}(q)\big\|$ and $\big\|\widehat{\SChol}_{\vu}^{-1}(q)-\SChol_{\vu}^{-1}(q)\big\|$. It is not hard to derive that the remaining terms are bounded in probability. Define $\bm{\calG}=\widehat{\MChol}_{\vu}(q)-\MChol_{\vu}(q)$ and denote its ($n\times n$) subblocks by $\bm{\calG}_{ij}$, $1\leq i,j \leq T$. This matrix $\bm{\calG}$ is banded in such a way that there are at most $2q-1$ nonzero block in the block-columns of $\bm{\calG}\bm{\calG}\tran$. Using this observation and various norm properties, we find
\begin{equation*}
\begin{aligned}
	\|\bm{\calG}\|^2
	&\leq \|\bm{\calG}\bm{\calG}'\|_1^{}
	\leq \max_{1\leq j\leq T}\sum_{i=1}^{T}\left\|\big(\bm{\calG}\bm{\calG}'\big)_{ij}\right\|_1
	\leq n(2q-1)\max_{1\leq j\leq T}\sum_{t=1}^{T}\|\bm{\calG}_{jt}\|_1^2 \\
	&= n(2q-1)\max_{1\leq\ell\leq q}\sum_{s=1}^{\ell}\big\|\widehat{\mA}_s(\ell)-\mA_s(\ell)\big\|_1^2
	\leq n^3(2q-1)\max_{1\leq\ell\leq q}\big\|\widehat{\mA}(\ell)-\mA(\ell)\big\|^2 =  O_p\left(\frac{q^3}{T}\right),
\end{aligned}
\end{equation*}
where the final step follows from Lemma \ref{lem:max_bound}.\footnote{More specifically, for any matrix $\mQ$ we have $\| \mQ \|^2\leq \|\mQ\mQ\tran\|_1$. Moreover, if $\mQ$ is an $(n\times n)$ matrix, then also $\|\mQ\|_1\leq \sqrt{n} \|\mQ\|_\calF\leq n \|\mQ\|$.} We conclude that $\big\|\widehat{\MChol}_{\vu}(q)-\MChol_{\vu}(q)\big\|=O_p\left( \sqrt{q^3/T} \right)$.
The difference $\widehat{\SChol}_{\vu}^{-1}(q)-\SChol_{\vu}^{-1}(q)$ forms a symmetric and block diagonal matrix, hence $\big\|\widehat{\SChol}_{\vu}(q)-\SChol_{\vu}(q)\big\|=\max\left\{\big\|\widehat{\mS}(0)-\mS(0)\big\|,\max_{1\leq\ell\leq q}\big\|\widehat{\mS}(\ell)-\mS(\ell)\big\|\right\}$. By Assumption \ref{assump3:residuals},
\begin{equation*}
\begin{aligned}
	\Big\|\widehat{\mS}(0)-\mS(0) \Big\| &=
	\left\| \frac{1}{T} \sum_{t=1}^T \widehat{\vu}_t^{} \widehat{\vu}_t'-\E(\vu_t^{}\vu_t') \right\|\leq \frac{1}{T}\big\|\widehat{\vu}-\vu \big\|^2+2 \left\|\frac{1}{T} \sum_{t=1}^T (\widehat{\vu}_t^{}-\vu_t^{}) \vu_t' \right\|\\
	&+ \left\| \frac{1}{T}\sum_{t=1}^T \vu_t^{}\vu_t' -\E(\vu_t^{}\vu_t') \right\|=O_p\big(T^{-1/2}\big).
\end{aligned}
\end{equation*}
Applying Lemma \ref{lem:max_bound}, we see $\big\|\widehat{\SChol}_{\vu}^{-1}(q)-\SChol_{\vu}^{-1}(q)\big\|\leq \big\|\widehat{\SChol}_{\vu}(q)-\SChol_{\vu}(q)\big\|~\big\|\widehat{\SChol}_{\vu}^{-1}(q)\big\|~\big\|\SChol_{\vu}^{-1}(q)\big\|= O_p\left(q/\sqrt{T}\right)$. Overall, recalling \eqref{eq:decom_partI}, a bound on the estimation eror is $\left\|\widehat{\mSigma_{\vu}^{-1}}(q)-\mSigma_{\vu}^{-1}(q)\right\|=O_p\left( \sqrt{q^3/T} \right)$. 

The bound on the truncation error, $\left\|\mSigma_{\vu}^{-1}(q)-\mSigma_{\vu}^{-1}\right\|\leq C\frac{1}{\sqrt{q}}\sum_{s=q+1}^{\infty}s\left\|\mA_{s}\right\|_{\calF}$, follows from a straightforward generalization of the results in Lemma 2 of \cite{chengingyu2015} and Propositions 2.1-2.2 of \cite{ingchiouguo2016}. See Lemma \ref{lem:biam_consistency_part2} in the supplementary material for details.
\end{proof}

\begin{proof}[\textbf{Proof of Theorem \ref{thm:fullymodifiedOLSGLS}}]
\textbf{(a)}-\textbf{(b)} See the proof of Proposition 1 in  \cite{wagnergrabarczykhong2019}. \textbf{(c)} Since the residuals $\{\widehat{\vu}_t\}$ are obtained by first stage OLS, we get $\|\widehat{\vu}-\vu\|^{2}\leq \|\mG_T^{-1}\big(\widehat{\vbeta}_{OLS}-\vbeta\big)\|^2~\|\mG_T\mZ'\mZ\mG_T\|=O_p(1)$. Assumption \ref{assump3:residuals} is thus satisfied and we can rely on the results in Theorem \ref{thm:consistentDECOMP}. From \eqref{eq:FMGLSestimator}, the definition of the FM-GLS estimator, we have
\begin{equation*}
	\mG_T^{-1}\left(\widehat{\vbeta}_{FGLS}^+-\vbeta\right)=\left(\mG_T\mZ'\widehat{\mSigma_{\vu}^{-1}}(q)\mZ\mG_T\right)^{-1}\left[\mG_T\mZ'\widehat{\mSigma_{\vu}^{-1}}(q)\vu-\mG_T\mZ'\Big(\mI_T\otimes \widehat{\mOmega}_{uu}^{-1}\widehat{\mOmega}_{uv}^{}\widehat{\mOmega}_{vv}^{-1}\Big)\vv-\mG_T\widehat{\vbias}^{+}\right].
\end{equation*}
Given Theorem \ref{thm:consistentDECOMP}, we have $\mG_T\mZ'\widehat{\mSigma_{\vu}^{-1}}(q)\mZ\mG_T=\mG_T\mZ'\mSigma_{\vu}^{-1}(q)\mZ\mG_T+o_p(1)$ and it converges weakly to the expression in \eqref{eq:infGLS_parameter_inverse}.

To continue, we define $\widehat{\lagpol{A}}_q(L)=\mI_n-\sum_{j=1}^{q}\widehat{\mA}_j(q)L^j$ and its BN decomposition $\widehat{\lagpol{A}}_q(L)=\widehat{\lagpol{A}}_q(1)+(1-L) \lagpol{A}_q^*(L)$ through $\lagpol{A}_q^*(L)=\sum_{j=1}^{q}\mA_j^*(q)L^{j-1}$ with $\mA_j^*(q)=\sum_{i=j}^q\widehat{\mA}_i(q)$. $\widehat{\lagpol{A}}_q(1)=\lagpol{A}_q(1)+o_p(1)$ and $\sum_{j=1}^{q}\big\|\mA_j^*(q)\big\|_{\calF}\leq \sum_{j=1}^{q}\big\|\widetilde{\mA}_j(q)\big\|_{\calF}+o_p(q)$ are obtained from the following two results: (1) $\big\|\widehat{\lagpol{A}}_q(1)-\lagpol{A}_q(1)\big\|_\calF\leq  \sum_{j=1}^q\big\|\widehat{\mA}_j(q)-\mA_j(q)\big\|_{\calF}\leq C\sqrt{q}\big\|\widehat{\mA}(q)-\mA(q)\big\|=O_p\Big(\frac{q^{3/2}}{T^{1/2}}\Big)=o_p(1)$, where the last step follows from Lemma \ref{lem:max_bound}, and (2) $\sum_{j=1}^{q}\big\|\mA_j^*(q)-\widetilde{\mA}_j(q)\big\|_{\calF}\leq q\sum_{j=1}^{q}\big\|\widehat{\mA}_j(q)-\mA_j(q)\big\|_{\calF}=o_p(q)$. Using the BN decomposition of $\widehat{\lagpol{A}}_q(L)$ and similar steps as those below \eqref{eq:2ndpart_mainterms}, we have
	\begin{equation*}
	\mG_T\mZ'\widehat{\mSigma_{\vu}^{-1}}(q)\vu=\sum_{t=1}^{T}\left(\widehat{\lagpol{A}}_q(L)\mZ_{t}'\mG_T^{}\right)'\widehat{\mS}^{-1}(q)\left(\widehat{\lagpol{A}}_q(L)\vu_t\right)+o_p(1)\wto\int_{0}^{1}\mJ(r)\mOmega_{uu}^{-1}d\brown_{u}^{}(r)+\vbias_{\epsilon\eta},
	\end{equation*}
	where $\widehat{\mS}(q)=\mS(q)+o_p(1)$ given in Lemma \ref{lem:max_bound}. Using the identity \eqref{eq:vec_identity}, it is not hard to deduce
	\begin{equation*}
	\mG_T\mZ'\Big(\mI_T\otimes \widehat{\mOmega}_{uu}^{-1}\widehat{\mOmega}_{uv}^{}\widehat{\mOmega}_{vv}^{-1}\Big)\vv\wto \int_{0}^{1}\mJ(r)\mOmega_{uu}^{-1}\mOmega_{uv}^{}\mOmega_{vv}^{-1}d\brown_{v}^{}(r)+\underbrace{\left[\begin{smallmatrix}
	\row_1\big(\mDelta_{vv}\big)\col_1\big(\mOmega_{vv}^{-1}\mOmega_{vu}^{}\mOmega_{uu}^{-1}\big)\vb_1\\
	\vdots\\
	\row_n\big(\mDelta_{vv}\big)\col_n\big(\mOmega_{vv}^{-1}\mOmega_{vu}^{}\mOmega_{uu}^{-1}\big)\vb_n
	\end{smallmatrix}\right]}_{\vbias_{vu}}.
	\end{equation*}
	Combining the results above leads to:
	\begin{equation*}
	\mG_T\mZ'\widehat{\mSigma_{\vu}^{-1}}(q)\vu-\mG_T\mZ'\Big(\mI_T\otimes \widehat{\mOmega}_{uu}^{-1}\widehat{\mOmega}_{uv}^{}\widehat{\mOmega}_{vv}^{-1}\Big)\vv\wto \int_{0}^{1}\mJ(r)\mOmega_{uu}^{-1}d\brown_{u.v}^{}(r)+\vbias^+,
	\end{equation*}
	where $\vbias^+:=\vbias_{\epsilon\eta}-\vbias_{vu}$. By construction, we have $\mG_T\widehat{\vbias}^{+}\wto \vbias^+$. Altogether this implies the limiting distribution in the theorem.	
\end{proof}

\begin{proof}[\textbf{Proof of Theorem \ref{thm:test}}] 
	We first introduce the appropriate scaling into the test statistic, that is
	\begin{equation*}
	\calW=\Big(\mR\widehat{\vbeta}_{FGLS}^+-\vr\Big)'\Big(\mR\mG_T^{-1}\mR'\Big)\left[\Big(\mR\mG_T^{-1}\mR'\Big)~\widehat{\mPhi}~\Big(\mR\mG_T^{-1}\mR'\Big)\right]^{-1}\Big(\mR\mG_T^{-1}\mR'\Big)\Big(\mR\widehat{\vbeta}_{FGLS}^+-\vr\Big).
	\end{equation*}
	Since the matrices $\mG_T^{-1}$ and $\mR'\mR$ commute and $\mR\mR'=\mI_k$, we have $(\mR\mG_T^{-1}\mR\tran)(\mR\widehat{\vbeta}_{FGLS}^+-\vr)=\mR\mG_T^{-1}(\widehat{\vbeta}_{FGLS}^+-\vbeta)$ under the null hypothesis. Conditional on $\calF_v=\sigma\big(\brown_v(r),0 \leq r \leq 1 \big)$, this quantity is asymptotically normally distributed by Theorem \ref{thm:fullymodifiedOLSGLS}(c) with asymptotic covariance matrix
	\begin{equation*}
	\mR\left(\int_{0}^{1}\mJ(r)\mOmega_{uu}^{-1}\mJ(r)'dr\right)^{-1}\left(\int_{0}^{1}\mJ(r)\mOmega_{uu}^{-1}\mOmega_{u.v}^{}\mOmega_{uu}^{-1}\mJ(r)'dr\right)\left(\int_{0}^{1}\mJ(r)\mOmega_{uu}^{-1}\mJ(r)'dr\right)^{-1}\mR'.
	\end{equation*}
The consistent estimation of all the quantities involved ensures that $(\mR\mG_T^{-1}\mR' )~\widehat{\mPhi}~(\mR\mG_T^{-1}\mR')$ has the same limit. Therefore, the Wald statistics is conditionally chi-square distributed with $k$ degrees of freedom. Since this distribution does not depend on $\calF_v$, we conclude that the unconditional distribution of $\calW$ is also $\chi_k^2$.
\end{proof}

\begin{proof}[\textbf{Proof of Theorem \ref{thm:kpss_subtest}}]
The results for $K_{j,b_T}^{SOLS}$ and $K_{j,b_T}^{SUR}$ follow from a straightforward multivariate extensions of the proof of Proposition 6 in \cite{wagnerhong2016}. For $K_{j,b_T}^{BIAM}$, we first define the population counterparts of $\vvarphi_{j,b_T}(\{\hat{\vu}_{FGLS}\})$ and $\widehat{\mSigma_{\vu}^{-1}}(q_T,b_T)$. That is, let $\vvarphi_{j,b_T}:=\big[ \vu_{j}',\sum_{s=j}^{j+1} \vu_{s}',\dots,\sum_{s=j}^{j+b_T-1}\vu_{s}'\big]'$ and let $\mSigma_{\vu}^{-1}(q_T,b_T)$ denote the subblock matrix of $\mSigma_{\vu}^{-1}(q_T)$ formed by taking the elements with row and column indices belonging to the set $\{ n(T-b_T)+1, n(T-b_T)+2 ,\ldots,nT\}$. By rearrangement, we have
	\begin{equation}\label{eq:kpss_biam_decomp}
	\begin{aligned}
	K_{j,b_T}^{BIAM} &=b_T^{-2}\vvarphi_{j,b_T}'~\mSigma_{\vu}^{-1}(q_T,b_T)~\vvarphi_{j,b_T}^{}
	+b_T^{-2}\vvarphi_{j,b_T}'\left(\widehat{\mSigma_{\vu}^{-1}}(q_T,b_T)-\mSigma_{\vu}^{-1}(q_T,b_T)\right)\vvarphi_{j,b_T}^{} \\
	& \qquad+R(q_T,j,b_T),
	\end{aligned}
	\end{equation}
	where the remainder term is bounded as $ |R(q_T,j,b_T)|\leq \big\|b_T^{-1}\big( \vvarphi_{j,b_T}(\{\hat{\vu}_{FGLS}\}) -\vvarphi_{j,b_T}\big)\big\|^2~\big\|\widehat{\mSigma_{\vu}^{-1}}(q_T,b_T)\big\|+2\big\| b_T^{-1}\big( \vvarphi_{j,b_T}(\{\hat{\vu}_{FGLS}\}) -\vvarphi_{j,b_T}\big)\big\|~\big\|\widehat{\mSigma_{\vu}^{-1}}(q_T,b_T)\big\|~\big\|b_T^{-1}\vvarphi_{j,b_T}\big\|$. Poincar\'{e}'s separation theorem (e.g. page 347-348 of \cite{abadirmagnus2005}) implies $\lambda_{min}(\mA)\leq \lambda_{min}(\mB)\leq\lambda_{max}\left(\mB\right)\leq \lambda_{max}\left(\mA\right)$ when $\mB$ is a principal submatrix of $\mA$. By this inequality and Theorem \ref{thm:consistentDECOMP}, we conclude that $\big\|\widehat{\mSigma_{\vu}^{-1}}(q_T,b_T)-\mSigma_{\vu}^{-1}(q_T,b_T)\big\|\leq \big\|\widehat{\mSigma_{\vu}^{-1}}(q_T)-\mSigma_{\vu}^{-1}(q_T)\big\|=o_p(1)$ and $\big\|\widehat{\mSigma_{\vu}^{-1}}(q_T,b_T)\big\|=O_p(1)$. Moreover, 
	\begin{equation*}
	\left\|b_T^{-1}\big(\widehat{\vvarphi}_{j,b_T}-\vvarphi_{j,b_T}\big)\right\|=\left(b_T^{-1}\sum_{t=j}^{j+b_T-1}\left\|b_T^{-1/2}\sum_{s=j}^{t}\left(\widehat{\vu}_{s,FGLS}-\vu_s\right)\right\|^2\right)^{1/2}=o_p(1)
	\end{equation*}
	where $\left\|b_T^{-1/2}\sum_{s=j}^{t}\left(\widehat{\vu}_{s,FGLS}-\vu_s\right)\right\|\leq \left\|b_T^{-1/2}\sum_{s=j}^{t}\mZ_s'\mG_T^{}\right\|~\big\|\mG_T^{-1}\big(\widehat{\vbeta}_{FGLS}^+-\vbeta\big)\big\|=o_p(1)$ by Theorem \ref{thm:fullymodifiedOLSGLS} and the assumption $b_T/T\rightarrow 0$ as $T\rightarrow\infty$. Standard weak convergence arguments imply $\big\|b_T^{-1}\vvarphi_{j,b_T}\big\|=O_p(1)$. Combining these results, we have $b_T^{-2}\vvarphi_{j,b_T}'\Big(\widehat{\mSigma_{\vu}^{-1}}(q_T,b_T)-\mSigma_{\vu}^{-1}(q_T,b_T)\Big)\vvarphi_{j,b_T}^{}=o_p(1)$ and $R(q_T,j,b_T)=o_p(1)$. By \eqref{eq:kpss_biam_decomp}, it remains to consider $b_T^{-2}\vvarphi_{j,b_T}'~\mSigma_{\vu}^{-1}(q_T,b_T)~\vvarphi_{j,b_T}^{}$.
Construct the $nT\times nb_T$ selection matrix $\mR_{j,b_T}$ such that
	\begin{equation*}
	\mR_{j,b_T}\vvarphi_{j,b_T}=\left[\vzeros',\cdots,\vzeros',\vu_{j}',\sum_{s=j}^{j+1} \vu_{s}',\dots,\sum_{s=j}^{j+b_T-1}\vu_{s}',\vzeros',\cdots,\vzeros'\right]'.
	\end{equation*}
	Then, by the MCD \eqref{eq:bimam}, we have
\begin{equation}
\begin{aligned}
	b_T^{-2}\vvarphi_{j,b_T}'~  \mSigma_{\vu}^{-1}(q_T,b_T)~\vvarphi_{j,b_T}^{} &=b_T^{-2}\vvarphi_{j,b_T}'\mR_{j,b_T}'~\mSigma_{\vu}^{-1}(q_T)~\mR_{j,b_T}^{}\vvarphi_{j,b_T}^{}\\
	&=b_T^{-2}\left(\MChol_{\vu}(q_T)\mR_{j,b_T}\vvarphi_{j,b_T}\right)'\SChol_{\vu}^{-1}(q_T)\left(\MChol_{\vu}^{}(q_T) \mR_{j,b_T} \vvarphi_{j,b_T}\right).
\end{aligned}
\end{equation}

As argued in the proof of Theorem \ref{thm:infeasGLS}, by the assumption $\frac{q_T}{b_T} \rightarrow 0$ as $T\rightarrow\infty$, we can treat the premultiplication of $\mR_{j,b_T}\vvarphi_{j,b_T}$ by $\MChol_{\vu}(q_T)$ as applying the filter $\lagpol{A}_{q_T}(L)$ block-wise. Under the same condition, $\SChol_{\vu}^{-1}(q_T)$ implies a scaling $\mS^{-1}(q_T)$. By the BN decomposition in Lemma \ref{lemma:BN}, and similarly $\lagpol{C}(L)=\lagpol{C}(1)+(1-L)\widetilde{\lagpol{C}}(L)$ with $\lagpol{C}(L)=[\lagpol{A}(L)]^{-1}$,
	\begin{equation*}
	\lagpol{A}_q(L)\sum_{s=j}^{t}\vu_{s}=\lagpol{A}_{q_T}(1)\sum_{s=j}^{t}\vu_{s}+\widetilde{\lagpol{A}}_{q_T}(L)\vu_t=\lagpol{A}_{q_T}(1)\lagpol{C}(1)\sum_{s=j}^{t}\veta_s+\widetilde{\lagpol{A}}_{q_T}(1)\widetilde{\lagpol{C}}(L)\big(\veta_t-\veta_{j-1}\big)+\widetilde{\lagpol{A}}_{q_T}(L)\vu_t.
	\end{equation*}
	For $t=j+[rb_T]-1$, a FCLT for i.i.d. sequences gives
	\begin{equation*}
	b_T^{-1/2}\mS^{-1/2}(q_T)\lagpol{A}_{q_T}(L)\sum_{s=j}^{t}\vu_{s}=\mS^{-1/2}(q_T)\lagpol{A}_{q_T}(1)\lagpol{C}(1)~b_T^{-1/2}\sum_{s=j}^{t}\veta_s+O_p(b_T^{-1/2})\wto \mW(r).
	\end{equation*}
The partial sum process $\sum_{s=j}^{t}\veta_s$ thus dominates the asymptotic distribution:
	\begin{equation*}
	b_T^{-2}\vvarphi_{j,b_T}'~\mSigma_{\vu}^{-1}(q_T,b_T)~\vvarphi_{j,b_T}^{}=b_T^{-1}\sum_{t=j}^{j+b_T-1}\left\|b_T^{-1/2}\mS^{-1/2}(q_T)\lagpol{A}_{q_T}(L)\sum_{s=j}^{t}\vu_{s}\right\|^2+o_p(1)\wto \int_{0}^{1} \left\|\mW(r)\right\|^2dr.
	\end{equation*}
An application of the continuous mapping theorem completes the proof.
\end{proof}

\newpage
\begin{table}[H]
	\centering
	\caption{Empirical MSE for the coefficient $\beta_{i,4}$ of $x_{it}^2$ with $i=1$ under error Setting A. The column labeled FGLS contains the numerical value of the MSE of feasible FM-GLS. Other MSEs are expressed relative to this benchmark. Values above 1 indicate a better performance of feasible FM-GLS.}
	\label{table:efficiencyQSUCPR}
	\resizebox{\textwidth}{!}{%
		\begin{tabular}{c r r r r r r r r r r r r}
			\toprule
			\multicolumn{1}{l}{} & \multicolumn{6}{c}{$n=3$}                             & \multicolumn{6}{c}{$n=5$}                             \\
			\midrule
			\multicolumn{1}{l}{$\rho$} &
			\multicolumn{1}{c}{SOLS} &
			\multicolumn{1}{c}{SUR} &
			\multicolumn{1}{c}{FGLS} &
			\multicolumn{1}{c}{infSOLS} &
			\multicolumn{1}{c}{infSUR} &
			\multicolumn{1}{c}{infGLS} &
			\multicolumn{1}{c}{SOLS} &
			\multicolumn{1}{c}{SUR} &
			\multicolumn{1}{c}{FGLS} &
			\multicolumn{1}{c}{infSOLS} &
			\multicolumn{1}{c}{infSUR} &
			\multicolumn{1}{c}{infGLS} \\                                                                                   
			\midrule
			\multicolumn{13}{l}{$T=100$}                                                                                                         \\
			\midrule
			0                    & 0.999  & 1.048  & 3.56E-05 & 0.937  & 0.937  & 0.937  & 0.988  & 1.047  & 3.81E-05 & 0.894  & 0.894  & 0.894  \\
			0.3                  & 1.170  & 1.077  & 5.50E-05 & 1.098  & 0.959  & 0.907  & 1.186  & 1.064  & 5.43E-05 & 1.090  & 0.899  & 0.851  \\
			0.6                  & 2.166  & 1.474  & 7.09E-05 & 2.017  & 1.180  & 0.864  & 2.260  & 1.426  & 6.85E-05 & 2.025  & 1.046  & 0.785  \\
			0.8                  & 5.247  & 2.607  & 7.51E-05 & 6.547  & 2.474  & 0.878  & 5.720  & 2.589  & 7.40E-05 & 6.591  & 1.984  & 0.676  \\
			\midrule
			\multicolumn{13}{l}{$T=200$}                                                                                                         \\
			\midrule
			0                    & 1.012  & 1.042  & 4.09E-06 & 0.961  & 0.961  & 0.961  & 1.019  & 1.069  & 4.25E-06 & 0.939  & 0.939  & 0.939  \\
			0.3                  & 1.146  & 1.043  & 6.64E-06 & 1.101  & 0.960  & 0.942  & 1.216  & 1.069  & 6.62E-06 & 1.124  & 0.937  & 0.907  \\
			0.6                  & 2.045  & 1.295  & 1.01E-05 & 1.920  & 1.101  & 0.906  & 2.222  & 1.345  & 9.37E-06 & 2.032  & 1.056  & 0.851  \\
			0.8                  & 5.206  & 2.434  & 1.29E-05 & 5.502  & 1.971  & 0.916  & 6.197  & 2.701  & 1.11E-05 & 6.106  & 1.778  & 0.813  \\
			\midrule
			\multicolumn{13}{l}{$T=500$}                                                                                                         \\
			\midrule
			0                    & 1.013  & 1.028  & 2.41E-07 & 0.988  & 0.988  & 0.988  & 1.016  & 1.039  & 2.55E-07 & 0.973  & 0.973  & 0.973  \\
			0.3                  & 1.195  & 1.054  & 4.03E-07 & 1.162  & 0.998  & 0.975  & 1.224  & 1.052  & 4.02E-07 & 1.176  & 0.978  & 0.961  \\
			0.6                  & 1.877  & 1.154  & 7.45E-07 & 1.789  & 1.054  & 0.952  & 2.142  & 1.217  & 6.49E-07 & 1.993  & 1.023  & 0.914  \\
			0.8                  & 4.158  & 1.771  & 1.26E-06 & 4.011  & 1.456  & 0.935  & 5.341  & 2.031  & 1.03E-06 & 5.038  & 1.394  & 0.880  \\
			\bottomrule
		\end{tabular}%
	}
\end{table}

\begin{table}[H]
	\centering
	\caption{Empirical MSE for the coefficient $\beta_{i,4}$ of $x_{it}^2$ with $i=1$. under error Setting B. The column labeled FGLS contains the numerical value of the MSE of feasible FM-GLS. Other MSEs are expressed relative to this benchmark. Values above 1 indicate a better performance of feasible FM-GLS.}
	\label{tab:efficiency_varma}
	\resizebox{\textwidth}{!}{%
		\begin{tabular}{cccccccccccccc}
			\toprule
			\multicolumn{1}{l}{}           & \multicolumn{1}{l}{} & \multicolumn{6}{c}{$n=3$}                        & \multicolumn{6}{c}{$n=5$}                        \\
			\midrule
			$\big(\underline{\lambda},\bar{\lambda}\big)$ & $T$ & SOLS & SUR & FGLS & infSOLS & infSUR & infGLS & SOLS & SUR & FGLS & infSOLS & infSUR & infGLS      \\
			\midrule
			\multicolumn{14}{l}{Panel A: Low endogeneity $\theta=0.3$}                                                                                                  \\
			\midrule
			\multirow{3}{*}{$(0.1, 0.5)$}  & 100                  & 1.290 & 1.270 & 9.46E-05 & 1.270 & 1.131 & 0.865 & 1.256 & 1.275 & 9.51E-05 & 1.215 & 1.052 & 0.822 \\
			& 200                  & 1.216 & 1.163 & 1.27E-05 & 1.191 & 1.077 & 0.916 & 1.220 & 1.164 & 1.22E-05 & 1.201 & 1.038 & 0.896 \\
			& 500                  & 1.137 & 1.088 & 8.58E-07 & 1.133 & 1.033 & 0.961 & 1.159 & 1.106 & 7.65E-07 & 1.151 & 0.997 & 0.936 \\
			\midrule
			\multirow{3}{*}{$(0.5, 0.8)$}  & 100                  & 1.790 & 1.825 & 3.99E-05 & 1.792 & 1.522 & 0.638 & 1.812 & 1.896 & 3.81E-05 & 1.802 & 1.431 & 0.594 \\
			& 200                  & 1.677 & 1.681 & 4.91E-06 & 1.597 & 1.408 & 0.763 & 1.682 & 1.645 & 4.73E-06 & 1.581 & 1.271 & 0.733 \\
			& 500                  & 1.467 & 1.404 & 3.34E-07 & 1.392 & 1.248 & 0.897 & 1.522 & 1.415 & 3.01E-07 & 1.432 & 1.185 & 0.864 \\
			\midrule
			\multirow{3}{*}{$(0.8, 0.95)$} & 100                  & 0.123 & 0.136 & 1.48E-04 & 0.212 & 0.145 & 0.022 & 0.001 & 0.001 & 2.67E-02 & 0.002 & 0.001 & 0.000 \\
			& 200                  & 2.257 & 2.322 & 8.59E-07 & 2.537 & 1.986 & 0.517 & 1.622 & 1.577 & 1.04E-06 & 2.218 & 1.343 & 0.366 \\
			& 500                  & 2.098 & 2.050 & 5.47E-08 & 2.021 & 1.686 & 0.703 & 2.031 & 2.014 & 4.90E-08 & 1.964 & 1.394 & 0.632 \\
			\midrule
			\multicolumn{14}{l}{Panel B: High endogeneity $\theta=0.5$}                                                                                                 \\
			\midrule
			\multirow{3}{*}{$(0.1, 0.5)$}  & 100                  & 1.327 & 1.269 & 7.48E-05 & 1.579 & 1.230 & 0.899 & 1.335 & 1.299 & 7.26E-05 & 1.825 & 1.211 & 0.817 \\
			& 200                  & 1.232 & 1.161 & 9.91E-06 & 1.394 & 1.131 & 0.949 & 1.269 & 1.225 & 9.26E-06 & 1.653 & 1.166 & 0.886 \\
			& 500                  & 1.179 & 1.091 & 6.66E-07 & 1.297 & 1.079 & 0.980 & 1.162 & 1.093 & 5.99E-07 & 1.464 & 1.071 & 0.942 \\
			\midrule
			\multirow{3}{*}{$(0.5, 0.8)$}  & 100                  & 1.880 & 1.892 & 3.08E-05 & 2.732 & 1.861 & 0.661 & 1.850 & 1.833 & 3.05E-05 & 3.170 & 1.718 & 0.575 \\
			& 200                  & 1.793 & 1.727 & 3.72E-06 & 2.132 & 1.610 & 0.798 & 1.713 & 1.639 & 3.76E-06 & 2.143 & 1.381 & 0.697 \\
			& 500                  & 1.480 & 1.370 & 2.62E-07 & 1.678 & 1.299 & 0.941 & 1.577 & 1.461 & 2.44E-07 & 1.739 & 1.256 & 0.863 \\
			\midrule
			\multirow{3}{*}{$(0.8, 0.95)$} & 100                  & 0.030 & 0.031 & 5.14E-04 & 0.172 & 0.077 & 0.005 & 0.000 & 0.000 & 3.27E+03 & 0.000 & 0.000 & 0.000 \\
			& 200                  & 2.381 & 2.514 & 6.85E-07 & 5.002 & 2.825 & 0.541 & 1.880 & 1.771 & 7.83E-07 & 5.026 & 2.030 & 0.389 \\
			& 500                  & 2.186 & 2.090 & 4.23E-08 & 2.756 & 1.850 & 0.712 & 2.184 & 2.044 & 4.03E-08 & 2.786 & 1.560 & 0.619 \\
			\bottomrule
		\end{tabular}%
	}
\end{table}

\newpage
\clearpage
\begin{table}[p]
	\centering
	\caption{Empirical size ($\%$) of the single-equation Wald tests for $H_0:\beta_{1,4}=-0.3$ and the joint Wald tests for $H_0:\beta_{1,4}=\beta_{2,4}=\cdots=\beta_{n,4}=-0.3$, where $\beta_{i,4}$ denote the coefficients in front of $x_{it}^2$. The columns labeled by Wald-SOLS and Wald-SUR display the results of the Wald-type test statistics given in Proposition 2 by \cite{wagnergrabarczykhong2019}. The Wald-FGLS test is constructed as in Theorem \ref{thm:test}.}
	\label{table:sizeQSUCPR}
	\resizebox{0.9\textwidth}{!}{%
		\begin{tabular}{crrrrrrrr}
			\toprule
			\multicolumn{1}{l}{} & \multicolumn{3}{c}{$n=3$}        & \multicolumn{3}{c}{$n=5$}        \\
			\midrule
			$\rho$               & Wald-SOLS & Wald-SUR & Wald-FGLS & Wald-SOLS & Wald-SUR & Wald-FGLS \\
			\midrule
			\multicolumn{7}{l}{Panel A:   Single-equation test}                                        \\
			\midrule
			\multicolumn{7}{l}{$T=100$}                                                                \\
			\midrule
			0                    & 11.75     & 13.44    & 9.89      & 14.30     & 17.84    & 13.44     \\
			0.3                  & 13.25     & 15.09    & 9.92      & 15.32     & 19.59    & 13.55     \\
			0.6                  & 16.13     & 19.03    & 7.54      & 18.80     & 27.65    & 11.86     \\
			0.8                  & 20.56     & 26.60    & 4.70      & 26.72     & 43.11    & 10.13     \\
			\midrule
			\multicolumn{7}{l}{$T=200$}                                                                \\
			\midrule
			0                    & 9.02      & 10.00    & 7.48      & 9.90      & 12.00    & 8.47      \\
			0.3                  & 9.96      & 10.93    & 7.24      & 11.32     & 13.62    & 8.81      \\
			0.6                  & 12.68     & 14.10    & 5.93      & 14.62     & 19.40    & 7.71      \\
			0.8                  & 15.72     & 19.50    & 2.95      & 19.58     & 30.95    & 4.90      \\
			\midrule
			\multicolumn{7}{l}{$T=500$}                                                                \\
			\midrule
			0                    & 7.02      & 7.44     & 5.89      & 7.32      & 8.47     & 6.19      \\
			0.3                  & 8.07      & 8.41     & 5.76      & 8.54      & 9.50     & 6.04      \\
			0.6                  & 9.41      & 9.68     & 4.78      & 10.34     & 12.54    & 5.94      \\
			0.8                  & 10.92     & 12.48    & 3.09      & 13.47     & 18.96    & 3.73      \\
			\midrule
			\multicolumn{7}{l}{Panel B: Joint test}                                                    \\
			\midrule
			\multicolumn{7}{l}{$T=100$}                                                                \\
			\midrule
			0                    & 17.85     & 21.68    & 14.60     & 29.13     & 39.40    & 26.65     \\
			0.3                  & 20.67     & 24.62    & 14.44     & 32.60     & 44.91    & 28.14     \\
			0.6                  & 27.13     & 33.59    & 11.40     & 45.65     & 65.44    & 27.72     \\
			0.8                  & 36.99     & 48.49    & 8.30      & 63.29     & 86.38    & 26.57     \\
			\midrule
			\multicolumn{7}{l}{$T=200$}                                                                \\
			\midrule
			0                    & 11.66     & 13.74    & 9.00      & 17.49     & 22.90    & 13.26     \\
			0.3                  & 14.11     & 16.25    & 8.77      & 21.35     & 28.23    & 14.55     \\
			0.6                  & 19.08     & 22.71    & 7.10      & 30.58     & 44.07    & 12.88     \\
			0.8                  & 26.21     & 33.92    & 4.19      & 44.82     & 69.44    & 10.60     \\
			\midrule
			\multicolumn{7}{l}{$T=500$}                                                                \\
			\midrule
			0                    & 8.70      & 9.64     & 6.60      & 10.83     & 13.38    & 7.91      \\
			0.3                  & 9.84      & 10.74    & 6.23      & 13.29     & 16.19    & 8.19      \\
			0.6                  & 12.79     & 14.09    & 5.13      & 18.69     & 25.15    & 7.67      \\
			0.8                  & 16.02     & 19.56    & 3.35      & 26.86     & 41.72    & 5.28      \\
			\bottomrule
		\end{tabular}%
	}
\end{table}

\newpage
\clearpage	
\begin{figure}[!]
	\centering
	\includegraphics[width=0.6\textwidth]{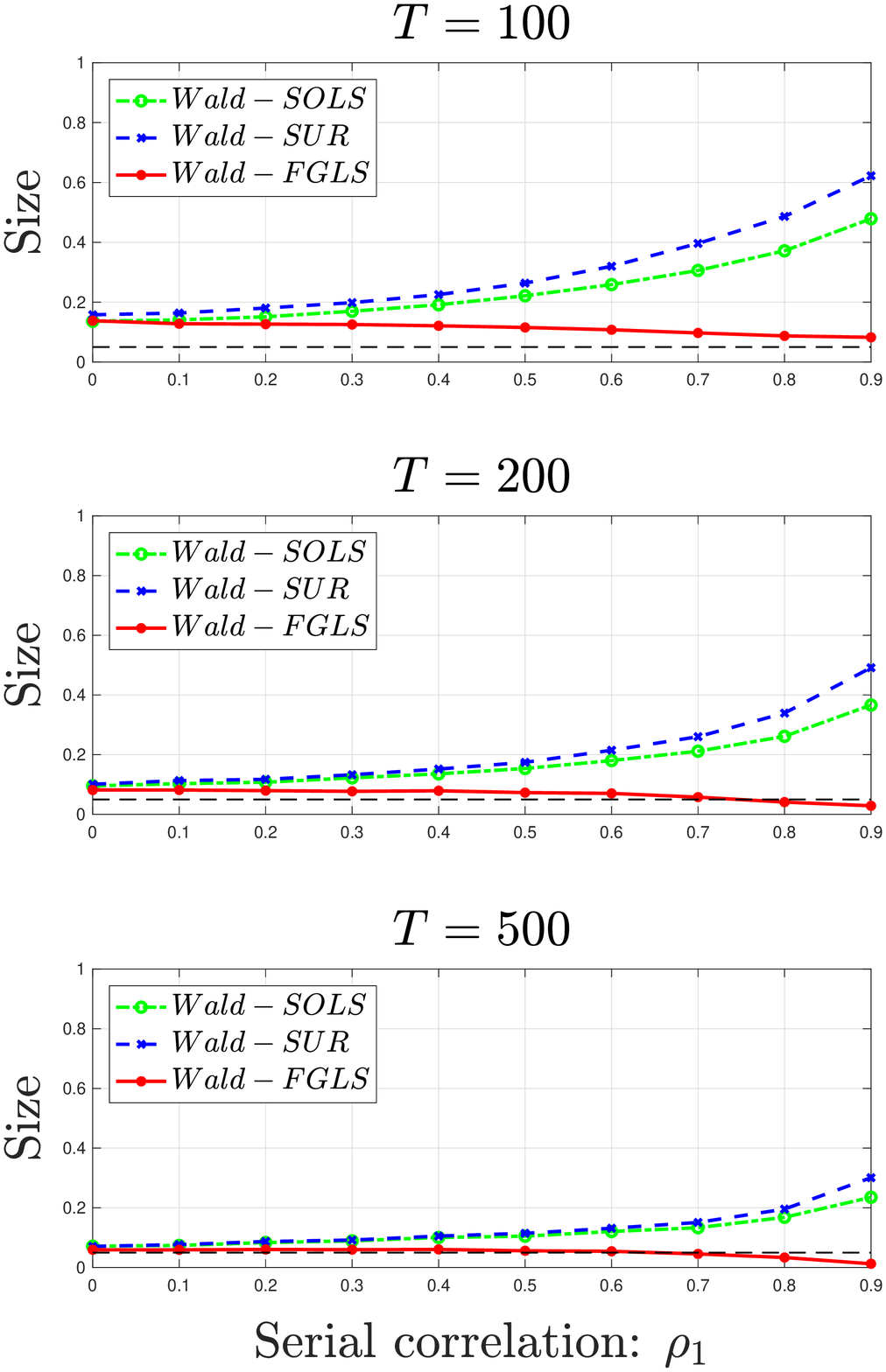}
	\caption{\footnotesize Empirical size of the joint Wald tests $H_0:\beta_{1,4}=\beta_{2,4}=\cdots=\beta_{n,4}=-0.3$, where $\beta_{i,4}$ denote the coefficients in front of $x_{it}^2$. The Wald-SOLS test (green) and Wald-SUR test (blue) are based on Proposition 2 by \cite{wagnergrabarczykhong2019}. The Wald-FGLS test (red) is found in Theorem \ref{thm:test}. We vary the serial correlation parameter $\rho_1$ from $0$ to $0.9$ while keeping $\rho_2=\rho_3=\rho_4=0.8$ fixed. The cross-sectional dimension is $n=3$.}
	\label{fig:size_sc}
\end{figure}	
	
\newpage
\clearpage	
\begin{figure}[!]
	\centering
	\includegraphics[width=0.6\textwidth]{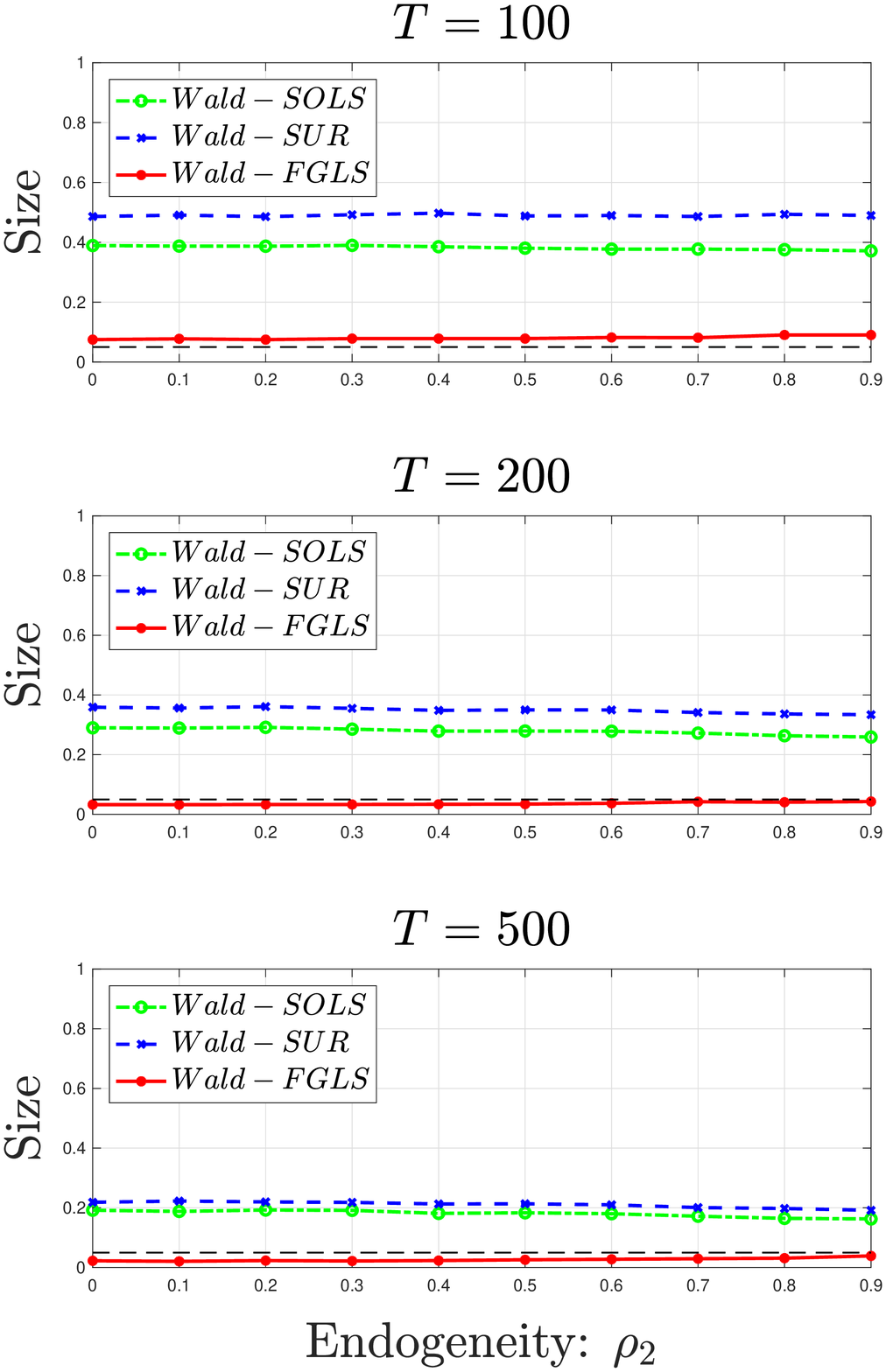}
	\caption{\footnotesize Empirical size of the joint Wald tests $H_0:\beta_{1,4}=\beta_{2,4}=\cdots=\beta_{n,4}=-0.3$, where $\beta_{i,4}$ denote the coefficients in front of $x_{it}^2$. The Wald-SOLS test (green) and Wald-SUR test (blue) are based on Proposition 2 by \cite{wagnergrabarczykhong2019}. The Wald-FGLS test (red) is found in Theorem \ref{thm:test}. We vary the endogeneity parameter $\rho_2$ from $0$ to $0.9$ while keeping $\rho_1=\rho_3=\rho_4=0.8$ fixed. The cross-sectional dimension is $n=3$.}
	\label{fig:size_endo}
\end{figure}	

\newpage
\clearpage	
\begin{figure}[!]
	\centering
	\includegraphics[width=\textwidth]{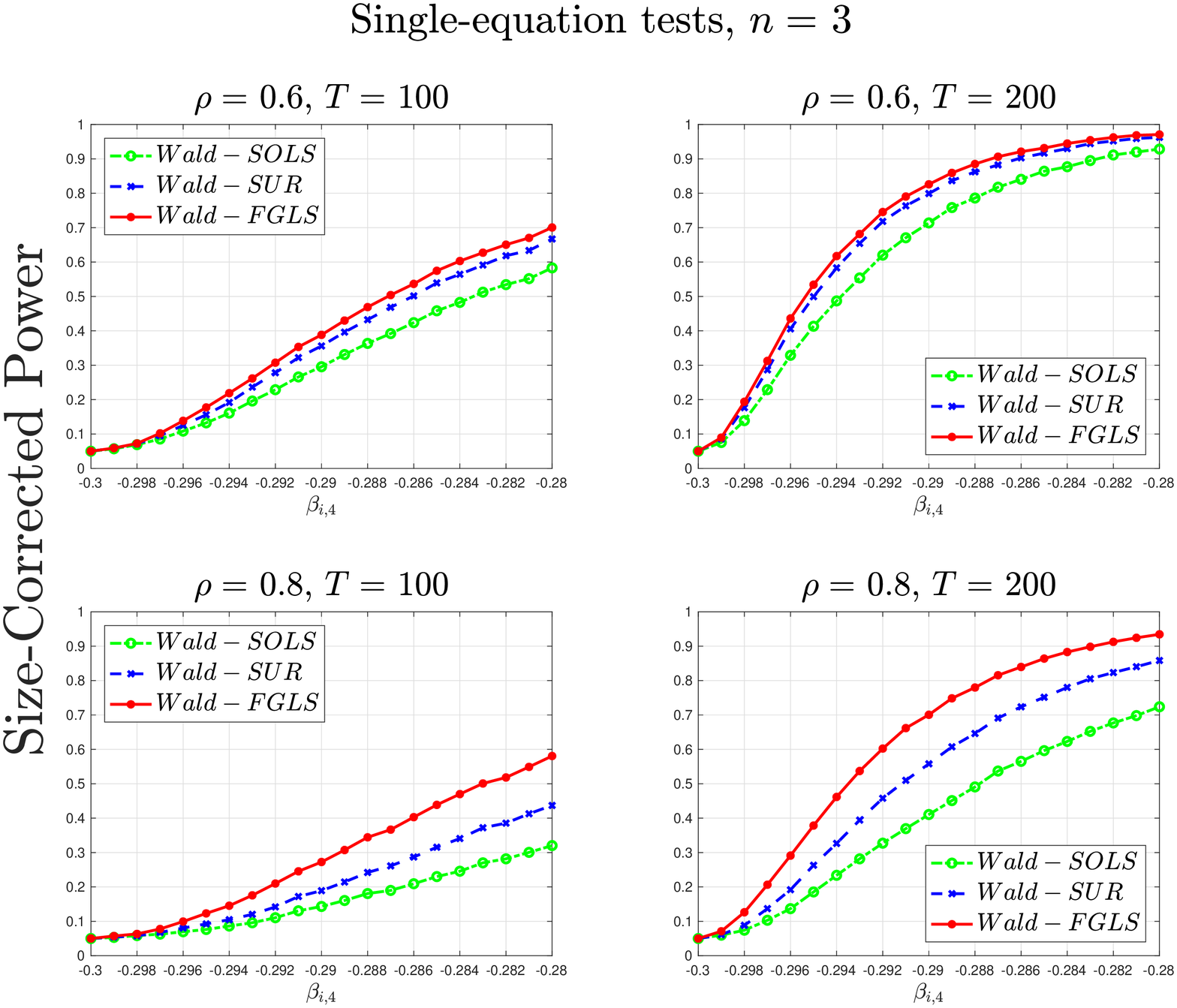}
	\caption{\footnotesize Empirical size-corrected power of the single-equation Wald tests $H_0:\beta_{1,4}=-0.3$ where $\beta_{1,4}$ is the coefficient in front of $x_{1t}^2$. We consider $n=3$, $T\in\{100,200\}$, and $\rho_1=\rho_2=\rho_3=\rho_4=\rho$ with $\rho\in\{0.6,0.8\}$. The Wald-SOLS test (green) and Wald-SUR test (blue) are based on Proposition 2 by \cite{wagnergrabarczykhong2019}. The Wald-FGLS test (red) is found in Theorem \ref{thm:test}.}
	\label{fig:power_singletest_n3}
\end{figure}

\newpage
\clearpage	
\begin{figure}[!]
	\centering
	\includegraphics[width=\textwidth]{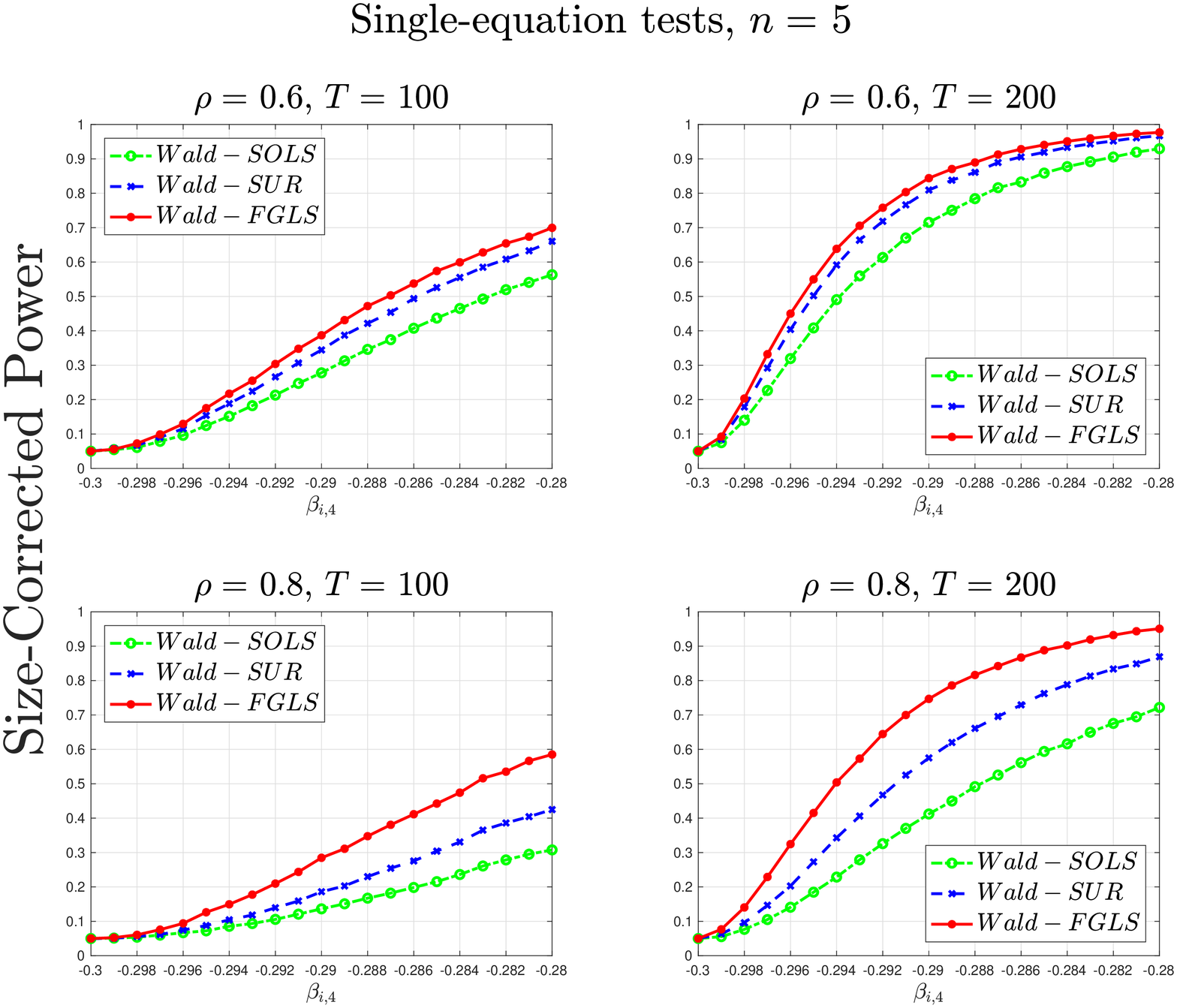}
	\caption{\footnotesize Empirical size-corrected power of the single-equation Wald tests $H_0:\beta_{1,4}=-0.3$ where $\beta_{1,4}$ is the coefficient in front of $x_{1t}^2$. We consider $n=5$, $T\in\{100,200\}$, and $\rho_1=\rho_2=\rho_3=\rho_4=\rho$ with $\rho\in\{0.6,0.8\}$. The Wald-SOLS test (green) and Wald-SUR test (blue) are based on Proposition 2 by \cite{wagnergrabarczykhong2019}. The Wald-FGLS test (red) is found in Theorem \ref{thm:test}.}
	\label{fig:power_singletest_n5}
\end{figure}	

\newpage
\clearpage	
\begin{figure}[!]
	\centering
	\includegraphics[width=\textwidth]{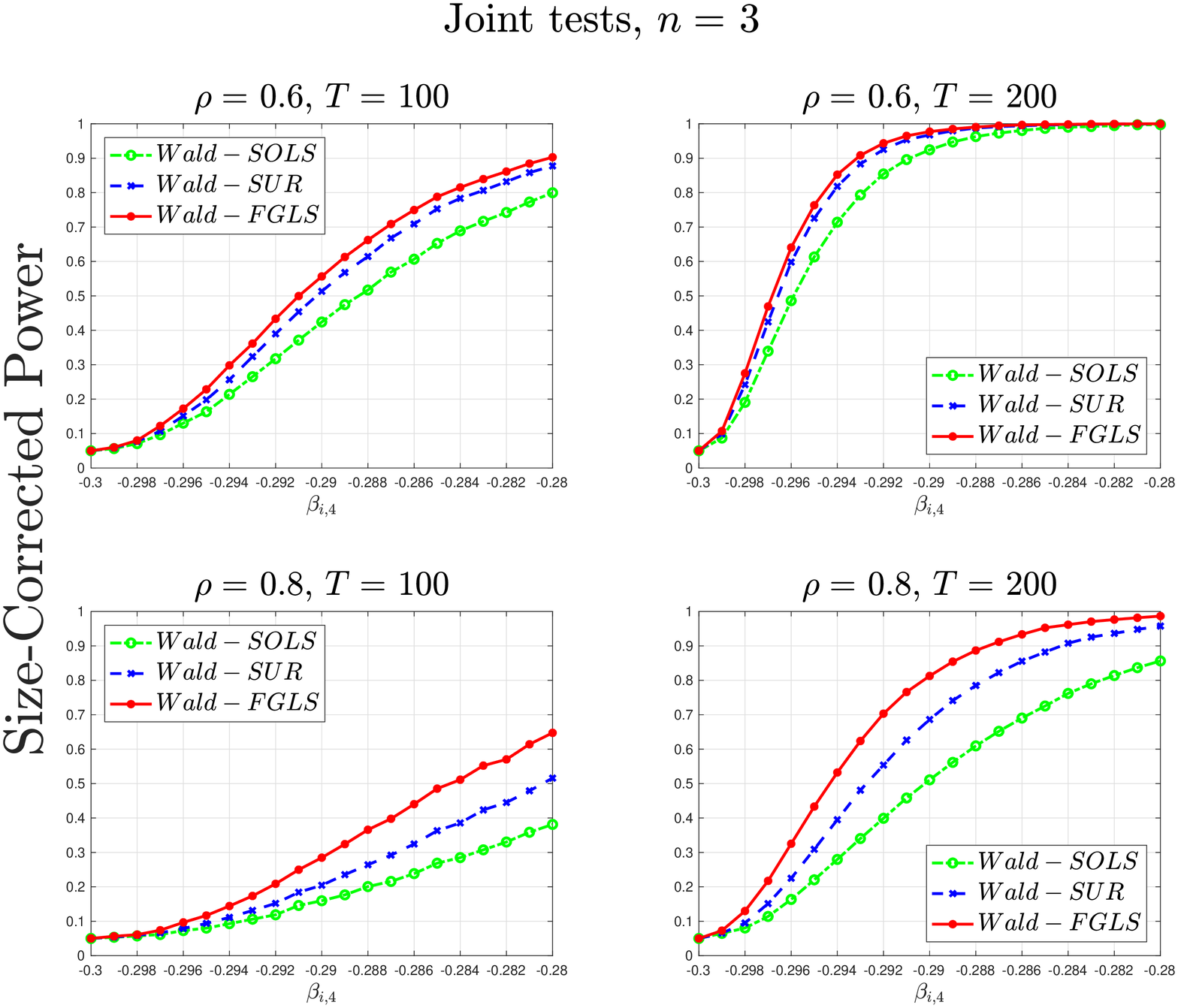}
	\caption{\footnotesize Empirical size-corrected power of the joint Wald tests $H_0:\beta_{1,4}=\beta_{2,4}=\cdots=\beta_{n,4}=-0.3$ where $\beta_{i,4}$ are the coefficients in front of $x_{it}^2$. We consider $n=3$, $T\in\{100,200\}$, and $\rho_1=\rho_2=\rho_3=\rho_4=\rho$ with $\rho\in\{0.6,0.8\}$. The Wald-SOLS test (green) and Wald-SUR test (blue) are based on Proposition 2 by \cite{wagnergrabarczykhong2019}. The Wald-FGLS test (red) is found in Theorem \ref{thm:test}.}
	\label{fig:power_jointtest_n3}
\end{figure}

\newpage
\clearpage	
\begin{figure}[!]
	\centering
	\includegraphics[width=\textwidth]{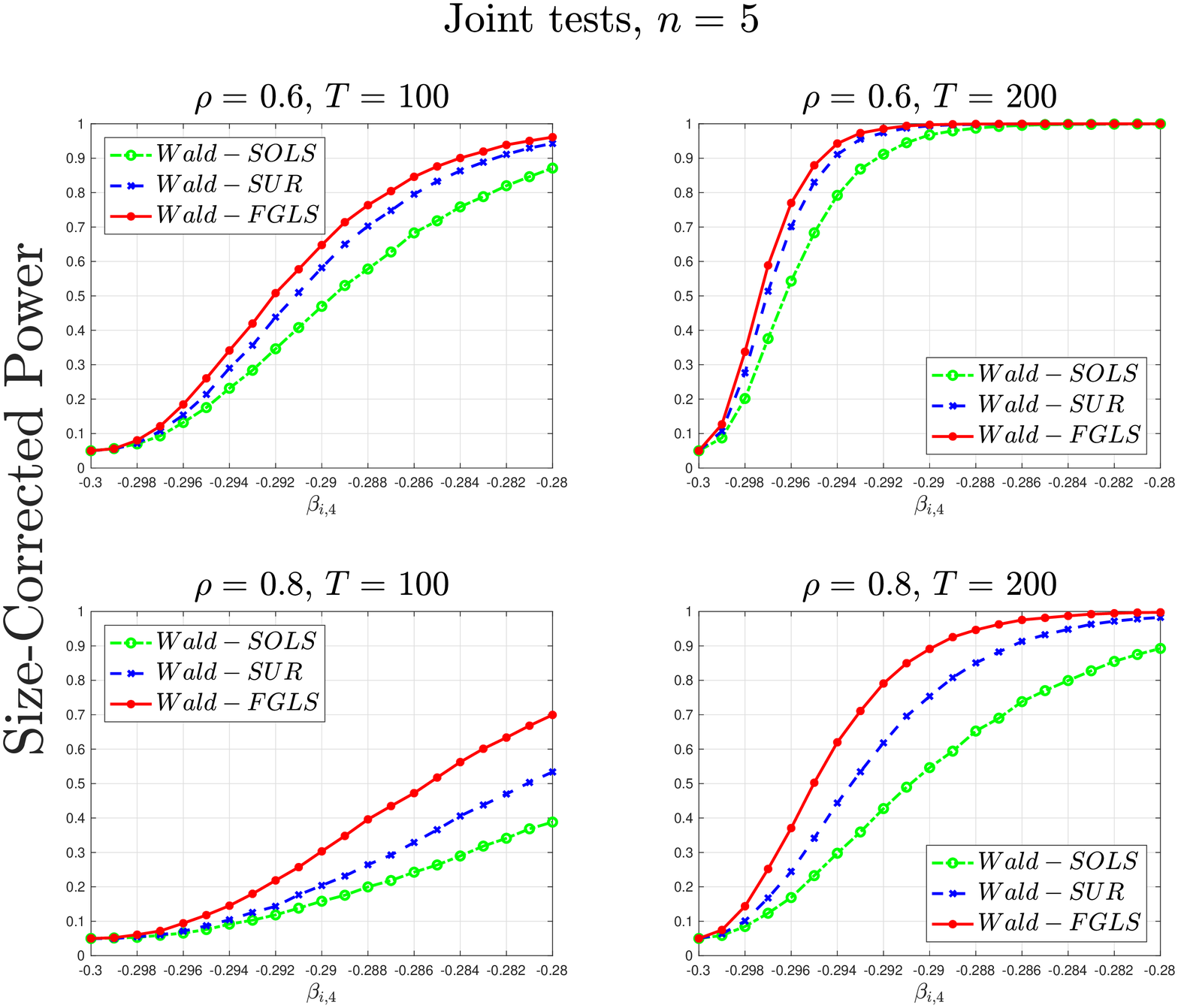}
	\caption{\footnotesize Empirical size-corrected power of the joint Wald tests $H_0:\beta_{1,4}=\beta_{2,4}=\cdots=\beta_{n,4}=-0.3$ where $\beta_{i,4}$ are the coefficients in front of $x_{it}^2$. We consider $n=5$, $T\in\{100,200\}$, and $\rho_1=\rho_2=\rho_3=\rho_4=\rho$ with $\rho\in\{0.6,0.8\}$. The Wald-SOLS test (green) and Wald-SUR test (blue) are based on Proposition 2 by \cite{wagnergrabarczykhong2019}. The Wald-FGLS test (red) is found in Theorem \ref{thm:test}.}
	\label{fig:power_jointtest_n5}
\end{figure}	

\newpage
\begin{table}[]
	\centering
	\caption{Empirical size ($\%$) and power ($\%$) of Bonferroni-type (multivariate) subsampling KPSS tests. The integer $J_1$ in Power DGP1 indicates the number of unit roots contained in errors $\{\vu_t\}$, $J_2$ is related to the number of equations that exclude cubic power terms $x_{it}^3$, $J_3$ specifies the number of spurious relations.}
	\label{tab:ct_tests}
	\resizebox{0.9\textwidth}{!}{%
		\begin{threeparttable}
		\begin{tabular}{lccrrrrrr}
			\toprule
			\multicolumn{1}{l}{} & \multicolumn{1}{l}{} & \multicolumn{1}{l}{} & \multicolumn{3}{c}{$n=3$} & \multicolumn{3}{c}{$n=5$} \\
			\midrule
			\multicolumn{1}{l}{} & \multicolumn{1}{l}{} & $T$ & \multicolumn{1}{c}{$K^{SOLS}$} & \multicolumn{1}{c}{$K^{SUR}$} & \multicolumn{1}{c}{$K^{BIAM}$} & \multicolumn{1}{c}{$K^{SOLS}$} & \multicolumn{1}{c}{$K^{SUR}$} & \multicolumn{1}{c}{$K^{BIAM}$} \\
			\midrule
			\multicolumn{9}{l}{Panel A: Size} \\
			\midrule
			\multirow{9}{*}{$\big(\underline{\lambda}, \bar{\lambda}\big)$, serial correlation} & \multirow{3}{*}{$(0.1, 0.5)$} & 100 & 0.37 & 0.34 & 0.47 & 0.33 & 0.25 & 0.26 \\
			&  & 200 & 0.42 & 0.37 & 0.53 & 0.42 & 0.28 & 0.18 \\
			&  & 500 & 0.71 & 0.60 & 1.24 & 0.64 & 0.58 & 0.80 \\
			\cmidrule(lr){2-9}
			& \multirow{3}{*}{$(0.5, 0.8)$} & 100 & 8.10 & 7.28 & 1.86 & 22.35 & 19.65 & 1.29 \\
			&  & 200 & 3.54 & 3.22 & 2.91 & 8.70 & 7.57 & 1.09 \\
			&  & 500 & 1.92 & 1.69 & 4.54 & 4.02 & 3.40 & 4.78 \\
			\cmidrule(lr){2-9}
			& \multirow{3}{*}{$(0.8, 0.95)$} & 100 & 50.63 & 48.67 & 16.17 & 90.99 & 89.69 & 18.41 \\
			&  & 200 & 40.48 & 38.56 & 14.26 & 83.71 & 81.78 & 12.16 \\
			&  & 500 & 21.87 & 20.36 & 16.25 & 60.09 & 56.87 & 15.49 \\
			\midrule
			\multicolumn{9}{l}{Panel B: Power DGP1} \\
			\midrule
			\multirow{9}{*}{$J_1$, $\#\left\{\text{unit roots}\right\}$} & \multirow{3}{*}{1} & 100 & 44.03 & 39.92 & 21.14 & 57.65 & 52.05 & 14.00 \\
			&  & 200 & 57.97 & 50.90 & 33.85 & 74.09 & 66.64 & 27.24 \\
			&  & 500 & 67.23 & 57.37 & 56.63 & 88.41 & 80.23 & 49.93 \\
			\cmidrule(lr){2-9}
			& \multirow{3}{*}{2} & 100 & 66.64 & 63.83 & 38.07 & 84.77 & 80.36 & 29.12 \\
			&  & 200 & 77.77 & 73.00 & 52.88 & 94.38 & 91.16 & 49.87 \\
			&  & 500 & 84.19 & 78.83 & 78.40 & 98.57 & 96.89 & 72.02 \\
			\cmidrule(lr){2-9}
			& \multirow{3}{*}{$n$} & 100 & 78.89 & 77.41 & 54.26 & 99.39 & 99.24 & 64.89 \\
			&  & 200 & 88.40 & 86.53 & 70.07 & 99.97 & 99.95 & 88.75 \\
			&  & 500 & 91.28 & 90.19 & 89.55 & 100.00 & 100.00 & 96.53 \\
			\midrule
			\multicolumn{9}{l}{Panel C: Power DGP2} \\
			\midrule
			\multirow{9}{*}{$J_2$, $\#\left\{\text{misspecified equations}\right\}$} & \multirow{3}{*}{1} & 100 & 10.92 & 10.51 & 3.47 & 18.67 & 18.40 & 1.91 \\
			&  & 200 & 26.41 & 25.58 & 7.72 & 41.90 & 41.41 & 4.95 \\
			&  & 500 & 55.72 & 55.28 & 24.10 & 77.76 & 77.22 & 18.36 \\
			\cmidrule(lr){2-9}
			& \multirow{3}{*}{2} & 100 & 17.83 & 16.97 & 6.15 & 30.77 & 29.77 & 3.70 \\
			&  & 200 & 40.29 & 39.04 & 15.57 & 62.01 & 61.15 & 10.51 \\
			&  & 500 & 69.64 & 68.50 & 38.48 & 91.66 & 91.27 & 30.87 \\
			\cmidrule(lr){2-9}
			& \multirow{3}{*}{$n$} & 100 & 23.16 & 22.53 & 9.88 & 53.66 & 51.21 & 10.45 \\
			&  & 200 & 47.25 & 45.60 & 22.83 & 83.78 & 82.73 & 29.05 \\
			&  & 500 & 75.60 & 73.98 & 56.79 & 98.23 & 98.15 & 69.35 \\
			\midrule
			\multicolumn{9}{l}{Panel D: Power DGP3} \\
			\midrule
			\multirow{9}{*}{$J_3$, $\#\left\{\text{spurious relations}\right\}$} & \multirow{3}{*}{1} & 100 & 77.03 & 77.06 & 28.34 & 94.26 & 94.12 & 21.04 \\
			&  & 200 & 89.55 & 89.14 & 35.64 & 98.97 & 98.84 & 31.79 \\
			&  & 500 & 97.26 & 97.05 & 55.61 & 99.91 & 99.91 & 48.48 \\
			\cmidrule(lr){2-9}
			& \multirow{3}{*}{2} & 100 & 87.40 & 86.75 & 52.85 & 98.95 & 98.81 & 43.05 \\
			&  & 200 & 94.95 & 94.33 & 57.95 & 99.89 & 99.89 & 58.15 \\
			&  & 500 & 98.69 & 98.49 & 77.59 & 100.00 & 100.00 & 69.59 \\
			\cmidrule(lr){2-9}
			& \multirow{3}{*}{$n$} & 100 & 91.55 & 90.60 & 71.95 & 99.97 & 99.92 & 87.87 \\
			&  & 200 & 96.71 & 96.33 & 74.85 & 99.99 & 100.00 & 93.86 \\
			&  & 500 & 98.04 & 97.62 & 91.66 & 100.00 & 100.00 & 98.02\\
			\bottomrule
		\end{tabular}%
	    \begin{tablenotes}[para,flushleft]
	    	Note: To decrease the computational burden we reduced the number of Monte Carlo replications to $10^4$.
    	\end{tablenotes}
    \end{threeparttable}
	}
\end{table}
\end{appendices}

\begin{table}[h!]
	\centering
	\caption{Outcomes for the joint tests of cointegration in the empirical study. At a significance level of 5\%, the null hypothesis of cointegration is rejected when `'rejection rule' is less than 5\%.}
	\label{tab:multivariateKPSS}
	\begin{tabular}{lrrr}
		\toprule
		& \multicolumn{1}{c}{$K^{SOLS}$} & \multicolumn{1}{c}{$K^{SUR}$} & \multicolumn{1}{c}{$K^{BIAM}$} \\
		\midrule
		Statistic 				& 16.54	& 12.66	& 8.19 \\
		Rejection Rule (in \%)	& 0.00	& 0.03	& 2.73 \\
		Block Size				& 22 		& 20 		& 22 \\
		$\#\{\text{blocks}\}$		& 6 		& 7 		& 6\\
		\bottomrule
	\end{tabular}
\end{table}

\begin{figure}[h!]
	\centering
	\includegraphics[width=\textwidth]{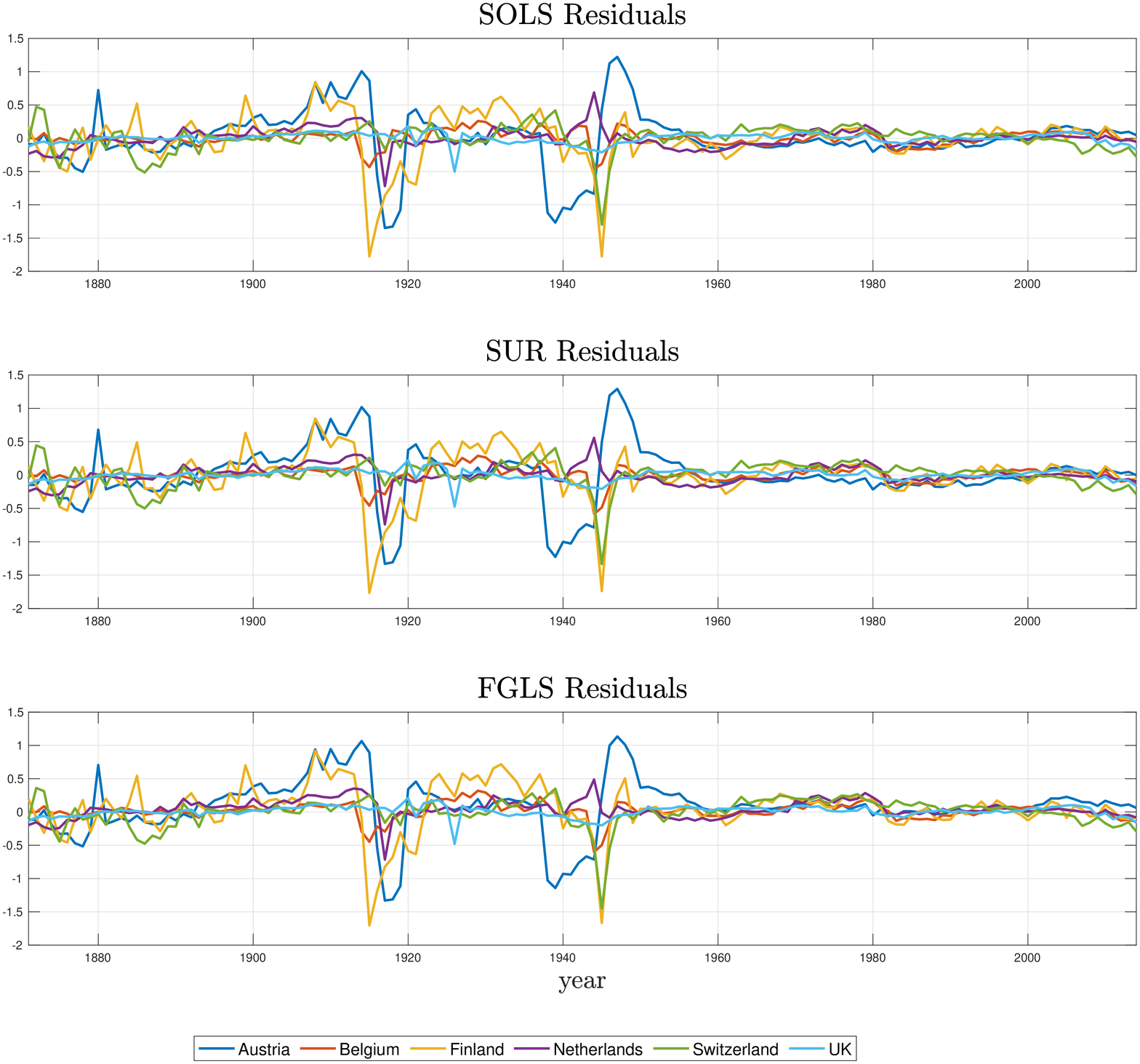}
	\caption{\footnotesize The plots of the residuals $\hat \vu_{t,SOLS}=\vy_t-  \mZ_t \widehat{\vbeta}_{SOLS}^+$ (top), $\hat \vu_{t,SUR}=\vy_t-  \mZ_t \widehat{\vbeta}_{SUR}^+$ (middle), and $\hat \vu_{t,FGLS}=\vy_t-  \mZ_t \widehat{\vbeta}_{FGLS}^+$ (bottom) for the empirical study.}
	\label{fig:KPSSresiduals}
\end{figure}

\begin{table}[h!]
	\centering
	\caption{Estimates for $\beta_{i,3}$ and $\beta_{i,4}$ from the quadratic EKC model in \eqref{eq:ekc_quadraticmodel}. The numbers between parentheses are 95\% asymptotic confidence intervals. Turning points are computed as $\exp\left(-\widehat\beta_{i,3}/2\,\widehat\beta_{i,4}\right)$.}
	\label{tab:estimationresults}
	\begin{tabular}{ll c c c }
		\toprule
									&  		& \multicolumn{1}{c}{$x_t$} & \multicolumn{1}{c}{$x_t^2$} & \multicolumn{1}{c}{Turning point} \\
		\midrule
									& SOLS 	& 9.173			& -0.411				& 68,900 \\
									&		& (6.797,11.548)	& (-0.535,-0.288)\\
		\multirow{2}{*}{Austria}			& SUR	& 8.494			& -0.370				& 76,211 \\
									&		& (6.764,10.223)	& (-0.464,-0.276)\\
									& FGLS	& 6.553			& -0.276				& 708,712 \\
									&		& (2.138,10.967)	& (-0.513,-0.040)\\
		\midrule	
									& SOLS 	& 12.927			& -0.645 				& 22,420 \\
									&		& (11.795,15.059)	& (-0.702,-0.589)\\
		\multirow{2}{*}{Belgium}			& SUR 	& 9.973			& -0.503				& 20,195\\
									&		& (9.158,10.787)	& (-0.545,-0.461)\\
									& FGLS	& 8.762			& -0.443				& 19,795\\
									&		& (7.236,10,287)	& (-0.521,-0.365)\\
		\midrule
									& SOLS	& 15.676			& -0.716				& 56,967\\
									&		& (14.162,17.289)	& (-0.788,-0.643)\\
		\multirow{2}{*}{Finland}			& SUR	& 15.136			& -0.684 				& 63,400\\
									&		& (14.030,16.242)	& (-0.742,-0.627)\\
									& FGLS	& 14.392			& -0.646				&68,775\\
									& 		& (12.075,16.708)	& (-0.769,-0.523)\\
		\midrule
									& SOLS 	& 11.382			& -0.540				&38,120 \\
									&		& (10.140,12.624)	& (-0.602,-0.477)\\
		\multirow{2}{*}{Netherlands}		& SUR 	& 10.063			& -0.475				&39,637 \\
									&		& (9.183,10.944)	& (-0.522,-0.429)\\
									& FGLS 	& 9.102			& -0.430				& 39,908\\
									&		& (7.606,10.597)	& (-0.506,-0.353)\\
		\midrule
									& SOLS	& 7.070			& -0.232				& $4.287\times 10^6$\\
									&		& (5.516,8.624)		& (-0.310,-0.153)\\
		\multirow{2}{*}{Switzerland}		& SUR	& 6.962			& -0.232				& $3.316\times 10^6$\\
									&		& (5.481,8.443)		& (-0.308,-0.156)\\
									& FGLS	& 7.052			& -0.254				& $1.074\times 10^6$\\
									&		& (5.173,8.932)		& (-0.350,-0.158)\\
		\midrule
									& SOLS 	& 8.450			& -0.429				& 20,242 \\
									&		& (6.976,10.020)	& (-0.502,-0.355)\\
		\multirow{2}{*}{United Kingdom}	& SUR	& 9.523			& -0.475				& 22,596\\
									&		& (8.486,10.560)	& (-0.527,-0.423)\\
									& FGLS	& 9.056			& -0.453				& 21,887\\
									&		& (7.244,10.868)	& (-0.542,-0.364)\\
		\bottomrule
	\end{tabular}
\end{table}

\clearpage
\newpage
\begin{center}
	Supplemental Appendix to:\\\bigskip
	\begin{Large}
		Efficient Estimation by Fully Modified GLS with an Application to the Environmental Kuznets Curve\\
		\bigskip
	\end{Large}
	\large{Yicong Lin, and Hanno Reuvers}
\end{center}

\appendix
\section{Additional Proofs}\label{sec:additional_proofs}

\begin{lemma}\label{L2Baxter}
 If $\{ \vu_t\}$ satisfies Assumption \ref{assumpt1:linearproc}, then for any $m\geq 1$, there exists a constant $C>0$ such that
 	\begin{equation}
	\sum_{j=1}^{m}\left\|\mA_j(m)-\mA_j\right\|_{\calF}^2\leq C\sum_{j=m+1}^{\infty}\left\|\mA_j\right\|_{\calF}^2.
	\label{eq:baxter}
	\end{equation}
\end{lemma}
\begin{proof}
	In view of page 257 of \cite{hannandeistler2012}, the summability condition of Assumption \ref{assumpt1:linearproc} implies that the spectral density matrix is bounded and bounded away from zero. The boundedness condition in \cite{chengpourahmadi1993} is thus satisfied and \eqref{eq:baxter} follows from their Theorem 2.2.
\end{proof}

\begin{lemma}[Implications of the First Moment Bound Theorem]\label{lem:fmbt}
Let Assumption \ref{assumpt1:linearproc} hold, and define
	\begin{equation}\label{eq:def_reverse_errors}
	\veta_{t+1,\ell}=\vu_{t+1}-\mA(\ell)\vu_t(\ell),\qquad \text{where}\quad \vu_t(\ell)=\left[\vu_t',\vu_{t-1}',\cdots,\vu_{t-\ell+1}'\right]'.
	\end{equation}
The following three inequalities are true:
	\begin{enumerate}[(a)]
		\item $\E\Big\|\frac{1}{T-q}\sum_{t=q}^{T-1}\vu_t(q)\vu_t(q)'-\E\left(\vu_t(q)\vu_t(q)'\right)\Big\|^2\leq C\frac{q^2}{T-q}$;
		\item $\E\left\|\frac{1}{T-\ell}\sum_{t=\ell}^{T-1}\big(\veta_{t+1,\ell}-\veta_{t+1}\big)\vu_t(\ell)'\right\|^r\leq C\left(\frac{\ell}{T-\ell}\right)^{r/2}\Big(\sum_{j=\ell+1}^{\infty}\left\|\mA_j\right\|_{\calF}^2\Big)^{r/2}$, for some $r\geq 2$ and any $1\leq\ell\leq q$;
		\item $\E\left\|\frac{1}{T-\ell}\sum_{t=\ell}^{T-1}\Big(\veta_{t+1,\ell}^{}\veta_{t+1,\ell}'\Big)-\E\Big(\veta_{t+1,\ell}^{}\veta_{t+1,\ell}'\Big)\right\|^r\leq C(T-\ell)^{-r/2}$, for some $r\geq 2$ and any $1\leq\ell\leq q$.
	\end{enumerate}
\end{lemma}
\begin{proof}
 \textbf{(a)} Since $\|\cdot\|^2\leq \|\cdot\|_{\calF}^2$, we obtain 
	\begin{multline}\label{eq:resulta}
	\E\Bigg\|\frac{1}{T-q}\sum_{t=q}^{T-1}\vu_t(q)\vu_t(q)'-\E\left(\vu_t(q)\vu_t(q)'\right)\Bigg\|^2\\
	\leq \frac{1}{(T-q)^2} \sum_{k,\ell=1}^{q} \sum_{i,j=1}^{n} \E \left[ \sum_{t=q}^{T-1}\left[u_{i,t-k+1} u_{j,t-\ell+1}-\E(u_{i,t-k+1}u_{j,t-\ell+1})\right] \right]^2,
	\end{multline}
	where $u_{i,t}$ denotes the $i^{th}$ element $\vu_t$. As remarked in the main text, the lag polynomial $\lagpol{A}(L)$ is invertible by Assumption \ref{assumpt1:linearproc}. Recall $\lagpol{C}(L)=[\lagpol{A}(L)]^{-1}=\sum_{j=0}^\infty \mC_j L^j$ with $\mC_0=\mI_n$, and $\sum_{j=0}^\infty j \,\| \mC_j \|_{\calF}<\infty$. We observe that $u_{i,t}=\sum_{j=0}^{\infty}\row_i\big(\mC_j\big)\veta_{t-j}$. By Proposition 10.2(b) of \cite{hamilton1994}, absolute summability of the coefficient matrices $\{\mC_j\}_{j=0}^\infty$ implies $\sum_{s=0}^\infty |\gamma_{u,k}(s)|<\infty$ where $\gamma_{u,k}(s)=\E(u_{k,t}u_{k,t-s})$. The conditions for the First Moment Bound Theorem (FMBT) in \cite{findleywei1993} are thus satisfied. Choosing $q(t,s)=1$ if $t=s\geq q$ (the banding parameter) and $q(t,s)=0$ otherwise,
	\begin{equation*}
	\begin{aligned}
	\E \Bigg[ \sum_{t=q}^{T-1} &\left[u_{i,t-k+1} u_{j,t-\ell+1}-\E(u_{i,t-k+1}u_{j,t-\ell+1})\right] \Bigg]^2=\E\Bigg[\sum_{t,s=1}^{T-1}q(t,s)\left[u_{i,t-k+1}u_{j,s-\ell+1}-\E(u_{i,t-k+1}u_{j,s-\ell+1})\right]\Bigg]^2 \\
	&\leq C\sum_{t,s,l,w=1}^{T-1}q(t,s)q(l,w)\gamma_{u,i}(t-l)\gamma_{u,j}(s-w)=  C\sum_{t,l=q}^{T-1} \gamma_{u,i}(t-l)\gamma_{u,j}(t-l)\\
	&\leq C\gamma_{u,j}(0)\sum_{t,l=q}^{T-1}|\gamma_{u,i}(t-l)|\leq C(T-q)\left[\gamma_{u,j}(0)\sum_{t=-\infty}^{\infty}|\gamma_{u,i}(t)|\right]\leq C(T-q),
	\end{aligned}
	\end{equation*}
	by the FMBT. This bound holds for general $k$, $\ell$, $i$ and $j$, and  \textbf{(a)} thereby follows from (\ref{eq:resulta}).
	
 \textbf{(b)} For $1\leq \ell\leq q$ and $r\geq 2$, we have
	\begin{equation}
	\E\left\|\frac{1}{T-\ell}\sum_{t=\ell}^{T-1}\left(\veta_{t+1,\ell}-\veta_{t+1}\right)\vu_t(\ell)'\right\|^r	\leq \ell^{r/2-1}(T-\ell)^{-r}\sum_{s=0}^{\ell-1}\E\left\|\sum_{t=\ell}^{T-1}\left(\veta_{t+1,\ell}-\veta_{t+1}\right)\vu_{t-s}'\right\|_{\calF}^r,
	\label{eq:partb1}
	\end{equation}
	by $\| \cdot \|^r\leq (\| \cdot \|_{\calF}^2)^{r/2}$ and the $c_r$-inequality. By assumption, $\veta_{t+1}$ is uncorrelated with $\big[\vu_{t-\ell+1}',\ldots,\vu_t'\big]'$ implying that $\E\big[\sum_{t=\ell}^{T-1} (\veta_{t+1,\ell}-\veta_{t+1})\vu_{t-s}'\big]=\mZeros$. The FMBT can thus be applied directly without having to express the quadratic form in deviations from the mean. However, some rewriting is needed to obtain expressions in scalar random sequences. To this end, use $A_{j,kl}$ and $A_{j,kl}(\ell)$ to denot the $(k,l)^{th}$ element of $\mA_j$ and $\mA_{j}(\ell)$, respectively. Setting $\mA_{j}(\ell)=\mZeros$ for $j>\ell$, we have $\veta_{t+1,\ell}-\veta_{t+1}=\sum_{j=1}^{\infty}\big[\mA_j-\mA_{j}(\ell)\big]\vu_{t+1-j}$,
	and hence
	\begin{equation}
	\begin{aligned}
	\E\Bigg\|\sum_{t=\ell}^{T-1}&\left(\veta_{t+1,\ell}-\veta_{t+1}\right)\vu_{t-s}'\Bigg\|_{\calF}^r=\E\left\{\sum_{k,m=1}^{n}\left[\sum_{t=\ell}^{T-1}u_{m,t-s}\sum_{j=1}^{\infty}\sum_{l=1}^{n}\left(A_{j,kl}-A_{j,kl}(\ell)\right)u_{l,t+1-j}\right]^2\right\}^{r/2}\\
	&\leq n^{r-2}\sum_{k,m=1}^{n}\E\left|\sum_{t=\ell}^{T-1}u_{m,t-s}\sum_{j=1}^{\infty}\sum_{l=1}^{n}\left(A_{j,kl}-A_{j,kl}(\ell)\right)u_{l,t+1-j}\right|^r=n^{r-2}\sum_{k,m=1}^{n} \E\left|\sum_{t=\ell}^{T-1}u_{m,t-s}^{}u_t^*\right|^r
	\end{aligned}
	\label{eq:univariateSeriespartb}
	\end{equation}
	with $u_t^*=\sum_{j=1}^{\infty}\sum_{l=1}^{n}\big(A_{j,kl}-A_{j,kl}(\ell)\big)u_{l,t+1-j}$, where we suppress the dependence on the index $k$ (also below) without confusion. To apply the FMBT, we define the autocovariances $\gamma_{u^*}(t-h)=\E\big(u_{t}^*u_{h}^*\big)$, the difference in lag polynomial coefficients $\va_{l}(\ell)=\left[A_{1,kl}-A_{1,kl}(\ell),A_{2,kl}-A_{2,kl}(\ell),\ldots\right]'$ and $\mSigma_{u_l,\infty}=\big[\gamma_{u,l}(i-j),1\leq i,j<\infty\big]$. By the Cauchy-Schwartz inequality, the $c_r$-inequality, and boundedness of the maximum eigenvalue of $\mSigma_{u_l,\infty}$, we obtain
	\begin{equation}
	\begin{aligned}
	\gamma_{u^*}(t-h) &\leq \gamma_{u^*}(0)=\E\left[\sum_{l=1}^{n}\sum_{j=1}^{\infty}\left(A_{j,kl}-A_{j,kl}(\ell)\right)u_{l,t+1-j}\right]^2 \leq n\sum_{l=1}^{n}\left[\va_{kl}(\ell)'\mSigma_{u_l,\infty}\va_{kl}(\ell)^{}\right] \\
	&\leq Cn\sum_{l=1}^{n}\|\va_{kl}(\ell)\|^2 \leq Cn\sum_{j=1}^{\infty}\left\|\mA_j-\mA_j(\ell)\right\|_{\calF}^2\leq Cn\sum_{j=\ell+1}^{\infty}\left\|\mA_j\right\|_{\calF}^2.
	\end{aligned}
	\label{eq:FMBTcovariance}
	\end{equation}
	Applying the FMBT, we have
	\begin{equation}
	\begin{aligned}
	\left(\E\left|\sum_{t=\ell}^{T-1}u_{m,t-s}^{}u_t^*\right|^r\right)^{1/r}&\leq C \left[\sum_{t,h=\ell}^{T-1}\gamma_{u,m}^{}(t-h)\gamma_{u^*}(t-h)\right]^{1/2}\\
	&\leq C\left[\gamma_{u^*}(0)(T-\ell)\sum_{t=-\infty}^{\infty}|\gamma_{u,m}(t)|\right]^{1/2}\leq C(T-\ell)^{1/2}\left(\sum_{j=\ell+1}^{\infty}\left\|\mA_j\right\|_{\calF}^2\right)^{1/2},
	\end{aligned}
	\label{eq:FMBTpartb}
	\end{equation}
	using (\ref{eq:FMBTcovariance}) and the absolute summability of $\{\gamma_{u,m}(t)\}$. Combining \eqref{eq:partb1}, \eqref{eq:univariateSeriespartb} and \eqref{eq:FMBTpartb} leads to the desired inequality.
	
 \textbf{(c)} The equality $\veta_{t+1,\ell}=\Big(\mI_n-\sum_{j=1}^\ell \mA_j(\ell) L^j \Big) \lagpol{C}(L)\veta_{t+1}$ shows that $\veta_{t+1,\ell}$ has a linear process representation in terms of $\veta_t$. Theorem 6.6.12 of \cite{hannandeistler2012} implies that $\sup_{1\leq \ell<\infty}\sum_{j=0}^{\ell}\|\mA_j(\ell)\|_{\calF}<\infty$. By Propositions 10.2(b) and 10.3 of \cite{hamilton1994}, both the coefficient matrices associated with $\Big(\mI_n-\sum_{j=1}^\ell \mA_j(\ell) L^j \Big) \lagpol{C}(L)$ and the autocovariances $\left\{\E\left(\eta_{k,t+1,\ell}\eta_{k,t+1-s,\ell}\right)\right\}_{s=0}^{\infty}$ are absolutely summable, where $\eta_{k,t+1,\ell}$ is the $k^{th}$ entry of $\veta_{t+1,\ell}$. The proof is completed using the  $c_r$-inequality and the FMBT, that is, for $r\geq 2$,
	\begin{equation*}
	\begin{aligned}
	&\E\left\|\frac{1}{T-\ell}\sum_{t=\ell}^{T-1}\Big(\veta_{t+1,\ell}^{}\veta_{t+1,\ell}'\Big)-\E\Big(\veta_{t+1,\ell}^{}\veta_{t+1,\ell}'\Big)\right\|^r\\
	&\leq n^{r-2}\frac{1}{(T-\ell)^r}\sum_{k,m=1}^{n}\E \Bigg| \sum_{t=\ell}^{T-1}\left[\eta_{k,t+1,\ell}\eta_{m,t+1,\ell}-\E\left(\eta_{k,t+1,\ell}\eta_{m,t+1,\ell}\right)\right]\Bigg|^r\\
	&\leq C n^{r-2}\frac{1}{(T-\ell)^r}\sum_{k,m=1}^{n}\Bigg[\sum_{t,l=\ell}^{T-1}\E\left(\eta_{k,t+1,\ell}\eta_{k,l+1,\ell}\right)\E\left(\eta_{m,t+1,\ell}\eta_{m,l+1,\ell}\right)\Bigg]^{r/2}\leq C\frac{1}{(T-\ell)^{r/2}}.
	\end{aligned}
	\end{equation*}
\end{proof}

\begin{lemma}\label{lem:max_bound}
If Assumptions \ref{assumpt1:linearproc}-\ref{assump4:bandingparameter} hold, then
\begin{equation}
 \max_{1\leq\ell\leq q}\big\|\widehat{\mA}(\ell)-\mA(\ell)\big\| = O_p\left(q/\sqrt{T}\, \right), \qquad
	\max_{1\leq\ell\leq q}\big\|\widehat{\mS}(\ell)-\mS(\ell)\big\| = O_p\left(q/\sqrt{T}\, \right).
\label{eq:max_diff_mAmS}
\end{equation}
\end{lemma}	
\begin{proof}
Recall the definition of $\veta_{t+1,\ell}$ and $\vu_t(\ell)$ in \eqref{eq:def_reverse_errors}. Similarly, define
	\begin{equation}\label{eq:def_reverse_residuals}
	\widehat{\veta}_{t+1,\ell}=\widehat{\vu}_{t+1}-\mA(\ell)\widehat{\vu}_t(\ell),\quad\widetilde{\veta}_{t+1,\ell}=\widehat{\vu}_{t+1}-\widehat{\mA}(\ell)\widehat{\vu}_t(\ell),\quad \widehat{\vu}_t(\ell)=\left[\widehat{\vu}_t',\widehat{\vu}_{t-1}',\cdots,\widehat{\vu}_{t-\ell+1}'\right]'.
	\end{equation}
We first prove $ \max_{1\leq\ell\leq q}\big\|\widehat{\mA}(\ell)-\mA(\ell)\big\| = O_p\left(q/\sqrt{T}\, \right)$. Since $\left(\widehat{\mA}(\ell)-\mA(\ell)\right)\widehat{\vu}_t(\ell) = \widehat{\veta}_{t+1,\ell} - \widetilde{\veta}_{t+1,\ell}$ and $\frac{1}{T-\ell}\sum_{t=\ell}^{T-1}\widetilde{\veta}_{t+1,\ell}\widehat{\vu}_t(\ell)'=\mZeros$ (the first-order condition from \eqref{eq:sampleMCD}), we have
	\begin{equation}\label{eq:diff_mA_decomposition}
	\left(\widehat{\mA}(\ell)-\mA(\ell)\right)\left(\frac{1}{T-\ell}\sum_{t=\ell}^{T-1}\widehat{\vu}_t(\ell)\widehat{\vu}_t(\ell)'\right)=\frac{1}{T-\ell}\sum_{t=\ell}^{T-1}\widehat{\veta}_{t+1,\ell}\widehat{\vu}_t(\ell)'.
	\end{equation}
If we can show that $\frac{1}{T-q}\sum_{t=q}^{T-1}\widehat{\vu}_t(q)\widehat{\vu}_t(q)'$ is asymptotically invertible, then $\frac{1}{T-\ell}\sum_{t=\ell}^{T-1}\widehat{\vu}_t(\ell)\widehat{\vu}_t(\ell)'$ must also be asymptotically invertible with probability 1, for any $1\leq\ell\leq q$.\footnote{If the matrix $\mQ$ is invertible, then each leading principle submatrix of $\mQ$ is invertible as well.\label{footnote:leadingprinciple}} By the triangular inequality,  $\left\|\frac{1}{T-q}\sum_{t=q}^{T-1}\widehat{\vu}_t(q)\widehat{\vu}_t(q)'-\E\big(\vu_t(q)\vu_t(q)'\big)\right\|\leq \Rmnum{1}_a+\Rmnum{1}_b$, where
	\begin{equation*}
	\Rmnum{1}_a=\left\|\frac{1}{T-q}\sum_{t=q}^{T-1}\vu_t(q)\vu_t(q)'-\E\big(\vu_t(q)\vu_t(q)'\big)\right\|=O_p\left(q/\sqrt{T} \, \right)
	\end{equation*}
	by Chebyshev's inequality and Lemma \ref{lem:fmbt} $(\rmnum{1})$, and 
	\begin{multline*}
	\Rmnum{1}_b=\left\|\frac{1}{T-q}\sum_{t=q}^{T-1}\widehat{\vu}_t(q)\widehat{\vu}_t(q)'-\frac{1}{T-q}\sum_{t=q}^{T-1}\vu_t(q)\vu_t(q)'\right\|\\
	\leq \frac{1}{T-q}\sum_{t=q}^{T-1}\big\| \widehat{\vu}_t(q)-\vu_t(q)\big\|^2+2\sqrt{\frac{1}{T-q}\sum_{t=q}^{T-1}\big\| \widehat{\vu}_t(q)-\vu_t(q)\big\|^2}\sqrt{\frac{1}{T-q}\sum_{t=q}^{T-1}\big\|\vu_t(q)\big\|^2},
	\end{multline*}
since $\left\|\sum_t\va_t^{}\va_t'-\sum_t\vb_t^{}\vb_t'\right\|\leq \sum_t\left\|\va_t-\vb_t\right\|^2+2\sqrt{\sum_t\left\|\va_t-\vb_t\right\|^2}\sqrt{\sum_t\left\|\vb_t\right\|^2}$. We have $\frac{1}{T-q}\sum_{t=q}^{T-1}\| \widehat{\vu}_t(q)-\vu_t(q)\|^2=\frac{1}{T-q}\sum_{t=q}^{T-1}\sum_{s=t-q+1}^{t}\|\widehat{\vu}_s-\vu_s\|^2\leq \frac{q}{T-q}\|\widehat{\vu}-\vu\|^{2}=\frac{q}{T} O_p\left(1\right)$ by Assumption \ref{assump3:residuals}. Because $\frac{1}{T-q}\sum_{t=q}^{T-1}\|\vu_t(q)\|=O_p(q)$ by Markov's inequality, we conclude $\Rmnum{1}_b=O_p\left(q/\sqrt{T}\right)$. Overall, this gives 
	\begin{equation}\label{eq:invertibility}
	\left\|\frac{1}{T-q}\sum_{t=q}^{T-1}\widehat{\vu}_t(q)\widehat{\vu}_t(q)'-\E\big(\vu_t(q)\vu_t(q)'\big)\right\|=O_p\left(\frac{q}{\sqrt{T}}\right) + O_p\left(\frac{q}{\sqrt{T}}\right)=o_p(1).
	\end{equation}
Now observe that $\E\big(\vu_t(q)\vu_t(q)'\big)$ is a leading principal submatrix of $\mSigma_{\vu}$ (thus invertible, see footnote \ref{footnote:leadingprinciple}). As a result, $\frac{1}{T-q}\sum_{t=q}^{T-1}\widehat{\vu}_t(q)\widehat{\vu}_t(q)'$ is asymptotically invertible. 
	
	We subsequently bound the RHS of \eqref{eq:diff_mA_decomposition} as follows: $\max_{1\leq\ell\leq q}\left\|\frac{1}{T-\ell}\sum_{t=\ell}^{T-1}\widehat{\veta}_{t+1,\ell}\widehat{\vu}_t(\ell)'\right\|\leq \Rmnum{2}_a+\ldots+\Rmnum{2}_e$, where $\Rmnum{2}_a=\max_{1\leq\ell\leq q}\left\|\frac{1}{T-\ell}\sum_{t=\ell}^{T-1}\veta_{t+1}\vu_t(\ell)'\right\|$, $\Rmnum{2}_b=\max_{1\leq\ell\leq q}\left\|\frac{1}{T-\ell}\sum_{t=\ell}^{T-1}\left(\veta_{t+1,\ell}-\veta_{t+1}\right)\vu_t(\ell)'\right\|$, and 
	\begin{align*}
	\Rmnum{2}_c&=\max_{1\leq\ell\leq q}\left\|\frac{1}{T-\ell}\sum_{t=\ell}^{T-1}\left(\widehat{\veta}_{t+1,\ell}-\veta_{t+1,\ell}\right)\vu_t(\ell)'\right\|,\\
	\Rmnum{2}_d&=\max_{1\leq\ell\leq q}\left\|\frac{1}{T-\ell}\sum_{t=\ell}^{T-1}\veta_{t+1,\ell}\left(\widehat{\vu}_t(\ell)-\vu_t(\ell)\right)'\right\|,\\
	\Rmnum{2}_e&=\max_{1\leq\ell\leq q}\left\|\frac{1}{T-\ell}\sum_{t=\ell}^{T-1}\left(\widehat{\veta}_{t+1,\ell}-\veta_{t+1,\ell}\right)\left(\widehat{\vu}_t(\ell)-\vu_t(\ell)\right)'\right\|.
	\end{align*}
	We consider these terms separately starting from $\Rmnum{2}_a$. Using the properties of Frobenius norm,
	\begin{equation*}
	\left\|\frac{1}{T-\ell}\sum_{t=\ell}^{T-1}\veta_{t+1}\vu_t(\ell)'\right\|^2\leq \sum_{s=0}^{\ell-1}\left\|\frac{1}{T-\ell}\sum_{t=\ell}^{T-1}\veta_{t+1}^{}\vu_{t-s}'\right\|_\calF^2=\frac{1}{(T-\ell)^{2}}\sum_{s=0}^{\ell-1}\sum_{i,j=1}^{n}\left|\sum_{t=\ell}^{T-1}\eta_{i,t+1}u_{j,t-s}\right|^2.
	\end{equation*}
	Assumption \ref{assumpt1:linearproc} justifies the use of Lemma 2 in \cite{wei1987} which gives $\E\left|\sum_{t=\ell}^{T-1}\eta_{i,t+1}u_{j,t-s}\right|^2\leq C\sum_{t=\ell}^{T-1}\E \left( u_{j,t-s}^2\right)\leq C(T-\ell)$. By Chebyshev's inequality, $\forall\varepsilon>0$, there exists $\alpha_\varepsilon>0$ such that
	\begin{equation*}
	\p\left(\Rmnum{2}_a\geq \alpha_\varepsilon\frac{q}{\sqrt{T}}\right)\leq \frac{1}{\alpha_\varepsilon^2}\frac{T}{q^2}\sum_{\ell=1}^{q}\E\left\|\frac{1}{T-\ell}\sum_{t=\ell}^{T-1}\veta_{t+1}\vu_t(\ell)'\right\|^2\leq \frac{C}{\alpha_\varepsilon^2}\leq  \varepsilon,
	\end{equation*} 
and thus $\Rmnum{2}_a=O_p\left(q/\sqrt{T}\,\right)$. Furthermore, we deduce that $\Rmnum{2}_b=O_p\left( q/ \sqrt{T}\,\right)$ by Lemma \ref{lem:fmbt}$(\rmnum{2})$ and Chebyshev's inequality. For $\Rmnum{2}_c$, if we write $\widehat{\veta}_{t+1,\ell}-\veta_{t+1,\ell}=\big[\mI_n,-\mA(\ell)\big]\big[\widehat{\vu}_{t+1}(\ell+1)-\vu_{t+1}(\ell+1)\big]$, then by Cauchy-Schwarz inequality and Baxter's inequality (leads to $\max_{1\leq\ell\leq q}\left\|\mA(\ell)\right\|^2\leq C$), 
	\begin{equation*}
	\Rmnum{2}_c\leq C\sqrt{\max_{1\leq\ell\leq q}\frac{1}{T-\ell}\sum_{t=\ell}^{T-1}\left\|\widehat{\vu}_{t+1}(\ell+1)-\vu_{t+1}(\ell+1)\right\|^2}\sqrt{\max_{1\leq\ell\leq q}\frac{1}{T-\ell}\sum_{t=\ell}^{T-1}\left\|\vu_t(\ell)\right\|^2}= O_p\left(\frac{q}{\sqrt{T}}\right),
	\end{equation*}
	where the last step follows from arguments similar to those preceding \eqref{eq:invertibility}. Similarly, $\Rmnum{2}_d=O_p\left( q/\sqrt{T}\,\right)$ and $\Rmnum{2}_e=O_p\left(q/T\,\right)$. Combining all results, we finally have 
	\begin{equation}\label{eq:diff_mA_decomposition_term2}
	\max_{1\leq\ell\leq q}\left\|\frac{1}{T-\ell}\sum_{t=\ell}^{T-1}\widehat{\veta}_{t+1,\ell}\widehat{\vu}_t(\ell)'\right\|=O_p\left(\frac{q}{\sqrt{T}}\right).
	\end{equation}
	By invertibility of $\frac{1}{T-\ell}\sum_{t=\ell}^{T-1}\widehat{\vu}_t(\ell)\widehat{\vu}_t(\ell)'$, \eqref{eq:diff_mA_decomposition} and \eqref{eq:diff_mA_decomposition_term2}, $\max_{1\leq\ell\leq q}\big\|\widehat{\mA}(\ell)-\mA(\ell)\big\| = O_p\left(q/\sqrt{T}\, \right)$ follows. 
	
We continue with $\max_{1\leq\ell\leq q}\big\|\widehat{\mS}(\ell)-\mS(\ell)\big\| = O_p\left(q/\sqrt{T}\, \right)$. Since $\widetilde{\veta}_{t+1,\ell}=\widehat{\veta}_{t+1,\ell}-\left(\widehat{\mA}(\ell)-\mA(\ell)\right)\widehat{\vu}_t(\ell)$ (see \eqref{eq:def_reverse_residuals}), we can use \eqref{eq:diff_mA_decomposition} and the invertibility of $\frac{1}{T-\ell}\sum_{t=\ell}^{T-1}\widehat{\vu}_t(\ell)\widehat{\vu}_t(\ell)'$ to write
	\begin{align}
	&\max_{1\leq\ell\leq q}\big\|\widehat{\mS}(\ell)-\mS(\ell)\big\|=\max_{1\leq\ell\leq q}\left\|\frac{1}{T-\ell}\sum_{t=\ell}^{T-1}\left(\widetilde{\veta}_{t+1,\ell}^{}\widetilde{\veta}_{t+1,\ell}'\right)-\E\left(\veta_{t+1,\ell}^{}\veta_{t+1,\ell}'\right)\right\|\nonumber\\
	&=\max_{1\leq\ell\leq q}\left\|\frac{1}{T-\ell}\sum_{t=\ell}^{T-1}\left(\widehat{\veta}_{t+1,\ell}^{}\widehat{\veta}_{t+1,\ell}'\right)-\E\left(\veta_{t+1,\ell}^{}\veta_{t+1,\ell}'\right)-\frac{1}{T-\ell}\sum_{t=\ell}^{T-1}\widehat{\veta}_{t+1,\ell}\widehat{\vu}_t(\ell)'\left(\widehat{\mA}(\ell)-\mA(\ell)\right)'\right\|\nonumber\\
	&\leq \max_{1\leq\ell\leq q}\left\|\frac{1}{T-\ell}\sum_{t=\ell}^{T-1}\left(\widehat{\veta}_{t+1,\ell}^{}\widehat{\veta}_{t+1,\ell}'\right)-\E\left(\veta_{t+1,\ell}^{}\veta_{t+1,\ell}'\right)\right\|+C\max_{1\leq\ell\leq q}\left\|\frac{1}{T-\ell}\sum_{t=\ell}^{T-1}\widehat{\veta}_{t+1,\ell}\widehat{\vu}_t(\ell)'\right\|^2\nonumber\\
	&\leq \max_{1\leq\ell\leq q}\left\|\frac{1}{T-\ell}\sum_{t=\ell}^{T-1}\left(\widehat{\veta}_{t+1,\ell}^{}\widehat{\veta}_{t+1,\ell}'-\veta_{t+1,\ell}^{}\veta_{t+1,\ell}'\right)\right\|+O_p\left(\sqrt{\frac{q}{T}}\right) + O_p\left(\frac{q^2}{T}\right),\label{eq:diff_mS_decomposition}
	\end{align}
	where the last step in \eqref{eq:diff_mS_decomposition} follows from $\max_{1\leq\ell\leq q}\left\|\frac{1}{T-\ell}\sum_{t=\ell}^{T-1}\left(\veta_{t+1,\ell}^{}\veta_{t+1,\ell}'\right)-\E\left(\veta_{t+1,\ell}^{}\veta_{t+1,\ell}'\right)\right\|=O_p\left(\sqrt{ q/T}\right)$ using Lemma \ref{lem:fmbt}$(\rmnum{3})$ and \eqref{eq:diff_mA_decomposition_term2}. By the inequality $\left\|\sum_t\va_t^{}\va_t'-\sum_t\vb_t^{}\vb_t'\right\|\leq \sum_t\left\|\va_t-\vb_t\right\|^2+2\sqrt{\sum_t\left\|\va_t-\vb_t\right\|^2}\sqrt{\sum_t\left\|\vb_t\right\|^2}$ and similar arguments as for $\Rmnum{2}_c$ and $\Rmnum{2}_d$ above, the first term in \eqref{eq:diff_mS_decomposition} is bounded by
	\begin{multline*}
	\max_{1\leq\ell\leq q}\frac{1}{T-\ell}\sum_{t=\ell}^{T-1}\left\|\widehat{\veta}_{t+1,\ell}-\veta_{t+1,\ell}\right\|^2\\
	+2\sqrt{\max_{1\leq\ell\leq q}\frac{1}{T-\ell}\sum_{t=\ell}^{T-1}\left\|\widehat{\veta}_{t+1,\ell}-\veta_{t+1,\ell}\right\|^2}\sqrt{\max_{1\leq\ell\leq q}\frac{1}{T-\ell}\sum_{t=\ell}^{T-1}\left\|\veta_{t+1,\ell}\right\|^2}= O_p\left(\frac{q}{\sqrt{T}}\right).
	\end{multline*}
	Overall, we obtain $\max_{1\leq\ell\leq q}\big\|\widehat{\mS}(\ell)-\mS(\ell)\big\| = O_p\left(q/\sqrt{T}\, \right)$ as well.	
\end{proof}

\begin{lemma}\label{lem:biam_consistency_part2}
	Under Assumptions \ref{assumpt1:linearproc} and \ref{assump4:bandingparameter}, we have
	\begin{equation}
	\left\|\mSigma_{\vu}^{-1}(q)-\mSigma_{\vu}^{-1}\right\|\leq C\frac{1}{\sqrt{q}}\sum_{s=q+1}^{\infty}s\left\|\mA_{s}\right\|_{\calF}.
	\end{equation}
\end{lemma}
\begin{proof}
Consider $\left\|\mSigma_{\vu}^{-1}(q)-\mSigma_{\vu}^{-1}\right\|$. A rewriting as in \eqref{eq:decom_partI} shows that $\big\|\MChol_{\vu}(q)-\MChol_{\vu}\big\|$ and $\big\|\SChol_{\vu}^{-1}(q)-\SChol_{\vu}^{-1}\big\|$ are the two important terms to bound. H\"{o}lder's inequality implies
	\begin{equation}
	\big\|\MChol_{\vu}(q)-\MChol_{\vu}\big\|\leq\sqrt{\big\|\MChol_{\vu}(q)-\MChol_{\vu}\big\|_1~\big\|\MChol_{\vu}(q)-\MChol_{\vu}\big\|_\infty} \,.
	\label{eq:holdereq}
	\end{equation}
	For the matrix $1$-norm we are concerned with the maximum absolute column sum. For an arbitrary $(nT\times nT)$ matrix $\mQ$ partitioned (block) column-wise, i.e. $\mQ=[\mQ_1,\mQ_2,\ldots,\mQ_T]$, we have the bound $\|\mQ\|_1=\max_{1\leq t\leq T}\|\mQ_t\|_1\leq \sqrt{n}\max_{1\leq t\leq T}\|\mQ_t\|_\calF$. This implies 
	$\big\|\MChol_{\vu}(q)-\MChol_{\vu}\big\|_1\leq \sqrt{n} \max\left\{\Rmnum{1}_a,\Rmnum{1}_b\right\}$ where
	\begin{equation*}
	\begin{aligned}
	\Rmnum{1}_a&=\max_{0\leq j\leq T-q-2}\left(\sum_{i=0}^{T-q-2-j}\left\|\mA_{q+1+i}(q+1+i+j)\right\|_{\calF}^2+\sum_{i=\max(1,q+1-j)}^{q}\left\|\mA_{i}(i+j)-\mA_{i}(q)\right\|_{\calF}^2\right)^{1/2},\\
	\Rmnum{1}_b&=\max_{T-q-1\leq j\leq T-2}\left(\sum_{i=1}^{T-1-j}\left\|\mA_{i}(i+j)-\mA_{i}(q)\right\|_{\calF}^2\right)^{1/2}.
	\end{aligned}
	\end{equation*}
	We will bound the three summations that are encountered in the expressions for $\Rmnum{1}_a$ and $\Rmnum{1}_b$. First, changing the summation index and using $c_r$-inequality,
	\begin{align*}
	\sum_{i=0}^{T-q-2-j} \big\| \mA_{q+1+i}(q+1+i+j)\big\|_{\calF}^2
	&=\sum_{s=q+1}^{T-1-j}\left\|\mA_{s}(s+j)\right\|_{\calF}^2\nonumber\\
	&\leq 2\sum_{s=q+1}^{T-1-j}\left\|\mA_{s}(s+j)-\mA_{s}\right\|_{\calF}^2+2\sum_{s=q+1}^{T-1-j}\left\|\mA_{s}\right\|_{\calF}^2\nonumber\\
	&\leq  2 \sum_{s=q+1}^{T-1-j} s^{-2} \left(s^2 \sum_{k=1}^{s+j} \left\|\mA_{k}(s+j)-\mA_{k}\right\|_{\calF}^2\right) + 2\sum_{s=q+1}^\infty \left\|\mA_{s}\right\|_{\calF}^2.
	\end{align*}
	For convenience, we define
	\begin{equation}
	\calK_q=\left(\sum_{s=q+1}^{\infty}s\left\|\mA_{s}\right\|_{\calF}\right)^2\left(\sum_{s=q+1}^{\infty}\frac{1}{s^2}\right).
	\end{equation}
	For any $j\geq 0$ and $s\geq q+1$, we have $s^2\sum_{k=1}^{s+j}\left\|\mA_{k}(s+j)-\mA_{k}\right\|_{\calF}^2\leq Cs^2\sum_{k=s+j+1}^{\infty}\left\|\mA_{k}\right\|_{\calF}^2\leq C\left(\sum_{k=s+j+1}^{\infty}k\left\|\mA_{k}\right\|_{\calF}\right)^2$ by the $L^2$-Baxter's inequality. The first term in the RHS above is thus bounded by $C\calK_q$. Moreover, by Cauchy-Schwartz inequality, the second term can be bounded as $\sum_{s=q+1}^{\infty}\left\|\mA_{s}\right\|_{\calF}^2 = \sum_{s=q+1}^{\infty}  \big(s^2 \left\|\mA_{s}\right\|_{\calF}^2 \big) s^{-2}\leq \left(\sum_{s=q+1}^{\infty}s^2\left\|\mA_{s}\right\|_{\calF}^2\right) \left(\sum_{s=q+1}^{\infty} s^{-2}\right) \leq \calK_q$. Now the second summation in $\Rmnum{1}_a$. We first consider the case $0\leq j \leq q$, or $\max(1,q+1-j)=q+1-j$, such that
	\begin{align*}
	\sum_{i=\max(1,q+1-j)}^{q} \left\|\mA_{i}(i+j)-\mA_{i}(q)\right\|_{\calF}^2&= \sum_{i=q+1-j}^{q}\left\|\mA_{i}(i+j)-\mA_{i}(q)\right\|_{\calF}^2\\
	&\leq 2 \sum_{i=q+1-j}^{q}\left\|\mA_{i}(i+j)-\mA_{i}\right\|_{\calF}^2+ 2\sum_{i=q+1-j}^{q}\left\|\mA_{i}(q)-\mA_{i}\right\|_{\calF}^2 \\
	&\leq  2 \sum_{s=q+1}^{q+j}\left\|\mA_{s-j}(s)-\mA_{s-j}\right\|_{\calF}^2 + 2\sum_{i=1}^{q}\left\|\mA_{i}(q)-\mA_{i}\right\|_{\calF}^2
	\leq C\calK_q
	\end{align*}
	using arguments detailed before. This upper bound remains valid for $q+1 \leq j \leq T-q-2$. It is likewise straightforward to derive $\sum_{i=1}^{T-1-j}\left\|\mA_{i}(i+j)-\mA_{i}(q)\right\|_{\calF}^2\leq C\calK_q$. Collecting all the results, we have $\Rmnum{1}_a\leq C\sqrt{\calK_q}$, $\Rmnum{1}_b\leq C\sqrt{\calK_q}$, and thus $\big\|\FChol_{\vu}(q)-\FChol_{\vu}\big\|_1\leq C\sqrt{\calK_q}$. 
	
	For $\|\MChol_{\vu}(q)-\MChol_{\vu}\|_\infty$, we are bounding the maximum absolute row sums. For an arbitrary $(nT\times nT)$ matrix $\mQ$ partitioned as $\mQ=[\mQ_1',\mQ_2',\ldots,\mQ_T']'$, we have $\|\mQ\|_\infty=\max_{1\leq t\leq T}\|\mQ_t\|_\infty\leq \sqrt{n}\max_{1\leq t\leq T}\|\mQ_t\|_\calF$, such that
	\begin{equation*}
	\|\MChol_{\vu}(q)-\MChol_{\vu}\|_\infty
	\leq \sqrt{n}\max_{q+1\leq m\leq T-1}\Bigg\{ \sum_{j=q+1}^{m}\left\|\mA_j(m)\right\|_{\calF}^2+\sum_{j=1}^{q}\left\|\mA_j(q)-\mA_j(m)\right\|_{\calF}^2 \Bigg\}^{1/2},
	\end{equation*}
	where $\sum_{j=q+1}^{m}\left\|\mA_j(m)\right\|_{\calF}^2\leq C\calK_q$ and $\sum_{j=1}^{q}\left\|\mA_j(q)-\mA_j(m)\right\|_{\calF}^2\leq C\calK_q$, for any $q+1\leq m \leq T-1$, using the $L^2$-Baxter's inequality and the previous upper bound on $\sum_{s=q+1}^{\infty}\left\|\mA_{s}\right\|_{\calF}^2$. We conclude that $\|\MChol_{\vu}(q)-\MChol_{\vu}\|_\infty\leq C\sqrt{\calK_q}$. Together with our previous we result, we obtain $\big\|\MChol_{\vu}(q)-\MChol_{\vu}\big\|\leq C\sqrt{\calK_q}$ from \eqref{eq:holdereq}.
	
	From $\big\|\SChol_{\vu}^{-1}(q)-\SChol_{\vu}^{-1}\big\|\leq \big\|\SChol_{\vu}(q)-\SChol_{\vu}\big\|~\big\|\SChol_{\vu}^{-1}\big\| ~ \big\|\SChol_{\vu}^{-1}(q)\big\|$ we see that it suffices to inspect $\big\|\SChol_{\vu}(q)-\SChol_{\vu}\big\|$ (the other norms are bounded). Exploiting the fact that both $\SChol_{\vu}(q)$ and $\SChol_{\vu}$ are block-diagonal, we have $\big\|\SChol_{\vu}(q)-\SChol_{\vu}\big\|=\max_{q+1\leq k\leq T-1}\left\|\mS(q)-\mS(k)\right\|
	\leq 2 \max_{q \leq k\leq T-1}  \left\|\mS(k)-\mSigma_{\eta\eta}\right\|$. Let $\mA_j(\ell)=\mZeros$ for $j>\ell$, and recall the definition of $\veta_{t+1,\ell}$ in \eqref{eq:def_reverse_errors}. We find, for any $k\geq q$,
	\begin{align}
	\left\|\mS(k)-\mSigma_{\eta\eta}\right\|&=\left\|\E\left(\veta_{t+1,k}-\veta_{t+1}\right)\left(\veta_{t+1,k}-\veta_{t+1}\right)'\right\|\nonumber\\
	&\leq C\sum_{s=1}^{\infty}\left\|\mA_s(k)-\mA_s\right\|^2\leq C\sum_{s=q+1}^{\infty}\left\|\mA_{s}\right\|_{\calF}^2\leq C\calK_q.\label{eq:convergence_mSq}
	\end{align}
	We thereby obtain $\big\|\SChol_{\vu}^{-1}(q)-\SChol_{\vu}^{-1}\big\|\leq C\calK_q$. Together with the bound on $\big\|\MChol_{\vu}(q)-\MChol_{\vu}\big\|$, we deduce
	\begin{equation*}
	\left\|\mSigma_{\vu}^{-1}(q)-\mSigma_{\vu}^{-1}\right\|\leq C\sqrt{\calK_q}=C\left(\sum_{s=q+1}^{\infty}s\left\|\mA_{s}\right\|_{\calF}\right)\sqrt{\sum_{s=q+1}^{\infty}\frac{1}{s^2}}\leq C\frac{1}{\sqrt{q}}\sum_{s=q+1}^{\infty}s\left\|\mA_{s}\right\|_{\calF}.
	\end{equation*}
    The proof is complete.
\end{proof}

\begin{theorem}\label{thm:cdf_kpss}
	Let $\mW(r)=[W_1(r),W_2(r),\ldots,W_n(r)]'$ denote an $n$-dimensional standard Brownian motion. The cumulative density function (CDF) of $\int_{0}^{1}\left\|\mW(r)\right\|^2dr$ is given by
	\begin{equation*}
	F_{n}(x)=2^{n/2}\sum_{j=0}^{\infty}k_{j,n}^{}\,\mathrm{Erfc}\left(\frac{l_{j,n}}{2\sqrt{x}}\right),\qquad x > 0,
	\label{eq:cdf_Wn}
	\end{equation*}
	where $k_{j,n}= (-1)^j \, \frac{\Gamma(j+n/2)}{j! \Gamma(n/2)}$, $l_{j,n}^{}=2\sqrt{2}j+\frac{n}{\sqrt{2}}$ and $\mathrm{Erfc}(x)=\frac{2}{\sqrt{\pi}}\int_{x}^{\infty}e^{-t^2}dt$.
\end{theorem}
\begin{proof}
	We follow the approach from Example 1 of \cite{andersondarling1952} or equivalently appendix B of \cite{choisaikkonen2010}. Let $f_{n}$ denote the probability density function of $\int_{0}^{1}\left\|\mW(r)\right\|^2dr$ and write $\calL\{ \cdot \}$ and $\calL^{-1}\{\cdot\}$ for the Laplace and inverse Laplace operator, respectively. From the equality $\int_{0}^{1}\left\|\mW(r)\right\|^2dr=\sum_{i=1}^n \int_{0}^{1} W_i(r)^2 dr$, independence of the components of $\mW(r)$, and the known univariate result in \cite{choisaikkonen2010}, we have
	\begin{equation}
	\calL\{ f_{n}(x) \}(t)= \int_0^\infty e^{-tx} f_{n}(x)  dx = \left[ \cosh\left(\sqrt{2t}\, \right) \right]^{-n/2}.
	\end{equation}
	According to equation (4.28) in \cite{andersondarling1952}, the CDF is
	\begin{equation}
	\begin{aligned}
	F_{n}(x) &= \calL^{-1}\left\{ \frac{1}{t} \left[ \cosh\left(\sqrt{2 t} \,\right) \right]^{-n/2} \right\}(x)=  \calL^{-1}\left\{\frac{1}{t} \left(\dfrac{e^{\sqrt{2t}}}{2}\right)^{-n/2} \left[1+e^{-2\sqrt{2t}} \right]^{-n/2} \right\}(x) \\
	&= \calL^{-1}\left\{\frac{1}{t} \left(\dfrac{e^{\sqrt{2t}}}{2}\right)^{-n/2} \sum_{j=0}^\infty k_{j,n} e^{-2 j\sqrt{2t}}  \right\}(x)= 2^{n/2} \sum_{j=0}^\infty k_{j,n} \calL^{-1}\left\{\frac{1}{t} e^{-l_{j,n} \sqrt{t}}\right\}(x)
	\end{aligned}
	\end{equation}
	where we use (1) a $t$ with a positive real part, (2) linearity of the inverse Laplace operator, and (3) the binomial expansion of $[1+e^{-2\sqrt{2t}}]^{-n/2}$. The identity from \cite{choisaikkonen2010}, $\calL^{-1}\left\{\frac{1}{t} e^{-u \sqrt{t}}\right\}(x)=1-\mathrm{Erf}\left(\frac{u}{2\sqrt{x}}\right)=\mathrm{Erfc}\left(\frac{u}{2\sqrt{x}}\right)$, completes the proof.
\end{proof}

\section{Estimation of Quantities for Fully Modified Inference} \label{detailsFMinference}
The FM-GLS estimator relies on $\mOmega$, $\mDelta$, and $\E(\vzeta_t^{} \vzeta_t\tran)$ (see Assumption \ref{assumpt1:linearproc}). For convenience, we denote this $(2n\times 2n)$ matrix 
$\left[\begin{smallmatrix}
\mSigma_{\eta\eta}     &  \mSigma_{\eta\epsilon}\\
\mSigma_{\epsilon\eta} &  \mSigma_{\epsilon\epsilon}
\end{smallmatrix}\right]$ by $\mSigma$. Please note the difference between $\mSigma$ and the large-dimensional matrix $\mSigma_{\vu}$. In this section, we consider the estimation of these three quantities within the BIAM framework. For conenience, we recall $\vxi_t=[\vu_t',\vv_t']'$ and define $\vxi=[\vxi_1',\vxi_2',\ldots,\vxi_T']'$. Similarly to the definition of $\mSigma_u$, we used $\mSigma_{\vxi}:=\E\left(\vxi\vxi'\right)$ to denote the $(2nT\times 2nT)$ autocovariance matrix of $\{\vxi_t\}$. As a sample counterpart, we stack $\widehat{\vu}_t$ and $\Delta\vx_t=\vv_t$ in the $2n$-dimensional vector $\widehat{\vxi}_t=\big[\widehat{\vu}_t',\vv_t'\big]'$. Using $\{\widehat{\vxi}_t \}_{t=1}^T$, the BIAM estimator for $\mSigma_{\vxi}$ is now constructed as $\widehat{\mSigma_{\vxi}}(q)=\widehat{\MChol}_{\vxi}^{-1}(q)\widehat{\SChol}_{\vxi}^{}(q)\widehat{\MChol}_{\vxi}^{-1\prime}(q)$, where the matrices $\widehat{\MChol}_{\vxi}(q)$ and $\widehat{\SChol}_{\vxi}(q)$ are defined similarly to $\widehat{\MChol}_{\vu}(q)$ and $\widehat{\SChol}_{\vu}(q)$, respectively. Since the BIAM estimator is fitting VAR processes up to order $q_T$ (see \eqref{eq:sampleMCD}), the coefficient estimates $\widehat{\mF}_j(q_T)$ of the $j^{th}$ lag when a VAR$(q_T)$ is fitted, $j=1,2,\dots,q_T$, are immediate byproducts of the BIAM procedure and can thus be used to construct our estimators. Finally, if $\lagpol{F}(L)=\diag\left[\lagpol{A}(L),\lagpol{D}(L)\right]:=\mI_{2n}-\sum_{j=1}^{\infty}\mF_jL^j$, where $\mF_j=\diag\big[\mA_j,\mD_j\big]$, then $\lagpol{F}(L)\vxi_t=\vzeta_t$ holds.

\begin{theorem}\label{thm:fmquantities}
Recall the definitions $\widehat{\mOmega}_{q_T}=\left(\mI_{2n}-\sum_{j=1}^{q_T}\widehat{\mF}_j(q_T)\right)^{-1}\widehat{\mSigma}_{q_T}\left(\mI_{2n}-\sum_{j=1}^{q_T}\widehat{\mF}_j(q_T)\right)^{-1\prime}$, $\widehat{\mDelta}_{q_T,r_T}=\mQ_{r_T}'\widehat{\mSigma_{\vxi}}(q_T)\mQ_1^{}$, and
	\begin{equation*}
	\widehat{\mSigma}_{q_T}=\frac{1}{T-q_T}\sum_{t=q_T+1}^{T}\big[\widehat{\vxi}_t-\widehat{\mF}_1(q_T) \widehat{\vxi}_{t-1}-\cdots-\widehat{\mF}_{q_T}(q_T)\widehat{\vxi}_{t-q_T}\big]\big[\widehat{\vxi}_t-\widehat{\mF}_1(q_T) \widehat{\vxi}_{t-1}-\cdots-\widehat{\mF}_{q_T}(q_T)\widehat{\vxi}_{t-q_T}\big]',
	\end{equation*}
	where $\mQ_r=\left[\mZeros_{2n \times 2n},\cdots,\mZeros_{2n \times 2n},\mI_{2n},\cdots,\mI_{2n}\right]'$ is an $(2nT\times 2n)$ block matrices of zeros of which the last $r$ blocks have been replaced by identity matrices. If Assumptions \ref{assumpt1:linearproc}-\ref{assump4:bandingparameter} and \ref{assumpt:rT} hold, then
	\begin{align}
	\left\|\widehat{\mSigma}_{q_T}-\mSigma\right\|&=O_p\left(\frac{q_T}{\sqrt{T}}\right)+O\left(\sum_{s=q_T+1}^{\infty}\|\mF_s\|_{\calF}^2\right)=o_p(1),\label{eq:shortrun_var}\\
	\left\|\widehat{\mOmega}_{q_T}-\mOmega\right\|&=O_p\left(\sqrt{\frac{q_T^3}{T}}+\frac{1}{q_T}\sum_{s=q_T+1}^{\infty}s\|\mF_s\|_{\calF}\right)=o_p(1),\label{eq:lrv}\\
	\left\|\widehat{\mDelta}_{q_T,r_T}-\mDelta\right\|&=o(r_T^{-1})+O_p\left(\sqrt{\frac{q_T^3}{T}r_T}+\sqrt{\frac{r_T}{q_T}}\sum_{s=q_T+1}^{\infty}s\|\mF_s\|_{\calF}\right)=o_p(1).\label{eq:onesided_lrv}
	\end{align} 
\end{theorem}


\begin{proof}
Note that $\widehat{\vxi}_t^{}-\vxi_t^{}=[(\widehat{\vu}_t-\vu_t)',\vzeros']'$ and hence $\|\widehat{\vxi}-\vxi\|^{2}=\|\widehat{\vu}-\vu\|^{2}=O_p(1)$ by Assumption \ref{assump3:residuals}. The conditions for Lemmas \ref{lem:max_bound} -- \ref{lem:biam_consistency_part2} and Theorem \ref{thm:consistentDECOMP} are thus satisfied and we can use these results in subsequent proofs. \textbf{(a)} The result \eqref{eq:shortrun_var} follows from the triangle inequality, Lemma \ref{lem:max_bound} and \eqref{eq:convergence_mSq}. \textbf{(b)} The second result \eqref{eq:lrv} is obtained by the definition $\mOmega=(\mI_{2n}-\sum_{j=1}^\infty \mF_j)^{-1}\mSigma(\mI_{2n}-\sum_{j=1}^\infty\mF_j)^{-1\prime}$, Lemma \ref{lem:max_bound} and a straightforward modification of \eqref{eq:decom_partI}. \textbf{(c)} By $\mDelta=\sum_{h=r_T}^{\infty}\E\big(\vxi_t^{}\vxi_{t+h}'\big)+\mQ_{r_T}'\mSigma_{\vxi}\mQ_1^{}$, the LHS of \eqref{eq:onesided_lrv} can be bounded 
	\begin{equation*}
	\left\|\widehat{\mDelta}_{q_T,r_T}-\mDelta\right\|\leq \sum_{h=r_T}^{\infty}\left\|\E\big(\vxi_t^{}\vxi_{t+h}'\big)\right\|_{\calF}+\big\|\mQ_{r_T}'\mSigma_{\vxi}\big\|\,\big\|\widehat{\mSigma_{\vxi}^{-1}}(q_T)-\mSigma_{\vxi}^{-1}\big\|\,\big\|\widehat{\mSigma_{\vxi}}(q_T)\mQ_1\big\|.
	\end{equation*}
	Since summability conditions on the coefficient matrices carry over to the autocovariances, we have $\sum_{h=r_T}^{\infty}\left\|\E\big(\vxi_t^{}\vxi_{t+h}'\big)\right\|_{\calF}\leq r_T^{-1}\sum_{h=r_T}^{\infty}h\left\|\E\big(\vxi_t^{}\vxi_{t+h}'\big)\right\|_{\calF}=o(r_T^{-1})$ by Assumption \ref{assumpt1:linearproc}. Moreover, $\left\|\mQ_{r_T}'\mSigma_{\vxi}\right\|\leq C\sqrt{r_T}$ and $\big\|\widehat{\mSigma_{\vxi}^{-1}}(q_T)-\mSigma_{\vxi}^{-1}\big\|$ is discussed in Theorem \ref{thm:consistentDECOMP}. Finally, showing $\big\|\widehat{\mSigma_{\vxi}}(q_T)\mQ_1\big\|=O_p(1)$ will complete the proof after a straightforward comparison of the established stochastic orders. It suffices to prove $\big\|\widehat{\mSigma_{\vxi}}(q_T)\big\|=O_p(1)$. Weyl's inequality (e.g. pages 40 and 46 in \cite{tao2012}) and Theorem \ref{thm:consistentDECOMP} imply
	\begin{equation*}
	\left|\lambda_{min}\Big(\widehat{\mSigma_{\vxi}^{-1}}(q_T)\Big)-\lambda_{min}\left(\mSigma_{\vxi}^{-1}\right)\right|\leq \big\|\widehat{\mSigma_{\vxi}^{-1}}(q_T)-\mSigma_{\vxi}^{-1}\big\|=o_p(1).
	\end{equation*}
	By the uniform boundedness of $\big\|\mSigma_{\vxi}\big\|$, for a sufficiently large $T$, there exists a constant $C>0$ such that $\big\|\widehat{\mSigma_{\vxi}}(q_T)\big\|^{-1}=\lambda_{min}\Big(\widehat{\mSigma_{\vxi}^{-1}}(q_T)\Big)\leq C$ and thus $\big\|\widehat{\mSigma_{\vxi}}(q_T)\big\|\leq C^{-1}$ with arbitrarily high probability.	
\end{proof}

\section{Additional information for the Empirical Application} \label{empiricalillustration}

\subsection{Model fit}

\begin{figure}[h!]
	\centering
	\includegraphics[width=0.9\textwidth]{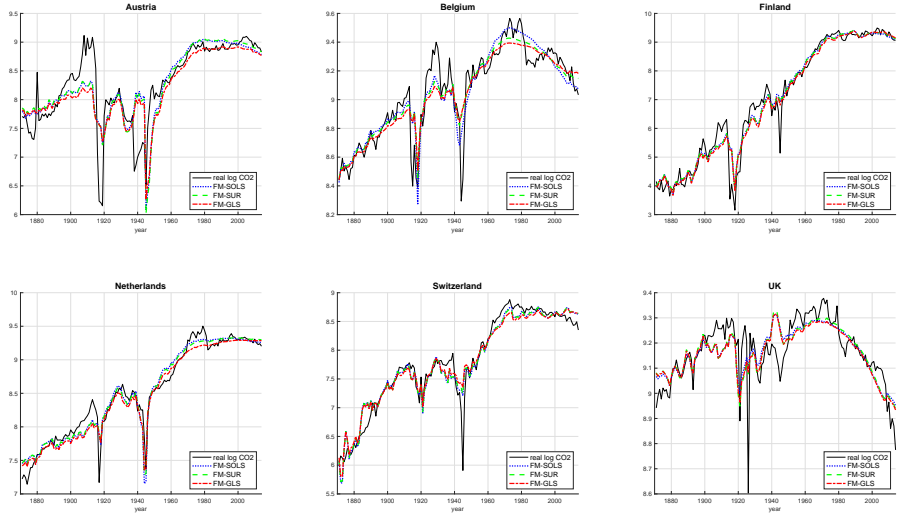}
	\caption{\footnotesize The fit of the FM-SOLS, FM-SUR, and FM-GLS estimates.}
	\label{fig:fitmethods}
\end{figure}

\subsection{Simulation DGP}
The following procedure was used to obtain a simulation DGP that closely mimics the data characteristics.
 \begin{enumerate}[(a)]
  \item Fit VAR($p$) models ($1\leq p \leq 8$) to the series $\{\hat \vu_{t,FGLS}\}$ and $\{\diff \vx_t \}$ individually. The BIC criterion select the VAR(1) specification for both series (Table \ref{tab:BICresults}). Store the coefficient matrices $\widehat{\mA}_{u}$ and $\widehat{\mA}_{v}$ as well as the residual series $\{\hat{\veta}_t\}$ and $\{\hat{\vepsi}_t\}$, respectively.
  \item Stack $\hat \vzeta_t = [\hat\veta_t^{\prime}, \hat\vepsi_t\tran]\tran$ and compute $\widehat{\mSigma} = \frac{1}{T} \sum_{t=2}^T \hat\vzeta_t \hat \vzeta_t\tran$.
  \item Denoting the estimated coefficients from the data by $\widehat{\vbeta}_{FGLS}^+$, we generate the new data according to the following equations:
  $$
  \begin{aligned}
    \vy_t &=\mZ_t\tran \widehat{\vbeta}_{FGLS}^+ +\vu_t \\
    \diff \vx_t &= \vv_t,\text{ with }\vx_0=\vzeros \\
    \vxi_t &=
    \begin{bmatrix}
     \vu_t \\
     \vv_t
    \end{bmatrix}
    = 
    \begin{bmatrix}
     \widehat{\mA}_{u}			& \mZeros \\
     \mZeros				& \widehat{\mA}_{v}
    \end{bmatrix}
    \vxi_{t-1} + \vzeta_t,\text{ with } \vzeta_t\stackrel{i.i.d.}{\sim} \rN\left( \vzeros , \widehat{\mSigma} \right).
  \end{aligned}
  $$
 \end{enumerate}
 
 \begin{table}[h!]
	\centering
	\caption{The numerical values for the BIC criterion.}
	\label{tab:BICresults}
	\begin{tabular}{c  c c c c c c c c }
	\toprule
						& $p=1$			& $p=2$			& $p=3$			& $p=4$			& $p=5$			& $p=6$ 			& $p=7$		& $p=8$ \\
	\midrule
	$\{\hat \vu_{t,FGLS}\}$	& -23.88			& -23.49			& -22.75			& -21.87			& -21.25			& -20.66			& -20.11		& -19.27\\
	$\{\diff \vx_t \}$			& -37.94			& -37.78			& -37.15			& -36.44			& -35.68			& -35.03			& -34.34		& -33.64\\
	\bottomrule
	\end{tabular}
 \end{table}
 
 \section{Details on Implementation}
The implementation of the BIAM estimator and the subsampling KPSS tests requires selecting the banding parameter $q$ and the block length $b$. In our simulations and empirical application, we follow the subsampling and risk-minimization approach previously used by \cite{bickellevina2008}, \cite{wupourahmadi2009} and \cite{ingchiouguo2016} to select $q$. The steps are as follows:
\begin{enumerate}[Step 1]
	\item Split the series of (first-step OLS) residuals, $\{\widehat{\vu}_t\}_{t=1}^T$, into $J_0$ non-overlapping subsequences of length $l_0$. These subsequences are $\{\widehat{\vu}_t\}_{t=(j-1)l_0+1}^{jl_0}$ for $j=1,\dots,J_0$ with $J_0=[T/l_0]$.
	\item Select an integer $H$, $1\leq H <l_0$, and construct the $(nH\times nH)$ sample autocovariance matrix $\widehat{\mPi}_{\vu,nH}=\frac{1}{T-H}\sum_{t=H}^{T-1}\widehat{\vu}_t(H)^{}\widehat{\vu}_t(H)'$ which is an estimator of $\mSigma_{\vu,nH}:=\E\left[\vu(H)\vu(H)'\right]$ with $\vu(H):=\left[\begin{smallmatrix}
	\vu_1\\
	\vdots\\
	\vu_H
	\end{smallmatrix}\right]$, where $\widehat{\vu}_t(\ell)=\big(\widehat{\vu}_t',\cdots,\widehat{\vu}_{t-\ell+1}'\big)'$. Compute $\widehat{\mPi}_{\vu,nH}^{-1}$.
	\item For every subsequence of residuals $1\leq j \leq J_0$, compute the BIAM estimate of $\mSigma_{\vu,nH}$ repeatedly for all possible banding parameters $1\leq \bar{q}< H$, denoted as $\widehat{\mSigma_{\vu,nH}^{-1}}(\bar{q};j)$.
	\item Select the banding parameter that minimizes the feasible average risk, i.e.
	\begin{equation*}
	q:=\argmin_{\bar{q} \in[1,H)}\frac{1}{J_0}\sum_{j=1}^{J_0}\left\|\widehat{\mSigma_{\vu,nH}^{-1}}(\bar{q};j)-\widehat{\mPi}_{\vu,nH}^{-1}\right\|_p.
	\end{equation*}
\end{enumerate}
We take $p=1$, $H=[2 \times T^{1/4}]$ and $l_0=[T/5]$ and obtain satisfactory results for all the settings we have explored. As mentioned in \cite{bickellevina2008}, the use of another vector norm (e.g. $p=2$) does not lead to qualitatively different results. 

When we implement the minimum volatility rule as mentioned in Section 3 to select $b$, the values of tuning parameters are adopted from \cite{wagnerhong2016}
, see their online supplementary material.


\newpage
\begin{thebibliography}{}

\bibitem[\protect\citeauthoryear{Abadir and Magnus}{Abadir and
  Magnus}{2005}]{abadirmagnus2005}
Abadir, K.~M. and J.~R. Magnus (2005).
\newblock {\em Matrix Algebra}.
\newblock Cambridge University Press.


\bibitem[\protect\citeauthoryear{Anderson and Darling}{Anderson and
  Darling}{1952}]{andersondarling1952}
Anderson, T.~W. and D.~A. Darling (1952).
\newblock Asymptotic theory of certain ``goodness of fit'' criteria based on
  stochastic processes.
\newblock {\em The Annals of Mathematical Statistics\/}~{\em 23}, 193--212.


\bibitem[\protect\citeauthoryear{Andrews}{Andrews}{1991}]{andrews1991}
Andrews, D. W.~K. (1991).
\newblock Heteroskedasticity and autocorrelation consistent covariance matrix
  estimation.
\newblock {\em Econometrica\/}~{\em 59}, 817--858.


\bibitem[\protect\citeauthoryear{Berk}{Berk}{1974}]{berk1974}
Berk, K.~N. (1974).
\newblock Consistent autoregressive spectral estimates.
\newblock {\em The Annals of Statistics\/}~{\em 2}, 489--502.


\bibitem[\protect\citeauthoryear{Beutner, Lin, and Smeekes}{Beutner
  et~al.}{2019}]{beutnerlinsmeekes2019}
Beutner, E., Y.~Lin, and S.~Smeekes (2019).
\newblock {GLS} estimation and confidence sets for the date of a single break
  in models with trends.
\newblock Working Paper.

\bibitem[\protect\citeauthoryear{Bickel and Levina}{Bickel and
  Levina}{2008}]{bickellevina2008}
Bickel, P.~J. and E.~Levina (2008).
\newblock Regularized estimation of large covariance matrices.
\newblock {\em The Annals of Statistics\/}~{\em 36}, 199--227.


\bibitem[\protect\citeauthoryear{Boden, Marland, and Andres}{Boden
  et~al.}{2017}]{cdiac2017}
Boden, T., G.~Marland, and R.~Andres (2017).
\newblock Global, regional, and national fossil-fuel $\text{CO}_2$ emissions.
\newblock Carbon Dioxide Information Analysis Center, Oak Ridge National
  Laboratory, U.S. Department of Energy, Oak Ridge, Tenn., U.S.A.
  \url{http://cdiac.ess-dive.lbl.gov/trends/emis/tre_coun.html}.

\bibitem[\protect\citeauthoryear{Bolt, Inklaar, de~Jong, and van Zanden}{Bolt
  et~al.}{2018}]{madison2018}
Bolt, J., R.~Inklaar, H.~de~Jong, and J.~L. van Zanden (2018).
\newblock Rebasing "maddison": New income comparisons and the shape of long-run
  economic development.
\newblock
  \url{https://www.rug.nl/ggdc/historicaldevelopment/maddison/releases/maddison-project-database-2018}.

\bibitem[\protect\citeauthoryear{Breitung}{Breitung}{2001}]{breitung2001}
Breitung, J. (2001).
\newblock Rank tests for nonlinear cointegration.
\newblock {\em Journal of Business \& Economic Statistics\/}~{\em 19},
  331--340.


\bibitem[\protect\citeauthoryear{Chang, Park, and Phillips}{Chang
  et~al.}{2004}]{chang2004}
Chang, Y., J.~Y. Park, and P.~C.~B. Phillips (2004).
\newblock Bootstrap unit root tests in panels with cross-sectional dependency.
\newblock {\em Journal of Econometrics\/}~{\em 120}, 263--293.


\bibitem[\protect\citeauthoryear{Cheng and Pourahmadi}{Cheng and
  Pourahmadi}{1993}]{chengpourahmadi1993}
Cheng, R. and M.~Pourahmadi (1993).
\newblock Baxter's inequality and convergence of finite predictors of
  multivariate stochastic processes.
\newblock {\em Probability Theory and Related Fields\/}~{\em 95}, 115--124.


\bibitem[\protect\citeauthoryear{Cheng, Ing, and Yu}{Cheng
  et~al.}{2015}]{chengingyu2015}
Cheng, T.-C.~F., C.-K. Ing, and S.-H. Yu (2015).
\newblock Toward optimal model averaging in regression models with time series
  errors.
\newblock {\em Journal of Econometrics\/}~{\em 189}, 321--334.


\bibitem[\protect\citeauthoryear{Choi and Saikkonen}{Choi and
  Saikkonen}{2010}]{choisaikkonen2010}
Choi, I. and P.~Saikkonen (2010).
\newblock Tests for nonlinear cointegration.
\newblock {\em Econometric Theory\/}~{\em 26}, 682--709.


\bibitem[\protect\citeauthoryear{Davidson}{Davidson}{1994}]{davidson1994}
Davidson, J. (1994).
\newblock {\em Stochastic Limit Theory}.
\newblock Oxford University Press.


\bibitem[\protect\citeauthoryear{Davidson and MacKinnon}{Davidson and
  MacKinnon}{2004}]{davidsonmackinnon2004}
Davidson, R. and J.~G. MacKinnon (2004).
\newblock {\em Econometric Theory and Methods}.
\newblock Oxford University Press.


\bibitem[\protect\citeauthoryear{de~Jong}{de~Jong}{2002}]{dejong2002}
de~Jong, R.~M. (2002).
\newblock Nonlinear estimators with integrated regressors but without
  exogeneity.
\newblock mimeo Michigan State University.

\bibitem[\protect\citeauthoryear{Findley and Wei}{Findley and
  Wei}{1993}]{findleywei1993}
Findley, D.~F. and C.-Z. Wei (1993).
\newblock Moment bounds for deriving time series {CLT}'s and model selection
  procedures.
\newblock {\em Statistica Sinica\/}~{\em 3}, 453--480.


\bibitem[\protect\citeauthoryear{Grossman and Krueger}{Grossman and
  Krueger}{1995}]{grossmankrueger1995}
Grossman, G.~M. and A.~B. Krueger (1995).
\newblock Economic growth and the environment.
\newblock {\em The Quarterly Journal of Economics\/}~{\em 110}, 353--377.


\bibitem[\protect\citeauthoryear{Hamilton}{Hamilton}{1994}]{hamilton1994}
Hamilton, J.~D. (1994).
\newblock {\em Time Series Analysis}.
\newblock Princeton University Press.


\bibitem[\protect\citeauthoryear{Hannan and Deistler}{Hannan and
  Deistler}{2012}]{hannandeistler2012}
Hannan, E. and M.~Deistler (2012).
\newblock {\em The Statistical Theory of Linear Systems}.
\newblock Society for Industrial and Applied Mathematics.


\bibitem[\protect\citeauthoryear{Hong and Phillips}{Hong and
  Phillips}{2010}]{hongphillips2010}
Hong, S.~H. and P.~C.~B. Phillips (2010).
\newblock Testing linearity in cointegrating relations with an application to
  purchasing power parity.
\newblock {\em Journal of Business \& Economic Statistics\/}~{\em 28}, 96--114.


\bibitem[\protect\citeauthoryear{Ing, Chiou, and Guo}{Ing
  et~al.}{2016a}]{ICG2016}
Ing, C.-K., H.-T. Chiou, and M.~Guo (2016a).
\newblock Estimation of inverse autocovariance matrices for long memory
  processes.
\newblock {\em Bernoulli\/}~{\em 22}, 1301--1330.


\bibitem[\protect\citeauthoryear{Ing, Chiou, and Guo}{Ing
  et~al.}{2016b}]{ingchiouguo2016}
Ing, C.-K., H.-T. Chiou, and M.~Guo (2016b).
\newblock Estimation of inverse autocovariance matrices for long memory
  processes.
\newblock {\em Bernoulli\/}~{\em 22}, 1301--1330.


\bibitem[\protect\citeauthoryear{Jansson}{Jansson}{2002}]{jansson2002}
Jansson, M. (2002).
\newblock Consistent covariance matrix estimation for linear processes.
\newblock {\em Econometric Theory\/}~{\em 18}, 1449--1459.


\bibitem[\protect\citeauthoryear{Kim and Zimmerman}{Kim and
  Zimmerman}{2012}]{kimzimmerman2012}
Kim, C. and D.~L. Zimmerman (2012).
\newblock Unconstrained models for the covariance structure of multivariate
  longitudinal data.
\newblock {\em Journal of Multivariate Analysis\/}~{\em 107}, 104--118.


\bibitem[\protect\citeauthoryear{Kohli, Garcia, and Pourahmadi}{Kohli
  et~al.}{2016}]{kohligarciapourahmadi2016}
Kohli, P., T.~P. Garcia, and M.~Pourahmadi (2016).
\newblock Modeling the cholesky factors of covariance matrices of multivariate
  longitudinal data.
\newblock {\em Journal of Multivariate Analysis\/}~{\em 145}, 87--100.


\bibitem[\protect\citeauthoryear{Lewis and Reinsel}{Lewis and
  Reinsel}{1985}]{lewisreinsel1985}
Lewis, R. and G.~C. Reinsel (1985).
\newblock Prediction of multivariate time series by autoregressive model
  fitting.
\newblock {\em Journal of Multivariate Analysis\/}~{\em 16}, 393--411.


\bibitem[\protect\citeauthoryear{Li, Phillips, and Gao}{Li
  et~al.}{2020}]{lihillipsgao2017}
Li, D., P.~C.~B. Phillips, and J.~Gao (2020).
\newblock Kernel-based inference in time-varying coefficient cointegrating
  regression.
\newblock {\em Journal of Econometrics\/}~{\em 215}, 607--632.


\bibitem[\protect\citeauthoryear{Mark, Ogaki, and Sul}{Mark
  et~al.}{2005}]{markogakisul2005}
Mark, N.~C., M.~Ogaki, and D.~Sul (2005).
\newblock Dynamic seemingly unrelated cointegrating regressions.
\newblock {\em The Review of Economic Studies\/}~{\em 72}, 797--820.


\bibitem[\protect\citeauthoryear{McMurry and Politis}{McMurry and
  Politis}{2010}]{mcmurypolitis2010}
McMurry, T.~L. and D.~N. Politis (2010).
\newblock Banded and tapered estimates for autocovariance matrices and the
  linear process bootstrap.
\newblock {\em Journal of Time Series Analysis\/}~{\em 31}, 471--482.


\bibitem[\protect\citeauthoryear{Moon and Perron}{Moon and
  Perron}{2005}]{moonperron2005}
Moon, H.~R. and B.~Perron (2005).
\newblock Efficient estimation of the seemingly unrelated regression
  cointegration model and testing for purchasing power parity.
\newblock {\em Econometric Reviews\/}~{\em 23}, 293--323.


\bibitem[\protect\citeauthoryear{Newey and West}{Newey and
  West}{1994}]{neweywest1994}
Newey, W.~K. and K.~D. West (1994).
\newblock Automatic lag selection in covariance matrix estimation.
\newblock {\em The Review of Economic Studies\/}~{\em 61}, 631--653.


\bibitem[\protect\citeauthoryear{Nordhaus}{Nordhaus}{2013}]{nordhaus2013}
Nordhaus, W.~D. (2013).
\newblock {\em The Climate Casino: Risk, Uncertainty, and Economics for a
  Warming World}.
\newblock Yale University Press.


\bibitem[\protect\citeauthoryear{Nyblom and Harvey}{Nyblom and
  Harvey}{2000}]{nyblomharvey2000}
Nyblom, J. and A.~Harvey (2000).
\newblock Tests of common stochastic trends.
\newblock {\em Econometric Theory\/}~{\em 16}, 176--199.


\bibitem[\protect\citeauthoryear{Park and Phillips}{Park and
  Phillips}{1999}]{parkphillips1999}
Park, J.~Y. and P.~C.~B. Phillips (1999).
\newblock Asymptotics for nonlinear transformations of integrated time series.
\newblock {\em Econometric Theory\/}~{\em 15\/}(3), 269--298.


\bibitem[\protect\citeauthoryear{Park and Phillips}{Park and
  Phillips}{2001}]{parkphillips2001}
Park, J.~Y. and P.~C.~B. Phillips (2001).
\newblock Nonlinear regressions with integrated time series.
\newblock {\em Econometrica\/}~{\em 69}, 117--161.


\bibitem[\protect\citeauthoryear{Phillips}{Phillips}{1991}]{phillips1991}
Phillips, P. C.~B. (1991).
\newblock Optimal inference in cointegrated systems.
\newblock {\em Econometrica\/}~{\em 59}, 283--306.


\bibitem[\protect\citeauthoryear{Phillips}{Phillips}{1995}]{phillips1995}
Phillips, P. C.~B. (1995).
\newblock Fully modified least squares and vector autoregression.
\newblock {\em Econometrica\/}~{\em 63}, 1023--1078.


\bibitem[\protect\citeauthoryear{Phillips and Hansen}{Phillips and
  Hansen}{1990}]{phillipshansen1990}
Phillips, P. C.~B. and B.~E. Hansen (1990).
\newblock Statistical inference in instrumental variables regression with
  {I}(1) processes.
\newblock {\em The Review of Economic Studies\/}~{\em 57}, 99--125.


\bibitem[\protect\citeauthoryear{Phillips and Park}{Phillips and
  Park}{1988}]{phillipspark1988}
Phillips, P. C.~B. and J.~Y. Park (1988).
\newblock Asymptotic equivalence of ordinary least squares and generalized
  least squares in regressions with integrated regressors.
\newblock {\em Journal of the American Statistical Association\/}~{\em 83},
  111--115.


\bibitem[\protect\citeauthoryear{Phillips and Solo}{Phillips and
  Solo}{1992}]{phillipssolo1992}
Phillips, P. C.~B. and V.~Solo (1992).
\newblock Asymptotics for linear processes.
\newblock {\em The Annals of Statistics\/}~{\em 20}, 971--1001.


\bibitem[\protect\citeauthoryear{Pourahmadi}{Pourahmadi}{1999}]{pourahmadi1999}
Pourahmadi, M. (1999).
\newblock Joint mean-covariance models with applications to longitudinal data:
  Unconstrained parameterisation.
\newblock {\em Biometrika\/}~{\em 86}, 677--690.


\bibitem[\protect\citeauthoryear{Prais and Winsten}{Prais and
  Winsten}{1954}]{praiswinston1954}
Prais, S.~J. and C.~B. Winsten (1954).
\newblock Trend estimators and serial correlation.
\newblock Cowles Foundation, Discussion Paper 383.

\bibitem[\protect\citeauthoryear{Romano and Wolf}{Romano and
  Wolf}{2001}]{romanowolf2001}
Romano, J.~P. and M.~Wolf (2001).
\newblock Subsampling intervals in autoregressive models with linear time
  trend.
\newblock {\em Econometrica\/}~{\em 69}, 1283--1314.


\bibitem[\protect\citeauthoryear{Saikkonen}{Saikkonen}{1992}]{saikkonen1992}
Saikkonen, P. (1992).
\newblock Estimation and testing of cointegrated systems by an autoregressive
  approximation.
\newblock {\em Econometric Theory\/}~{\em 8}, 1--27.


\bibitem[\protect\citeauthoryear{Shin}{Shin}{1994}]{shin1994}
Shin, Y. (1994).
\newblock A residual-based test of the null of cointegration against the
  alternative of no cointegration.
\newblock {\em Econometric Theory\/}~{\em 10}, 91--115.


\bibitem[\protect\citeauthoryear{Sims, Stock, and Watson}{Sims
  et~al.}{1990}]{simsstockwatson1990}
Sims, C.~A., J.~H. Stock, and M.~W. Watson (1990).
\newblock Inference in linear time series mmodels with some unit roots.
\newblock {\em Econometrica\/}~{\em 58}, 113--144.


\bibitem[\protect\citeauthoryear{Stern}{Stern}{2004}]{stern2004}
Stern, D.~I. (2004).
\newblock The rise and fall of the environmental {K}uznets curve.
\newblock {\em World Development\/}~{\em 32}, 1419--1439.


\bibitem[\protect\citeauthoryear{Stern}{Stern}{2017}]{stern2017}
Stern, D.~I. (2017).
\newblock The environmental {K}uznets curve after 25 years.
\newblock {\em Journal of Bioeconomics\/}~{\em 19}, 7--28.


\bibitem[\protect\citeauthoryear{Tao}{Tao}{2012}]{tao2012}
Tao, T. (2012).
\newblock {\em Topics in Random Matrix Theory}.
\newblock American Mathematical Society.


\bibitem[\protect\citeauthoryear{Tj\o{}stheim}{Tj\o{}stheim}{2020}]{tjostheim2020}
Tj\o{}stheim, D. (2020).
\newblock Some notes on nonlinear cointegration: A partial review with some
  novel perspectives.
\newblock {\em Econometric Reviews\/}~{\em 39}, 655--673.


\bibitem[\protect\citeauthoryear{Vogelsang and Wagner}{Vogelsang and
  Wagner}{2014}]{vogelsangwagner2014}
Vogelsang, T.~J. and M.~Wagner (2014).
\newblock Integrated modified {OLS} estimation and fixed-b inference for
  cointegrating regressions.
\newblock {\em Journal of Econometrics\/}~{\em 178}, 741--760.


\bibitem[\protect\citeauthoryear{Wagner, Grabarczyk, and Hong}{Wagner
  et~al.}{2020}]{wagnergrabarczykhong2019}
Wagner, M., P.~Grabarczyk, and S.~H. Hong (2020).
\newblock Fully modified {OLS} estimation and inference for seemingly unrelated
  cointegrating polynomial regressions and the environmental {K}uznets curve
  for carbon dioxide emissions.
\newblock {\em Journal of Econometrics\/}~{\em 214}, 216--255.


\bibitem[\protect\citeauthoryear{Wagner and Hong}{Wagner and
  Hong}{2016}]{wagnerhong2016}
Wagner, M. and S.~H. Hong (2016).
\newblock Cointegrating polynomial regressions: Fully modified {OLS} estimation
  and inference.
\newblock {\em Econometric Theory\/}~{\em 32}, 1289--1315.


\bibitem[\protect\citeauthoryear{Wang}{Wang}{2015}]{wang2015}
Wang, Q. (2015).
\newblock {\em Limit Theorems for Nonlinear Cointegrating Regression}.
\newblock World Scientific.


\bibitem[\protect\citeauthoryear{Wang and Phillips}{Wang and
  Phillips}{2009}]{wangphillips2009}
Wang, Q. and P.~C.~B. Phillips (2009).
\newblock Structural nonparametric cointegrating regression.
\newblock {\em Econometrica\/}~{\em 77}, 1901--1948.


\bibitem[\protect\citeauthoryear{Wei}{Wei}{1987}]{wei1987}
Wei, C.-Z. (1987).
\newblock Adaptive prediction by least squares predictors in stochastic
  regression models with applications to time series.
\newblock {\em The Annals of Statistics\/}~{\em 15}, 1667--1682.


\bibitem[\protect\citeauthoryear{Wu and Pourahmadi}{Wu and
  Pourahmadi}{2009}]{wupourahmadi2009}
Wu, W.~B. and M.~Pourahmadi (2009).
\newblock Banding sample autocovariance matrices of stationary processes.
\newblock {\em Statistica Sinica\/}~{\em 19}, 1755--1768.


\bibitem[\protect\citeauthoryear{Zellner}{Zellner}{1962}]{zellner1962}
Zellner, A. (1962).
\newblock An efficient method of estimating seemingly unrelated regressions and
  tests for aggregation bias.
\newblock {\em Journal of the American Statistical Association\/}~{\em 57},
  348--368.


\end{thebibliography}
\end{document}